\documentclass[a4paper]{llncs}

\usepackage[english]{babel}
\usepackage{german}
\usepackage{amssymb}
\usepackage{amsfonts}
\usepackage{amsmath}
\usepackage[latin1]{inputenc}
\usepackage{fullpage}
\usepackage{graphicx}
\usepackage{subfigure}
\usepackage{algorithm}
\usepackage[noend]{algorithmic}
\usepackage{amsmath,amssymb,cite}
\usepackage{lineno,footmisc,marvosym}
\usepackage{wasysym,vmargin,stackrel}

\newenvironment{proofofclaim}[1][Proof]{\noindent\textbf{#1.} }{\ \rule{0.5em}{0.5em}}
\newtheorem{observation}{Observation}
\newtheorem{myclaim}{Claim}
\setmarginsrb{3.5cm}{2.25cm}{3.5cm}{3.05cm}{0.3cm}{0.3cm}{-0.3cm}{1.0cm}
   \renewenvironment{thebibliography}[1]{
 \begin{oldthebibliography}{#1}
 \setlength{\parskip}{0.0ex} \setlength{\itemsep}{1ex}}  {\end{oldthebibliography}}

\begin{document}

\title{The Recognition of Simple-Triangle Graphs \\
and of Linear-Interval Orders is Polynomial\thanks{This work was partially supported by~the EPSRC Grant~EP/K022660/1.} 
\thanks{A preliminary conference version of this work appeared in the \emph{Proceedings of the 21st European Symposium on Algorithms (ESA)}, 
Sophia Antipolis, France, pages 719--730, 2013.}}
\author{George B. Mertzios}
\institute{School of Engineering and Computing Sciences, Durham University, UK.\\
Email: \texttt{george.mertzios@durham.ac.uk}\vspace{-0.2cm}}
\maketitle

\begin{abstract}
Intersection graphs of geometric objects have been extensively studied, both due to their interesting structure and their numerous applications; 
prominent examples include interval graphs 
and permutation graphs. 
In this paper we study a natural graph class that generalizes both interval and permutation graphs, namely \emph{simple-triangle} graphs. 
Simple-triangle graphs --~also known as \emph{PI} graphs (for Point-Interval)~-- are the intersection graphs of triangles that are defined
by a point on a line $L_{1}$ and an interval on a parallel line $L_{2}$. 
They lie naturally between permutation and trapezoid graphs, which are the intersection graphs of line segments
between $L_{1}$ and $L_{2}$ and of trapezoids between $L_{1}$ and $L_{2}$, respectively. 
Although various efficient recognition algorithms for permutation and trapezoid graphs are well known to exist, 
the recognition of simple-triangle graphs has remained an open problem since their introduction by~Corneil and~Kamula three decades ago. 
In this paper we resolve this problem by proving that simple-triangle graphs can be recognized in polynomial time. 
As a consequence, our algorithm also solves a longstanding open problem in the area of partial orders, 
namely the recognition of \emph{linear-interval~orders}, 
i.e.~of partial orders $P=P_{1}\cap P_{2}$, where $P_{1}$ is a linear order and $P_{2}$ is an interval order. 
This is one of the first results on recognizing partial orders~$P$ that are the 
intersection~of~orders from two different classes $\mathcal{P}_{1}$ and $\mathcal{P}_{2}$. 
In complete contrast to this, partial orders~$P$ which are~the~intersection of orders from the same class $\mathcal{P}$ have been extensively investigated, 
and in most cases the complexity status of these recognition problems has been already established.\newline

\end{abstract}

\section{Introduction\label{intro-sec}}

A graph $G$ is the \emph{intersection graph} of a family $\mathcal{F}$ of
sets if we can bijectively assign sets of~$\mathcal{F}$ to vertices of $G$
such that two vertices of $G$ are adjacent if and only if the corresponding
sets have a non-empty intersection. 
It turns out that many graph classes with important applications can be
described as intersection graphs of set families that are derived from some
kind of geometric configuration. One of the most prominent examples is that
of \emph{interval} graphs, i.e.~the intersection graphs of intervals on the
real line, which have natural applications in several fields,
including~bioinformatics and involving the physical mapping of DNA and the
genome~reconstruction\footnote{%
Benzer~\cite{Benzer59} earned the prestigious Lasker Award (1971) and
Crafoord Prize (1993) partly for showing that the set of intersections of a
large number of fragments of genetic material in a virus form an interval
graph.\vspace{-0.2cm}}~\cite{Goldberg95,Golumbic04,Carrano88}.

Generalizing the intersections on the real line, consider two parallel
horizontal lines on the plane,~$L_{1}$ (the upper line) and~$L_{2}$ (the
lower~line). A graph~$G$ is a \emph{simple-triangle} graph if it is the
intersection graph of triangles that have one endpoint on $L_{1}$ and the
other two on $L_{2}$. Furthermore, $G$ is a \emph{triangle} graph if it is
the intersection graph of triangles with endpoints on $L_{1}$ and $L_{2}$,
but now there is no restriction on which line contains one endpoint of every
triangle and which contains the other two. Simple-triangle and triangle
graphs are also known as~\emph{PI} and~\emph{PI$^{\ast }$} graphs,
respectively~\cite{Brandstaedt99,CorKam87,Spinrad03}, where~PI stands 
for ``Point-Interval'' . Such representations of simple-triangle and of
triangle graphs are called \emph{simple-triangle}~(or~\emph{PI})~and~\emph{%
triangle} (or~\emph{PI}$^{\ast }$) \emph{representations}, respectively.
Simple-triangle and triangle graphs lie naturally between \emph{permutation}
graphs (i.e.~the intersection graphs of line segments with one endpoint on~$%
L_{1}$ and one on~$L_{2}$) and \emph{trapezoid} graphs (i.e.~the
intersection graphs of trapezoids with one interval on~$L_{1}$ and the
opposite interval on~$L_{2}$)~\cite{Brandstaedt99,Spinrad03}. Note that,
using the notation \emph{PI }for simple-triangle graphs, permutation graphs
are~\emph{PP} (for ``Point-Point'') graphs, while trapezoid graphs are~\emph{%
II} (for ``Interval-Interval'') graphs~\cite{CorKam87}.

A \emph{partial order} is a pair $P=(U,R)$, where $U$ is a finite set and $R$
is an irreflexive transitive binary relation on $U$. Whenever $(x,y)\in R$
for two elements $x,y\in U$, we write $x<_{P}y$. If $x<_{P}y$ or $y<_{P}x$,
then $x$ and $y$ are \emph{comparable}, otherwise they are \emph{incomparable%
}. $P$ is a \emph{linear order} if every pair of elements in $U$ are
comparable. Furthermore, $P$ is an \emph{interval order} if each element $%
x\in U$ is assigned to an interval $I_{x}$ on the real line such that $%
x<_{P}y$ if and only if $I_{x}$ lies completely to the left~of~$I_{y}$. One
of the most fundamental notions on partial orders is \emph{dimension}. For
any partial order $P$ and any class $\mathcal{P}$ of partial orders
(e.g.~linear order, interval order, semiorder, etc.), the $\mathcal{P}$\emph{%
-dimension}~of~$P$ is the smallest $k$ such that $P$ is the intersection of $%
k$ orders from $\mathcal{P}$. In particular, when $\mathcal{P}$ is the class
of linear orders, the $\mathcal{P}$-dimension of $P$ is known as the \emph{%
dimension} of $P$. Although in most cases we can efficiently recognize
whether a partial order belongs to a class $\mathcal{P}$, this is not
the case for higher dimensions. Due to a classical result of
Yannakakis~\cite{Yannakakis82}, it is NP-complete to decide whether the
dimension, or the interval dimension, of a partial order is at most~$k$, 
where~$k\geq 3$.

There is a natural correspondence between graphs and partial orders. For a
partial order ${P=(U,R)}$, the \emph{comparability} (resp.~\emph{%
incomparability}) \emph{graph} $G(P)$ of~$P$ has elements of $U$ as vertices
and an edge between every pair of comparable (resp.~incomparable) elements.
A graph $G$ is a \emph{(co)comparability graph} if $G$ is the
(in)comparability graph of a partial order $P$. 
There has been a long line of research in order to establish the complexity
of recognizing partial orders of $\mathcal{P}$-dimension at most~$2$
(e.g.~where $\mathcal{P}$ is linear orders~\cite{Spinrad03} or interval
orders~\cite{MaSpinrad94}). 
In particular, since permutation (resp.~trapezoid) graphs
are the incomparability graphs of partial orders with dimension
(resp.~interval dimension) at~most~$2$~\cite{Dagan88,Spinrad03}, permutation
and trapezoid graphs can be recognized efficiently by the corresponding
partial order algorithms~\cite{MaSpinrad94,Spinrad03}.

In contrast, not much is known so far for the recognition of partial orders $%
P$ that are the intersection of orders from different classes $\mathcal{P}%
_{1}$ and $\mathcal{P}_{2}$. 
One of the longstanding open problems in this area is
the recognition of \emph{linear-interval orders} $P$, i.e.~of partial orders 
$P=P_{1}\cap P_{2}$, where $P_{1}$ is a linear order and $P_{2}$ is an
interval order. In~terms of graphs, this problem is equivalent to the
recognition of simple-triangle (i.e.~PI) graphs, since~PI~graphs are the
incomparability graphs of linear-interval orders; this problem is well known and remains open 
since the introduction of PI graphs in 1987~\cite{CorKam87} (cf.~for instance the books~\cite{Spinrad03,Brandstaedt99}).

\vspace{0.1cm}

\noindent \textbf{Our contribution.}~In this article we establish the
complexity of recognizing simple-triangle (PI) graphs, and therefore also
the complexity of recognizing linear-interval orders. Given a graph $%
G$ with $n$ vertices, such that its complement~$\overline{G}$ has $m$ edges,
we provide an algorithm with running time $O(n^{2}m)$ that either computes a
PI representation of $G$, or it announces that $G$ is not a PI graph.
Equivalently, given a partial order $P=(U,R)$ with $|U|=n$ and $|R|=m$, our
algorithm either computes in $O(n^{2}m)$ time a linear order $P_{1}$ and an
interval order $P_{2}$ such that $P=P_{1}\cap P_{2}$, or it announces that
such orders $P_{1},P_{2}$ do not exist. Surprisingly, it turns out that the
seemingly small difference in the definition of simple-triangle~(PI) graphs
and triangle (PI$^{\ast }$) graphs results in a very different behavior of
their recognition problems; only recently it has been proved that the
recognition of triangle graphs is NP-complete~\cite{Mertzios-PI-ast-tcs}. In
addition, our polynomial time algorithm is in contrast to the recognition
problems for the related classes of \emph{bounded tolerance} (i.e.~\emph{%
parallelogram}) graphs~\cite{MSZ-SICOMP-11} and of \emph{max-tolerance}
graphs~\cite{Kaufmann06}, which have already been proved to be NP-complete.

As the main tool for our algorithm we introduce the notion of a \emph{linear-interval cover} of bipartite graphs. 
As a second tool we identify a new tractable~subclass~of~$3$SAT, called \emph{gradually mixed} formulas, 
for which we provide a linear time algorithm. The class of gradually mixed formulas is \emph{hybrid}, 
i.e.~it is characterized by both \emph{relational} and \emph{structural} restrictions on the clauses. 
Then, using the notion of a linear-interval cover, we are able to reduce our problem to the satisfiability problem of gradually mixed formulas.

Our algorithm proceeds as follows. First, it computes from the given graph $%
G $ a bipartite graph $\widetilde{G}$, such that $G$ is a PI graph if and
only if $\widetilde{G}$ has a linear-interval cover. Second, it computes a 
gradually mixed Boolean formula $\phi $ such that $\phi $ is
satisfiable if and only if $\widetilde{G}$ has a linear-interval cover. This
formula $\phi $ can be written as $\phi =\phi _{1}\wedge \phi _{2}$, where
every clause of $\phi _{1}$ has $3$ literals and every clause of $\phi _{2}$
has $2$ literals. The construction of $\phi _{1}$ and $\phi _{2}$ is based
on the fact that a necessary condition for $\widetilde{G}$ to admit a
linear-interval cover is that its edges can be colored with two different
colors (according to some restrictions). Then the edges of $\widetilde{G}$
correspond to literals of $\phi $, while the two edge colors encode the
truth value of the corresponding variables. Furthermore every clause of $%
\phi _{1}$ corresponds to the edges of an \emph{alternating cycle} in $%
\widetilde{G}$ (i.e.~a closed walk that alternately visits edges and
non-edges) of length $6$, while the clauses of $\phi _{2}$ correspond to
specific pairs of edges of $\widetilde{G}$ that are not allowed to receive
the same color. Finally, the equivalence between the existence of a
linear-interval cover of $\widetilde{G}$ and a satisfying truth assignment
for $\phi $ allows us to use our linear algorithm to solve satisfiability on gradually
mixed formulas in order to complete our recognition algorithm.

\vspace{0.1cm}

\noindent \textbf{Organization of the paper.}~We present in Section~\ref%
{tractable-subclass-3SAT-sec} the class of gradually mixed formulas and a
linear time algorithm to solve satisfiability on this class. In Section~\ref{prelim-sec}
we provide the necessary notation and preliminaries on threshold graphs and
alternating cycles. Then in Section~\ref{linear-interval-sec} we introduce
the notion of a linear-interval cover~of bipartite graphs to characterize PI
graphs, and in Section~\ref{linear-interval-satisfiability-sec} we translate
the linear-interval~cover~problem~to~the~satisfiability problem on a gradually mixed
formula. Finally, in Section~\ref{recognition-sec} we~present~our~PI~graph~recognition~algorithm.

\section{A tractable subclass of $3$SAT\label{tractable-subclass-3SAT-sec}}

In this section we introduce the class of \emph{gradually mixed} formulas
and we provide a linear time algorithm for solving satisfiability on this
class. Any gradually mixed formula $\phi$ is a mix of binary and ternary
clauses. That is, there exist a $3$-CNF formula $\phi _{1}$ (i.e.~a formula
in conjunctive normal form with at most $3$ literals per clause) and a $2$%
-CNF formula $\phi _{2}$ (i.e.~with at most $2$ literals per clause) such
that $\phi =\phi _{1}\wedge \phi _{2}$, while $\phi $ satisfies some
constraints among its clauses. Before we define gradually mixed formulas
(cf.~Definition~\ref{gradually-mixed-def}), we first define \emph{dual}
clauses. 

\begin{definition}
\label{dual-clause-def}Let $\phi _{1}$ be a $3$-CNF formula. If $\alpha
=(\ell _{1}\vee \ell _{2}\vee \ell _{3})$ is a clause of~$\phi _{1}$, then
the $\overline{\alpha }=(\overline{\ell _{1}}\vee \overline{\ell _{2}}\vee 
\overline{\ell _{3}})$ is the \emph{dual} clause of $\alpha $.
\end{definition}

Note by Definition~\ref{dual-clause-def} that, whenever $\alpha $ is a
clause of a formula $\phi _{1}$, the dual clause $\overline{\alpha }$ of $%
\alpha $ may belong, or may not belong, to $\phi _{1}$.

\begin{definition}
\label{gradually-mixed-def}Let $\phi _{1}$ and $\phi _{2}$ be CNF formulas
with $3$ literals and $2$ literals in each clause, respectively. The mixed
formula $\phi =\phi _{1}\wedge \phi _{2}$ is \emph{gradually mixed} if the
next two conditions are satisfied:\vspace{-0,1cm}

\begin{enumerate}
\item Let $\alpha$ and $\beta$ be two clauses of $\phi_{1}$. Then $\alpha$
does not share exactly one literal with either the clause $\beta$ or the
clause $\overline{\beta}$.\vspace{0,15cm}

\item If $\alpha =(\ell _{1}\vee \ell _{2}\vee \ell _{3})$ is a clause of $%
\phi _{1}$ and $(\ell _{0}\vee \overline{\ell _{1}})$ is a clause of $\phi
_{2}$, then $\phi _{2}$ contains also (at least) one of the clauses $\{(\ell
_{0}\vee \ell _{2}),(\ell _{0}\vee \ell _{3})\}$.
\end{enumerate}
\end{definition}

As an example of a gradually mixed formula, consider the formula~$\phi
=\phi_{1}\wedge \phi_{2}$, where ${\phi_{1} = (x_{1} \vee \overline{x_{2}}
\vee x_{3}) \wedge (\overline{x_{1}} \vee x_{2} \vee x_{4}) \wedge (x_{5}
\vee x_{6} \vee \overline{x_{7}})}$ and~$\phi_{2} = (x_{8} \vee \overline{%
x_{3}}) \wedge (x_{8} \vee x_{1}) \wedge (x_{8} \vee x_{4}) \wedge (%
\overline{x_{8}} \vee x_{9}) \wedge (x_{5} \vee x_{10}) \wedge (\overline{%
x_{6}} \vee x_{10})$.

Note by Definition~\ref{gradually-mixed-def} that the class of gradually
mixed formulas contains $2$SAT as a proper subclass, since every $2$-CNF
formula $\phi_{2}$ can be written as a gradually mixed formula $\phi =\phi
_{1}\wedge \phi _{2}$ where $\phi_{1}=\emptyset$. Furthermore the class of
gradually mixed formulas $\phi$ is a \emph{hybrid} class, since the
conditions of Definition~\ref{gradually-mixed-def} concern simultaneously 
\emph{relational} restrictions (i.e.~where the clauses are restricted to be
of certain types) and \emph{structural} restrictions (i.e.~where there are
restrictions on how different clauses interact with each other). The
intuition for the term \emph{gradually mixed} in Definition~\ref%
{gradually-mixed-def} is that, whenever the sub-formulas $\phi _{1}$ and $%
\phi _{2}$ share more variables, the number of clauses of $\phi _{2}$ that
are imposed by condition~2 of Definition~\ref{gradually-mixed-def}
increases. In the next theorem we use resolution to prove that
satisfiability can be solved in linear time on gradually mixed formulas.

\begin{theorem}
\label{gradually-mixed-algorithm-thm}There exists a linear time algorithm
which decides whether a given gradually mixed formula $\phi $ is satisfiable
and computes a satisfying truth assignment of $\phi $, if one exists.
\end{theorem}

\begin{proof}
Let $\phi =\phi _{1}\wedge \phi _{2}$, where $\phi _{1}$ is a $3$-CNF
formula and $\phi _{2}$ is a $2$-CNF formula. We first scan through all
clauses of $\phi $ to remove all tautologies, i.e.~all clauses which contain
both a literal and its negation, since such clauses are always satisfiable.
Furthermore we eliminate all double literal occurrences in every clause. In
the remainder of the proof we denote by $\phi $ the resulting formula after
the removal of tautologies and the elimination of double literal occurrences
in the clauses. Note that, during this elimination procedure, some clauses
of $\phi _{1}$ may become $2$-CNF clauses. In the resulting formula 
we denote by $\phi _{1}^{\prime }$ the conjunction of the clauses that have $%
3$ literals each, and by $\phi _{1}^{\prime \prime }$ the conjunction of the
clauses of $\phi _{1}$ that remain with $1$ or $2$ literals each. In
particular, since also in every clause of $\phi _{1}$ no literal is the
negation of another one (as we removed from $\phi$ all tautologies), the
literals of every clause in $\phi_{1}^{\prime}$ correspond to three distinct
variables.

Then we compute a $2$-CNF formula $\phi _{0}$ (in time linear to the size of 
$\phi $) as follows. Initially $\phi _{0}$ is empty. First we mark all
literals $\ell $ for which the $2$-CNF formula $\phi _{1}^{\prime \prime
}\wedge \phi _{2}$ includes the clause $(\ell )$. Then we scan through all
clauses of the $3$-CNF formula $\phi _{1}^{\prime }$. For every clause $%
(\ell _{1}\vee \ell _{2}\vee \ell _{3})$ of $\phi _{1}^{\prime }$, such that
the literal $\overline{\ell _{1}}$ (resp.~$\overline{\ell _{2}}$~or~$%
\overline{\ell _{3}}$) has been marked, we add to $\phi _{0}$ the clause $%
(\ell _{2}\vee \ell _{3})$ (resp.~the clause $(\ell _{1}\vee \ell _{3})$ or $%
(\ell _{1}\vee \ell _{2})$).

If $\phi \wedge \phi _{0}$ is satisfiable then clearly $\phi $ is also
satisfiable as a sub-formula of $\phi \wedge \phi _{0}$. Conversely, suppose
that $\phi $ is satisfied by the truth assignment $\tau$. Let $\gamma =(\ell
_{1}\vee \ell _{2})$ be an arbitrary clause of $\phi _{0}$. The existence of 
$\gamma $ in $\phi _{0}$ implies the existence of some clauses $\alpha =(%
\overline{\ell _{3}})$ and $\beta =(\ell _{1}\vee \ell _{2}\vee \ell _{3})$
in $\phi $. Therefore, since $\alpha =\beta =1$ in $\tau $ by assumption, it
follows that $\ell _{3}=0$ in $\tau $. Thus the clause $\beta $ equals $%
(\ell _{1}\vee \ell _{2})$ in $\tau $, and therefore $\gamma =1$ in $\tau $.
That is, $\tau $ satisfies also $\phi _{0}$. Therefore $\phi $ is
satisfiable if and only if $\phi \wedge \phi _{0} $ is satisfiable.

In the remainder of the proof, we prove that $\phi \wedge \phi _{0}$ is
satisfiable if and only if the $2$-CNF formula $\phi _{1}^{\prime \prime
}\wedge \phi _{2}\wedge \phi _{0}$ is satisfiable. The one direction is
immediate, i.e.~if $\phi \wedge \phi _{0}$ is satisfiable then $\phi
_{1}^{\prime \prime }\wedge \phi _{2}\wedge \phi _{0}$ is also satisfiable
as a sub-formula of $\phi \wedge \phi _{0}$. Conversely, suppose that $\phi
_{1}^{\prime \prime }\wedge \phi _{2}\wedge \phi _{0}$ is satisfiable and
let $\tau $ be a satisfying truth assignment of this formula. If $\tau $
satisfies all clauses of $\phi_{1}^{\prime}$, then clearly $\tau $ is also a
satisfying truth assignment of $\phi \wedge \phi _{0}$. Otherwise let $%
\alpha =(\ell _{1}\vee \ell _{2}\vee \ell _{3})$ be a clause of $%
\phi_{1}^{\prime}$ that is not satisfied by $\tau $. Then $\ell _{1}=\ell
_{2}=\ell _{3}=0$ in $\tau $. In this case, we construct the truth
assignment $\tau ^{\prime }$ from $\tau $ by flipping the value of one
(arbitrary) literal of $\{\ell _{1},\ell _{2},\ell _{3}\}$ in $\tau $.
Assume without loss of generality that the value of $\ell _{1}$ flips from $%
\tau $ to $\tau ^{\prime }$, while the values of all other variables remain
the same in both $\tau $ and $\tau ^{\prime }$. Recall that the literals $%
\{\ell _{1},\ell _{2},\ell _{3}\}$ correspond to three distinct variables,
since we eliminated all double occurrences of literals in all clauses in $%
\phi _{1}$. Therefore $\ell _{1}=\overline{\ell _{2}}=\overline{\ell _{3}}=1$
in $\tau ^{\prime }$, and thus~$\alpha =1$~in~$\tau ^{\prime }$.

Suppose that there exists a clause $\beta =(\ell _{4}\vee \ell _{5}\vee \ell
_{6})$ of $\phi _{1}^{\prime}$ where $\beta =1$ in $\tau $ and $\beta =0$ in 
$\tau ^{\prime }$. Then clearly one of the literals of $\beta$ equals $%
\overline{\ell_{1}}$, since $\overline{\ell_{1}}$ is the only literal whose
value changes in $\tau ^{\prime }$ from $1$ to~$0$. Assume without loss of
generality that $\ell _{4}=\overline{\ell_{1}}$, i.e.~$\alpha$ shares at
least one literal with $\overline{\beta} = (\overline{\ell_{4}} \vee 
\overline{\ell_{5}} \vee \overline{\ell_{6}})$. Therefore, since $\phi $ is
a gradually mixed formula by assumption, it follows by Definition~\ref%
{gradually-mixed-def} that $\alpha$ shares at least one more literal with $%
\overline{\beta}$. Assume without loss of generality that $\ell _{5}=%
\overline{\ell_{2}}$. Then, since by assumption $\ell_{2}=0$ in both $\tau$
and $\tau^{\prime}$, it follows that the clause $\beta =(\ell _{4}\vee
\ell_{5}\vee \ell _{6}) = (\overline{\ell _{1}}\vee \overline{\ell_{2}}\vee
\ell _{6})$ is satisfied in $\tau ^{\prime }$, which is a contradiction to
our assumption. Therefore for every clause~$\beta$ of $\phi_{1}^{\prime}$,
if~${\beta=1}$ in~$\tau$ then also~${\beta=1}$ in~${\tau^{\prime}}$.

We now prove that all clauses of the $2$-CNF formula $\phi _{1}^{\prime
\prime }\wedge \phi _{2}\wedge \phi _{0}$ remain satisfied in $\tau ^{\prime
}$. First consider an arbitrary clause $\gamma $ of $\phi _{0}$ that
contains one of the literals $\{\ell _{1},\overline{\ell _{1}}\}$. If $%
\gamma $ contains the literal $\ell _{1}$ then $\gamma =1$ in $\tau ^{\prime
}$, since $\ell _{1}=1$ in $\tau ^{\prime }$. Let $\gamma $ contain the
literal $\overline{\ell _{1}}$, and let $\gamma =(\overline{\ell _{1}}\vee
\ell _{4})$. Then it follows by the construction of the formula $\phi _{0}$
that there exists a literal $\ell _{5}$, such that $(\overline{\ell _{1}}%
\vee \ell _{4}\vee \ell _{5})$ is a clause of $\phi _{1}^{\prime }$ and $(%
\overline{\ell _{5}})$ is a clause of $\phi _{1}^{\prime \prime }\wedge \phi
_{2}$. Note that $(\overline{\ell _{1}}\vee \ell _{4}\vee \ell _{5})=1$ in $%
\tau $, since $\ell _{1}=0$ in $\tau $ by assumption. Therefore also $(%
\overline{\ell _{1}}\vee \ell _{4}\vee \ell _{5})=1$ in $\tau ^{\prime }$ by
the previous paragraph. Thus, since $\overline{\ell _{1}}=0$ in $\tau
^{\prime }$, it follows that $(\ell _{4}\vee \ell _{5})=1$ in $\tau ^{\prime
}$. Furthermore, since $\tau $ satisfies $\phi _{1}^{\prime \prime }\wedge
\phi _{2}$ by assumption, it follows that $(\overline{\ell _{5}})=1$ in $%
\tau $, and thus $\ell _{5}=0$ in both $\tau $ and $\tau ^{\prime }$.
Therefore $\ell _{4}=1$ in $\tau ^{\prime }$, since $(\ell _{4}\vee \ell
_{5})=1$ in $\tau ^{\prime }$, and thus $\gamma =(\overline{\ell _{1}}\vee
\ell _{4})=1$ in $\tau ^{\prime }$. That is, all clauses $\gamma $ of $\phi
_{0}$ remain satisfied in the assignment $\tau ^{\prime }$.

Now consider a clause $\gamma $ of $\phi _{2}$ that contains one of the
literals $\{\ell _{1},\overline{\ell _{1}}\}$. If $\gamma $ contains $\ell
_{1}$ then $\gamma =1$ in $\tau ^{\prime }$, since $\ell _{1}=1$ in $\tau
^{\prime }$. Let $\gamma $ contain the literal $\overline{\ell _{1}}$, and
let $\gamma =(\overline{\ell _{1}}\vee \ell _{4})$. Note that $\ell _{4}\neq
\ell _{1}$, since we removed all tautologies from $\phi $. Suppose that $%
\ell _{4}=\overline{\ell _{1}}$, i.e.~$\gamma =(\overline{\ell _{1}})$.
Then, since $\alpha =(\ell _{1}\vee \ell _{2}\vee \ell _{3})$ is a clause of 
$\phi _{1}$ by assumption, the formula $\phi _{0}$ contains (by
construction) the clause $(\ell _{2}\vee \ell _{3})$. Thus, since $\tau $
satisfies $\phi _{0}$ by assumption, it follows that $\ell _{2}=1$ or $\ell
_{3}=1$ in $\tau $. This is a contradiction, since $\ell _{1}=\ell _{2}=\ell
_{3}=0$ in $\tau $. Therefore $\ell _{4}\notin \{\ell _{1},\overline{\ell
_{1}}\}$. Thus, since $\phi $ is a gradually mixed formula by assumption, it
follows by Definition~\ref{gradually-mixed-def} that $\phi _{2}$ has also
one of the clauses $\{(\ell _{4}\vee \ell _{2}),(\ell _{4}\vee \ell _{3})\}$%
. Assume without loss of generality that $\phi _{2}$ has the clause $(\ell
_{4}\vee \ell _{2})$. Then, since $\tau $ satisfies $\phi _{2}$ by
assumption and $\ell _{2}=0$ in $\tau $, it follows that $\ell _{4}=1$ in $%
\tau $. Furthermore, since $\ell _{4}\notin \{\ell _{1},\overline{\ell _{1}}%
\}$, it remains $\ell _{4}=1$ in $\tau ^{\prime }$, and thus $\gamma =(%
\overline{\ell _{1}}\vee \ell _{4})=1$ in $\tau ^{\prime }$. That is, all
clauses $\gamma $ of $\phi _{2}$ remain satisfied in the assignment $\tau
^{\prime }$.

Finally consider a clause $\gamma $ of $\phi _{1}^{\prime \prime }$ that
contains one of the literals $\{\ell _{1},\overline{\ell _{1}}\}$. If $%
\gamma $ contains $\ell _{1}$ then $\gamma =1$ in $\tau ^{\prime }$, since $%
\ell _{1}=1$ in $\tau ^{\prime }$. Let $\gamma $ contain the literal $%
\overline{\ell _{1}}$, and let $\gamma =(\overline{\ell _{1}}\vee \ell _{4})$%
. Note that $\ell _{4}\neq \ell _{1}$, since we removed all tautologies from 
$\phi $. Suppose that $\ell _{4}=\overline{\ell _{1}}$, i.e.~$\gamma =(%
\overline{\ell _{1}})$. Then, since $\alpha =(\ell _{1}\vee \ell _{2}\vee
\ell _{3})$ is a clause of $\phi _{1}$ by assumption, the formula $\phi _{0}$
contains by construction the clause $(\ell _{2}\vee \ell _{3})$. Thus $\ell
_{2}=1$ or $\ell _{3}=1$ in $\tau $, since $\tau $ satisfies $\phi _{0}$ by
assumption. This is a contradiction, since $\ell _{1}=\ell _{2}=\ell _{3}=0$
in $\tau $. Therefore $\ell _{4}\notin \{\ell _{1},\overline{\ell _{1}}\}$.
Recall that $\phi _{1}^{\prime \prime }$ contains exactly those clauses of $%
\phi _{1}$ which remain with $1$ or $2$ literals each, after eliminating all
double literal occurrences in every clause of $\phi $. That is, the clause $%
\gamma $ was before the double literal elimination one of the clauses $(%
\overline{\ell _{1}}\vee \ell _{4}\vee \ell _{4})$ and $(\overline{\ell _{1}}%
\vee \overline{\ell _{1}}\vee \ell _{4})$. Furthermore $\alpha =(\ell
_{1}\vee \ell _{2}\vee \ell _{3})$ and $\gamma $ are two different clauses
of $\phi _{1}$, since $\alpha $ belongs to $\phi _{1}^{\prime }$ and $\gamma 
$ belongs to $\phi _{1}^{\prime \prime }$. Moreover $\alpha $ shares the
literal $\ell _{1}$ with the dual clause $\overline{\gamma }$ of $\gamma $.
If $\gamma $ was the clause $(\overline{\ell _{1}}\vee \ell _{4}\vee \ell
_{4})$ before the double literal elimination, then Definition~\ref%
{gradually-mixed-def} implies that $\ell _{4}=\overline{\ell _{2}}$ or $\ell
_{4}=\overline{\ell _{3}}$. Therefore $\ell _{4}=1$ in $\tau ^{\prime }$,
since $\ell _{2}=\ell _{3}=0$ in both $\tau $ and $\tau ^{\prime }$, and
thus $\gamma =(\overline{\ell _{1}}\vee \ell _{4})=1$ in $\tau ^{\prime }$.
Otherwise, if $\gamma $ was the clause $(\overline{\ell _{1}}\vee \overline{%
\ell _{1}}\vee \ell _{4})$ before the double literal elimination, then
Definition~\ref{gradually-mixed-def} implies that $\ell _{1}=\ell _{2}$, or $%
\ell _{1}=\ell _{3}$, or $\ell _{4}=\overline{\ell _{2}}$, or $\ell _{4}=%
\overline{\ell _{3}}$. Recall that $\alpha $ is a clause of $\phi
_{1}^{\prime }$ by assumption, and thus $\ell _{1}\neq \ell _{2}$ and $\ell
_{1}\neq \ell _{3}$. Therefore $\ell _{4}=\overline{\ell _{2}}$ or $\ell
_{4}=\overline{\ell _{3}}$, and thus $\ell _{4}=1$ in $\tau ^{\prime }$,
since $\ell _{2}=\ell _{3}=0$ in both $\tau $ and $\tau ^{\prime }$.
Therefore $\gamma =(\overline{\ell _{1}}\vee \ell _{4})=1$ in $\tau ^{\prime
}$. That is, all clauses $\gamma $ of $\phi _{1}^{\prime \prime }$ remain
satisfied in the assignment $\tau ^{\prime }$.

Summarizing, all clauses of the $2$-CNF formula $\phi _{1}^{\prime \prime
}\wedge \phi _{2}\wedge \phi _{0}$ remain satisfied in $\tau ^{\prime }$.
Furthermore, $\alpha =1$ in $\tau ^{\prime }$, while for every clause $\beta 
$ of $\phi _{1}^{\prime }$, if $\beta =1$ in $\tau $ then also $\beta =1$ in 
$\tau ^{\prime }$. Thus, according to the above transition from $\tau $ to $%
\tau ^{\prime }$, we can modify iteratively the truth assignment $\tau $ to
a truth assignment $\tau ^{\prime \prime }$ that satisfies all clauses of $%
\phi \wedge \phi _{0}$. Therefore $\phi \wedge \phi _{0}$ is satisfiable if
and only if the $2$-CNF formula $\phi _{1}^{\prime \prime }\wedge \phi
_{2}\wedge \phi _{0}$ is satisfiable.

Since the transition from the assignment $\tau $ to the assignment $\tau
^{\prime }$ can be done in constant time (we only need to flip locally the
value of one literal $\ell _{1}$ in the clause $\alpha =(\ell _{1}\vee \ell
_{2}\vee \ell _{3})$ of $\phi _{1}^{\prime }$), the computation of $\tau
^{\prime \prime }$ from $\tau $ can be done in time linear to the size of $%
\phi \wedge \phi _{0}$. Therefore, since a satisfying truth assignment $\tau 
$ of the $2$-CNF formula $\phi _{1}^{\prime \prime }\wedge \phi _{2}\wedge
\phi _{0}$ (if one exists) can be computed in linear time using any standard
linear time algorithm for the $2$-SAT problem (e.g.~\cite{Even76}), a
satisfying truth assignment~$\tau ^{\prime \prime }$of~$\phi \wedge \phi
_{0} $ (if one exists) can be also computed in time linear to the size of $%
\phi \wedge \phi _{0}$ (and thus also in time linear to the size of $\phi $%
). This completes the proof of the theorem.\qed
\end{proof}

The conditions of Definition~\ref{gradually-mixed-def} which guarantee the
tractability of gradually mixed formulas are \emph{minimal}, in the sense
that, if we remove any of these two conditions, the resulting subclass of $3$%
SAT is NP-complete.

Indeed, assume that we impose only the \emph{first} condition of Definition~%
\ref{gradually-mixed-def} to the mixed formula $\phi =\phi _{1}\wedge \phi
_{2}$. Then we can reduce $3$SAT to this subclass as follows. Let $\phi _{0}$
be an instance of $3$SAT. We define $\phi _{1}$ to be the formula obtained
by $\phi _{0}$ if we replace every literal $\ell $ of $\phi _{0}$ by a new
variable $x_{\ell }$. For every two of these new variables $x_{\ell }$ and $%
x_{\ell ^{\prime }}$ in $\phi _{1}$, we add to $\phi _{2}$ the clauses $%
(x_{\ell }\vee \overline{x_{\ell ^{\prime }}})\wedge (\overline{x_{\ell }}%
\vee x_{\ell ^{\prime }})$ if $\ell =\ell ^{\prime }$ in $\phi _{0}$, and we
add to $\phi _{2}$ the clauses $(x_{\ell }\vee x_{\ell ^{\prime }})\wedge (%
\overline{x_{\ell }}\vee \overline{x_{\ell ^{\prime }}})$ if $\ell =%
\overline{\ell ^{\prime }}$ in $\phi _{0}$. Then $\phi =\phi _{1}\wedge \phi
_{2}$ satisfies the first condition of Definition~\ref{gradually-mixed-def}
(since no two clauses of $\phi _{1}$ share any variable), while $\phi _{0}$
is satisfiable if and only if $\phi $ is satisfiable.

On the other hand, assume that we impose only the \emph{second} condition of
Definition~\ref{gradually-mixed-def} to the mixed formula $\phi =\phi
_{1}\wedge \phi _{2}$. Then, by setting $\phi_{2} = \emptyset$, we can
include in the resulting class \emph{every} $3$-CNF formula, and thus this
class is NP-complete.

\section{Preliminaries\label{prelim-sec}}

\subsection{Notation\label{notation-subsec}}

In the remainder of this article we consider finite, simple, and undirected
graphs. Given a graph $G$, we denote by $V(G)$ and $E(G)$ the sets of its
vertices and edges, respectively. An edge between two vertices $u$ and $v$
of a graph $G=(V,E)$ is denoted by $uv$, and in this case $u$ and $v$ are
said to be \emph{adjacent}. The \emph{neighborhood} of a vertex $u\in V$ is
the set $N(u)=\{v\in V\ |\ uv\in E\}$ of its adjacent vertices. The
complement of $G$ is denoted by $\overline{G}$, i.e.~$\overline{G}=(V,%
\overline{E})$, where $uv\in \overline{E}$ if and only if $uv\notin E$. For
any subset $E_{0}\subseteq E$ of the edges of $G$, we denote for simplicity $%
G-E_{0}=(V,E\setminus E_{0})$. A subset $S\subseteq V$ of its vertices
induces an \emph{independent set} in $G$ if $uv\notin E$ for every pair of
vertices $u,v\in S$. Furthermore, $S$ induces a \emph{clique} in $G$ if $%
uv\in E$ for every pair $u,v\in S$. For two graphs $G_{1}=(V,E_{1})$ and $%
G_{2}=(V,E_{2})$, we denote $G_{1}\subseteq G_{2}$ whenever $E_{1}\subseteq
E_{2}$. Moreover, we denote for simplicity by $G_{1}\cup G_{2}$ and $%
G_{1}\cap G_{2}$ the graphs $(V,E_{1}\cup E_{2})$ and $(V,E_{1}\cap E_{2})$,
respectively. A graph $G$ is a \emph{split graph} if its vertices can be
partitioned into a clique $K$ and an independent set $I$. Furthermore, $%
G=(V,E)$ is a \emph{threshold graph} if we can assign to each vertex $v\in V$
a real weight $a_{v}$, such that $uv\in E$ if and only if $a_{u}+a_{v}\geq 1$.

A \emph{proper }$k$\emph{-coloring} of a graph $G$ is an assignment of $k$
colors to the vertices of $G$, such that adjacent vertices are assigned
different colors. The smallest $k$ for which there exists a proper $k$%
-coloring of $G$ is the \emph{chromatic number} of $G$, denoted by $\chi (G)$%
. If $\chi (G)=2$ then $G$ is a \emph{bipartite}~graph; in this case the
vertices of $G$ are partitioned into two independent sets, the \emph{color
classes}. A bipartite graph $G$ is denoted by $G=(U,V,E)$, where $U$ and $V$
are its color classes and $E$ is the set of edges between them. For a
bipartite graph ${G=(U,V,E)}$, its \emph{bipartite complement} is the graph $%
{\widehat{G}=(U,V,\widehat{E})}$, where for two vertices $u\in U$ and $v\in
V $, ${uv\in \widehat{E}}$ if and only if ${uv\notin E}$. A bipartite graph~$%
{G=(U,V,E)}$ is a \emph{chain graph} if the vertices of each color class can
be ordered by inclusion of their neighborhoods, i.e.~${N(u)\subseteq N(v)}$
or ${N(v)\subseteq N(u)}$ for any two vertices $u,v$ in the same color
class. Note that chain graphs are closed under bipartite complementation,
i.e.~$G$ is a chain graph if and only if $\widehat{G}$ is a chain graph.

For any graph $G=(V,E)$ and any graph class $\mathcal{G}$, the $\mathcal{G}$%
\emph{-cover number} of $G$ is the smallest $k$ such that $%
E=\bigcup_{i=1}^{k}E_{i}$, where $G_{i}=(V,E_{i})\in \mathcal{G}$, $1\leq
i\leq k$; in this case the graphs $\{G_{i}\}_{i=1}^{k}$ are a $\mathcal{G}$%
\emph{-cover} of $G$. For several graph classes $\mathcal{G}$ it is
NP-complete to decide whether the $\mathcal{G}$-cover number of a graph is
at most $k$, where $k\geq 3$, see e.g.~\cite{Yannakakis82}.
Throughout the paper, whenever a set of the chain graphs $\{G_{i}\}_{i=1}^{k}$ form a
chain-cover of a bipartite graph $G$, then all these graphs are assumed to
have the same color classes as~$G$.

For any partial order $P=(U,R)$, we denote by $\overline{P}=(U,\overline{R})$
the \emph{inverse} partial order of $P$, i.e.~for any two elements $u,v\in U$%
, $u<_{\overline{P}}v$ if and only if $v<_{P}u$. For any two partial orders $%
P_{1}=(U,R_{1})$ and $P_{2}=(U,R_{2})$, we denote $P_{1}\subseteq P_{2}$
whenever $R_{1}\subseteq R_{2}$. Moreover, we denote for simplicity $%
P_{1}\cup P_{2}$ and $P_{1}\cap P_{2}$ for the partial orders $(U,R_{1}\cup
R_{2})$ and $(U,R_{1}\cap R_{2})$, respectively. If $P_{2}$ is a linear
order and $P_{1}\subseteq P_{2}$, then $P_{2}$ is a \emph{linear extension}
of $P_{1}$. The orders $P_{1}$ and $P_{2}$ \emph{contradict each other} if
there exist two elements $u,v\in U$ such that $u<_{P_{1}}v$ and $v<_{P_{2}}u$%
. The \emph{linear-interval dimension} of a partial order $P$ is the
lexicographically smallest pair $(k,\ell )$ such that $P=%
\bigcap_{i=1}^{k}P_{i}$, where $\{P_{i}\}_{i=1}^{k}$ are interval orders and
exactly $\ell $ among them are not linear orders. In particular, $P$ is a 
\emph{linear-interval order} if its linear-interval dimension is at most $%
(2,1)$, i.e.~$P=P_{1}\cap P_{2}$, where $P_{1}$ is a linear order and $P_{2}$
is an interval order.

\subsection{Threshold graphs and alternating cycles\label{threshold-alternating-cycles-subsec}}

In this section we provide preliminary definitions and known results on alternating cycles and on threshold graphs, which will be useful for the remainder of the paper. 

\begin{definition}
\label{AC-2k-def}Let ${G=(V,E)}$ be a graph, $\widetilde{E} \subseteq E$ be
an edge subset, and ${k\geq 2}$. A set of~$2k$ (not necessarily distinct)
vertices $v_{1},v_{2},\ldots ,v_{2k}\in V$ builds an \emph{alternating cycle 
}$AC_{2k}$ \emph{in }$\widetilde{E}$, if $v_{i}v_{i+1}\in \widetilde{E}$
whenever $i$ is even and $v_{i}v_{i+1}\notin E$ whenever $i$ is odd (where
indices are \hspace{-0.2cm}$\mod 2k$). Furthermore, we say that $G$ has an
alternating cycle $AC_{2k}$, whenever $G$ has an $AC_{2k}$ in the edge set $%
\widetilde{E}=E$.
\end{definition}

For instance, for $k=3$, there exist two different possibilities for an $%
AC_{6}$, which are illustrated in Figures~\ref{AC6-1-fig} and~\ref{AC6-2-fig}%
. These two types of an $AC_{6}$ are called an \emph{alternating path of
length}~$5$ or \emph{of length}~$6$, respectively ($AP_{5}$ and $AP_{6}$ for
short, respectively). In an $AP_{6}$ on vertices $%
v_{1},v_{2},v_{3},v_{4},v_{5},v_{6}$, if there exist the edges $v_{1}v_{3}$
and $v_{2}v_{6}$ (or, symmetrically, the edges $v_{3}v_{5}$ and $v_{4}v_{2}$%
, or the edges $v_{5}v_{1}$ and $v_{6}v_{4}$), then this $AP_{6}$ is called
a \emph{double} $AP_{6}$, cf.~Figure~\ref{AC6-3-fig}.

\begin{definition}
\label{base-ceiling-AP6-def}Let $G=(V,E)$ be a graph and $v_{1},\ldots
,v_{6} $ be the vertices of an $AP_{6}$. Then the non-edge $v_{1}v_{2}$
(resp.~the non-edge $v_{3}v_{4}$, $v_{5}v_{6}$) is a \emph{base} of the $%
AP_{6}$ and the edge $v_{4}v_{5}$ (resp.~the edge $v_{6}v_{1}$, $v_{2}v_{3}$%
) is the corresponding \emph{ceiling} of this $AP_{6}$.
\end{definition}

\begin{figure}[tbh]
\centering 
\subfigure[] {\label{AC6-1-fig} 
\includegraphics[scale=0.6]{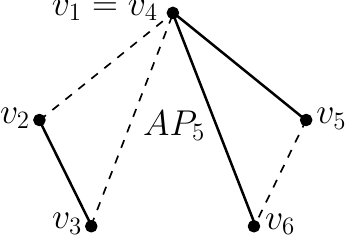}} \hspace{0.2cm} 
\subfigure[] {\label{AC6-2-fig} 
\includegraphics[scale=0.6]{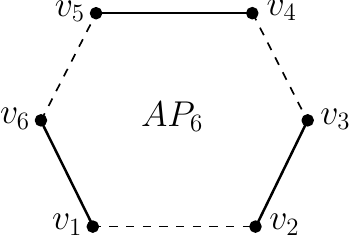}} \hspace{0.2cm} 
\subfigure[] {\label{AC6-3-fig} 
\includegraphics[scale=0.6]{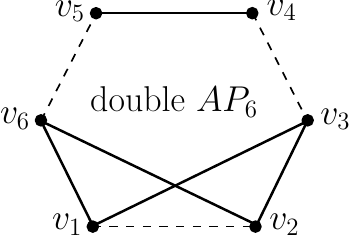}}
\caption{All possibilities for an $AC_{6}$: (a)~an alternating path $AP_{5}$
of length $5$, (b)~an alternating path $AP_{6}$ of length $6$, and (c)~a
double $AP_{6}$. The solid lines denote edges of the graph and the dashed
lines denote non-edges of the graph.}
\label{AC6-both-fig}
\end{figure}

Furthermore, note that for $k=2$, a set of four vertices $%
v_{1},v_{2},v_{3},v_{4}\in V$ builds an alternating cycle $AC_{4}$ if $%
v_{1}v_{2},v_{3}v_{4}\in E$ and $v_{1}v_{4},v_{2}v_{3}\notin E$. There are
three possible graphs on four vertices that build an alternating cycle $%
AC_{4}$, namely $2K_{2}$, $P_{4}$, and $C_{4}$, which are illustrated in
Figure~\ref{AC4-fig}.

\begin{figure}[tbh]
\centering 
\subfigure[] {\label{2K2-fig} 
\includegraphics[scale=0.6]{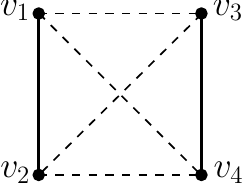}} \hspace{0.2cm} 
\subfigure[] {\label{P4-fig} 
\includegraphics[scale=0.6]{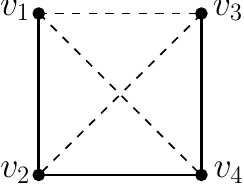}} \hspace{0.2cm} 
\subfigure[] {\label{C4-fig} 
\includegraphics[scale=0.6]{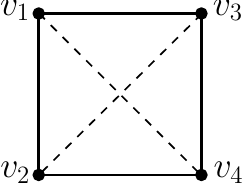}}
\caption{The three possible $AC_{4}$'s: (a)~a $2K_{2}$, (b)~a $P_{4}$, and
(c)~a $C_{4}$.}
\label{AC4-fig}
\end{figure}

Alternating cycles can be used to characterize threshold and chain graphs.
In particular, threshold graphs are the graphs with no induced $AC_{4}$, and chain
graphs are the bipartite graphs with no induced $2K_{2}$~\cite%
{MahadevPeled95}. We define now for any bipartite graph $G$ the associated
split graph of $G$, which we will use extensively in the remainder of the
paper.

\begin{definition}
\label{bipartite-split-def}Let $G=(U,V,E)$ be a bipartite graph. The \emph{%
associated split graph of} $G$ is the split graph $H_{G}=(U\cup V,E^{\prime
})$, where $E^{\prime }=E\cup (V\times V)$, i.e.~$H_{G}$ is the split graph
made by $G$ by replacing the independent set $V$ of $G$ by a clique.
\end{definition}

\begin{observation}
\label{bipartite-split-AC4-obs}Let $G$ be a bipartite graph and $H_{G}$ be
the associated split graph of $G$. Then, $G$ has an induced $2K_{2}$ if and
only if $H_{G}$ has an induced $AC_{4}$, and in this case this $AC_{4}$ is a 
$P_{4}$.
\end{observation}

The next lemma connects the chain cover number $ch(G)$ of a bipartite graph $%
G$ with the threshold cover number $t(H_{G})$ of the associated split graph $%
H_{G}$ of $G$. Recall that the problem of deciding whether a graph $G$ has
threshold cover number at most a given number $k$~is NP-complete for $k\geq
3 $~\cite{Yannakakis82}, while it is polynomial for $k=2$~\cite%
{RaschleSimon95}.

\begin{lemma}[\hspace{-0.001mm}\protect\cite{MahadevPeled95}]
\label{chain-threshold-number-lem}Let $G=(U,V,E)$ be a bipartite graph. Then 
$ch(G)=t(H_{G})$.
\end{lemma}

The next two definitions of a \emph{conflict} between two edges and the 
\emph{conflict graph} are essential for our results.

\begin{definition}
\label{conflict-edges-def}Let $G=(V,E)$ be a graph and $e_{1},e_{2}\in E$.
If the vertices of $e_{1}$ and $e_{2}$ build an $AC_{4}$ in $G$, then $e_{1}$
and $e_{2}$ are in \emph{conflict}, and in this case we denote $e_{1}||e_{2}$
in $G$. Furthermore, an edge $e\in E$ is \emph{committed} if there exists an
edge $e^{\prime }\in E$ such that $e||e^{\prime }$; otherwise $e$ is \emph{%
uncommitted}.
\end{definition}

\begin{definition}[\hspace{-0.001mm}\protect\cite{RaschleSimon95}]
\label{conflict-graph-def}Let ${G=(V,E)}$ be a graph. The \emph{conflict
graph} ${G^{\ast }=(V^{\ast },E^{\ast })}$ \emph{of} $G$ is defined~by%
\vspace{-0.1cm}

\begin{itemize}
\item $V^{\ast }=E$ and

\item for every $e_{1},e_{2}\in E$, $e_{1}e_{2}\in E^{\ast }$ if and only if 
$e_{1}||e_{2}$ in $G$.
\end{itemize}
\end{definition}

\begin{observation}
\label{uncommitted-isolated-obs}Let ${G=(V,E)}$ be a graph and let $e\in E$.
If $e$ is uncommitted, then $e$ is an isolated vertex in the conflict graph $%
G^{\ast }$ of $G$.
\end{observation}

\begin{observation}
\label{split-graph-clique-no-conflict-obs}Let $G=(V,E)$ be a split graph.
Let $K$ and $I$ be a partition of $V$, such that $K$ is a clique and $I$ is
an independent set (such a partition always exists for split graphs). Then,
every edge of $K$ is uncommitted.
\end{observation}

\begin{lemma}
\label{AC6-if-not-in-conflict-all-not-isolated-lem}Let $G$ be a graph and
let the vertices $v_{1},\ldots ,v_{6}$ of $G$ build an $AP_{6}$ (an
alternating path of length $6$). Assume that among the three edges $%
\{v_{2}v_{3},v_{4}v_{5},v_{6}v_{1}\}$ of this $AP_{6}$, no pair of edges is
in conflict. Then the edges $v_{3}v_{6},v_{4}v_{1},v_{5}v_{2}$ exist in $G$
and $v_{4}v_{5}||v_{3}v_{6}$, $v_{2}v_{3}||v_{4}v_{1}$, and $%
v_{6}v_{1}||v_{5}v_{2}$.
\end{lemma}

\begin{proof}
Suppose that $v_{3}v_{6}$ is not an edge of $G$. Then the edges $v_{2}v_{3}$
and $v_{6}v_{1}$ are in conflict, since $v_{1}v_{2}$ is not an edge of $G$
(cf.~Figure~\ref{AC6-2-fig}), which is a contradiction to the assumption of
the lemma. Therefore $v_{3}v_{6}$ is an edge of $G$. By symmetry, it follows
that also $v_{4}v_{1}$ and $v_{5}v_{2}$ are edges in $G$. Note now that the
edges $v_{4}v_{5}||v_{3}v_{6}$ are in conflict, since $v_{3}v_{4}$ and $%
v_{5}v_{6}$ are not edges of $G$. By symmetry, it follows that also $%
v_{2}v_{3}||v_{4}v_{1}$, and $v_{6}v_{1}||v_{5}v_{2}$.\qed
\end{proof}

Note that the threshold cover number $t(G)$ of a graph $G=(V,E)$ equals the
smallest $k$, such that the edge set $E$ of $G$ can be partitioned into $k$
sets $E_{1},E_{2},\ldots ,E_{k}$, each having a \emph{threshold completion in%
} $G$ (that is, there exists for every $i=1,2,\ldots ,k$ an edge set $%
E_{i}^{\prime }$, such that $E_{i}\subseteq E_{i}^{\prime }\subseteq E$ and $%
(V,E_{i}^{\prime })$ is a threshold graph). The following characterization
of subgraphs that admit a threshold completion in a given graph $G$ has been
proved in~\cite{HIP81}.

\begin{lemma}[\hspace{-0.001cm}\protect\cite{HIP81}]
\label{threshold-completion-lem}Let $H$ be a subgraph of a graph $G=(V,E)$.
Then $H$ has a threshold completion in $G$ if and only if $G$ has no $%
AC_{2k} $, $k\geq 2$, on the edges of $H$.
\end{lemma}

If the conditions of Lemma~\ref{threshold-completion-lem} are satisfied,
then such a threshold completion of $H$ in $G$ can be computed in linear
time, as the next lemma states.

\begin{lemma}[\hspace{-0.001cm}\protect\cite{RaschleSimon95}]
\label{threshold-completion-construction-lem}If a subgraph $H$ of $G=(V,E)$
has a threshold completion in $G$, then it can be computed in $O(|V|+|E|)$
time.
\end{lemma}

\begin{corollary}
\label{chain-number-2-char-cor}Let $G=(V,E)$ be a graph. Then, $t(G)=1$ if
and only if $G$ has no $AC_{2k}$, $k\geq 2$. Furthermore, $t(G)\leq 2$ if
and only if the set $E$ of edges can be partitioned into two sets $E_{1}$
and~$E_{2}$, such that $G$ has no $AC_{2k}$, $k\geq 2$, in each $E_{i}$, $%
i=1,2$.
\end{corollary}

\begin{proof}
First note that $t(G)=1$ if and only if $G$ is a threshold graph. Therefore,
Lemma~\ref{threshold-completion-lem} implies that $t(G)=1$ if and only if $G$
has no $AC_{2k}$, $k\geq 2$.

Recall that the threshold cover number $t(G)$ of a graph $G=(V,E)$ equals
the smallest $k$, such that the edge set $E$ of $G$ can be partitioned into $%
k$ sets $E_{1},E_{2},\ldots ,E_{k}$, each having a \emph{threshold
completion in} $G$. Therefore, if $t(G)\leq 2$, Lemma~\ref%
{threshold-completion-lem} implies that $E$ can be partitioned into two sets 
$E_{1}$ and~$E_{2}$, such that $G$ has no $AC_{2k}$, $k\geq 2$, in each $%
E_{i}$, $i=1,2$. Note that, in the case where $t(G)=1$ (i.e.~$G$ is a
threshold graph), we can set $E_{1}=E$ and $E_{2}=\emptyset $. Conversely,
suppose that $E$ can be partitioned into two such sets $E_{1}$ and $E_{2}$.
Then Lemma~\ref{threshold-completion-lem} implies that both graphs $%
G_{1}=(V,E_{1})$ and $G_{2}=(V,E_{2})$ have a threshold completion in $G$,
where $G_{1}\cup G_{2}=G$. Therefore $t(G)\leq 2$.\qed
\end{proof}

For every graph $G$, it can be easily proved that the chromatic number $\chi
(G^{\ast })$ of its conflict graph $G^{\ast }$ provides a lower bound for
the threshold cover number $t(G)$ of $G$, as the text lemma states.

\begin{lemma}[\hspace{-0.001mm}\protect\cite{MahadevPeled95}]
\label{chromatic-number-conflict-lower-bound-lem}Let $G$ be a graph. Then $%
\chi (G^{\ast })\leq t(G)$.
\end{lemma}

Lemma~\ref{chromatic-number-conflict-lower-bound-lem} immediately implies
that a necessary condition for a graph $G$ to have threshold cover number $%
t(G)\leq 2$ is that $\chi (G^{\ast })\leq 2$, i.e.~that $G^{\ast }$ is a
bipartite graph. The main result of~\cite{RaschleSimon95} is the next
theorem, which proves that this is also a sufficient condition for graphs $G$
with $\chi (G^{\ast })\leq 2$. 

\begin{theorem}[\hspace{-0.001cm}\protect\cite{RaschleSimon95}]
\label{main-result-simon95-thm}If the conflict graph $G^{\ast }$ of a graph $%
G=(V,E)$ is bipartite (i.e.~$\chi (G^{\ast })\leq 2$), then $t(G)\leq 2$.
Moreover, $E$ can be partitioned in $O(|E|(|V|+|E|))$ time into two sets~$%
E_{1}$ and~$E_{2}$, such that $G$ has no $AC_{2k}$, $k\geq 2$, in each $%
E_{i} $, $i=1,2$.
\end{theorem}

Due to the next theorem, it suffices for bipartite conflict graphs $G^{\ast
} $ to consider only small alternating cycles $AC_{2k}$ with $k\leq 3$.

\begin{theorem}[\hspace{-0.001cm}\protect\cite{HIP81}]
\label{AC6-suffices-thm}Suppose that the conflict graph $G^{\ast }$ of a
graph ${G=(V,E)}$ is bipartite (i.e.~${\chi (G^{\ast })\leq 2}$), with
(vertex) color classes $E_{1}$ and $E_{2}$. If $G$ has an $AC_{2k}$ on the
edges of~$E_{1}$ (resp.~of~$E_{2}$), where $k\geq 3$, then $G$ has also an $%
AC_{6}$ in $E_{1}$ (resp.~of $E_{2}$).
\end{theorem}

\begin{lemma}[\hspace{-0.001mm}\protect\cite{IbarakiPeled81}]
\label{split-AC6-alternating-lem}Let $G=(V,E)$ be a split graph. Let $K$ and 
$I$ be a partition~of~$V$~such~that~$K$ induces a clique and $I$ induces an
independent set in $G$. Assume that the vertices~${v_{1},\ldots,v_{6}}$
build an $AP_{6}$ in $G$. Then either $v_{1},v_{3},v_{5}\in K$ and $%
v_{2},v_{4},v_{6}\in I$, or $v_{1},v_{3},v_{5}\in I$ and~$%
v_{2},v_{4},v_{6}\in K$.
\end{lemma}

\begin{lemma}
\label{split-no-AP5-double-AP6-lem}Any split graph $G$ does not contain any $%
AP_{5}$ or any double~$AP_{6}$.
\end{lemma}

\begin{proof}
The fact that a split graph $G$ does not contain any $AP_{5}$ has been
proved in~\cite{IbarakiPeled81}. Let now $K$ and $I$ be a partition of the
vertices $V$ of $G$, such that $K$ induces a clique and $I$ induces an
independent set in $G$ (such a partition exists by definition, since $G$ is
a split graph). Suppose that $G$ has an $AP_{6}$ on the vertices $%
v_{1},v_{2},v_{3},v_{4},v_{5},v_{6}$, cf.~Figure~\ref{AC6-1-fig}. Then,
Lemma~\ref{split-AC6-alternating-lem} implies that either $%
v_{1},v_{3},v_{5}\in K$ and $v_{2},v_{4},v_{6}\in I$, or $%
v_{1},v_{3},v_{5}\in I$ and~$v_{2},v_{4},v_{6}\in K$. In both cases, none of
the pairs of edges $\{v_{1}v_{3},v_{2}v_{6}\}$, $\{v_{3}v_{5},v_{4}v_{2}\}$,
and $\{v_{5}v_{1},v_{6}v_{4}\}$ can exist simultaneously in $G$. Therefore, $%
G$ has no double $AP_{6}$. This completes the proof of the lemma.\qed
\end{proof}

\section{Linear-Interval covers of bipartite graphs\label{linear-interval-sec}}

In this section we introduce the crucial notion of a \emph{linear-interval cover} of bipartite graphs (cf.~Definition~\ref{linear-interval-coverable-def}). 
Then we use linear-interval covers to provide a new characterization of PI graphs (cf.~Theorem~\ref{linear-interval-coverable-char-thm}), 
which is one of the main tools for our PI graph recognition algorithm. 
First we provide in the next theorem the characterization of PI graphs using linear orders and interval orders.

\begin{theorem}
\label{PI-char-thm}Let $G=(V,E)$ be a cocomparability graph and $P$ be a
partial order of $\overline{G}$. Then~$G$ is a PI graph if and only if $%
P=P_{1}\cap P_{2}$, where $P_{1}$ is a linear order and $P_{2}$ is an
interval order.
\end{theorem}

\begin{proof}
For the purposes of the proof, a partial order $P=(U,R)$ is called a PI
order~\cite{Cerioli08}, if there exists a PI representation (i.e.~a
simple-triangle representation) $R$, such that for any two $u,v\in U$, $%
u<_{P}v$ if and only if the triangle associated to $u$ lies in $R$ entirely
to the left of the triangle associated to $v$.

Suppose that $P=P_{1}\cap P_{2}$ for two partial orders $P_{1}$ and $P_{2}$,
where $P_{1}$ is a linear order and $P_{2}$ is an interval order. Then $P$
is a~PI~order~\cite{Cerioli08}, and thus $G$ is a~PI~graph. Conversely,
suppose that $G$ is a~PI~graph. Equivalently, $P$ is a~PI~order, and thus $%
lidim(P)\leq (2,1)$~\cite{Cerioli08}. That is, $P=P_{1}\cap P_{2}$ for two
partial orders $P_{1}$ and $P_{2}$, where $P_{1}$ is a linear order and $%
P_{2}$ is an interval order. Moreover, whenever we are given a partial order 
$P$ such that $P=P_{1}\cap P_{2}$, where $P_{1}$ is a linear order and $%
P_{2} $ is an interval order, it is straightforward to compute a PI model
for $P$ (cf.~\cite{Cerioli08}). Equivalently, we can easily construct in
this case a PI representation of the incomparability graph $G$~of $P$
(cf.~lines~\ref{alg-PI-10}-\ref{alg-PI-12} of Algorithm~\ref%
{alg-PI-repr-constr} below).\qed
\end{proof}

For every partial order $P$ we define now the domination bipartite graph $C(P)$, which has been used to characterize interval orders~\cite{MaSpinrad94}. 
Here ``C'' stands for ``Comparable'', since the definition of $C(P)$ uses the comparable elements~of~$P$.

\begin{definition}[\hspace{-0.001cm}\protect\cite{MaSpinrad94}]
\label{C(P)-def}Let ${P=(U,R)}$ be a partial order, where ${%
U=\{u_{1},u_{2},\ldots ,u_{n}\}}$. Furthermore let~${V=\{v_{1},v_{2},\ldots
,v_{n}\}}$. The \emph{domination bipartite graph} $C(P)=(U,V,E)$ is defined
such that ${u_{i}v_{j}\in E}$ if~and~only if~$u_{i}<_{P}u_{j}$.
\end{definition}

\begin{lemma}[\hspace{-0.001cm}\protect\cite{MaSpinrad94}]
\label{interval-order-C(P)-lem}Let $P=(U,R)$ be a partial order. Then, $P$
is an interval order if and only if $C(P)$ is a chain graph.
\end{lemma}

Extending the notion of $C(P)$, we now introduce the bipartite graph $NC(P)$ to characterize linear orders (cf.~Lemma~\ref{linear-order-NC(P)-lem}). 
Here ``NC'' stands for ``Non-strictly Comparable''. Namely, this graph
can be obtained by adding to the graph $C(P)$ the perfect matching $\{u_{i}v_{i}\ |\ i=1,2,\ldots ,n\}$ on the vertices of $U$ and $V$.

\begin{definition}
\label{NC(P)-def}Let ${P=(U,R)}$ be a partial order, where ${%
U=\{u_{1},u_{2},\ldots ,u_{n}\}}$. Furthermore let~${V=\{v_{1},v_{2},\ldots
,v_{n}\}}$. Then, $NC(P)=(U,V,E)$ is the bipartite graph, such that $%
u_{i}v_{j}\in E$ if~and~only if $u_{i}\leq _{P}u_{j}$.
\end{definition}

\begin{lemma}
\label{linear-order-NC(P)-lem}Let $P=(U,R)$ be a partial order. Then, $P$ is
a linear order if and only if $NC(P)$ is a chain graph.
\end{lemma}

\begin{proof}
Let $U=\{u_{1},u_{2},\ldots ,u_{n}\}$. Suppose that $P$ is a linear order,
i.e.~${u_{1}<_{P}u_{2}<_{P}\ldots <_{P}u_{n}}$. Then, by Definition~\ref%
{NC(P)-def}, the set of neighbors of a vertex $u_{i}\in U$ in the graph $%
NC(P)$ is~${N(u_{i})=\{v_{i},v_{i+1},\ldots ,v_{n}\}}$. Therefore, $%
N(u_{n})\subset N(u_{n-1})\subset \ldots \subset N(u_{1})$, and thus $NC(P)$
is a chain graph.

Suppose now that $NC(P)$ is a chain graph. Then the sets of neighbors of the
vertices of~$U$ in the graph $NC(P)$ can be linearly ordered by inclusion.
Let without loss of generality $N(u_{1})\subseteq N(u_{2})\subseteq \ldots
\subseteq N(u_{n})$. Therefore, since $v_{i}\in N(u_{i})$ in $NC(P)$ for
every $i=1,2,\ldots ,n$, it follows that $v_{i}\in N(u_{j})$ in $NC(P)$
whenever $i<j$. Therefore, by Definition~\ref{NC(P)-def}, $u_{j}<_{P}u_{i}$
whenever $i<j$. That is, $u_{n}<_{P}u_{n-1}<_{P}\ldots <_{P}u_{1}$, i.e.~$P$
is a linear order.\qed
\end{proof}

We introduce now the notion of a \emph{linear-interval cover} of a bipartite
graph. This notion is crucial for our main result of this section,
cf.~Theorem~\ref{linear-interval-coverable-char-thm}.

\begin{definition}
\label{linear-interval-coverable-def}Let $G=(U,V,E)$ be a bipartite graph,
where $U=\{u_{1},u_{2},\ldots ,u_{n}\}$ and $V=\{v_{1},v_{2},\ldots ,v_{n}\}$%
. Let $E_{0}=\{u_{i}v_{i}\ |\ 1\leq i\leq n\}$ and suppose that $%
E_{0}\subseteq E$. Then, $G$ is \emph{linear-interval coverable} if there
exist two chain graphs ${G_{1}=(U,V,E_{1})}$~and~${G_{2}=(U,V,E_{2})}$%
,~such~that ${G=G_{1}\cup G_{2}}$ and ${E_{0}\subseteq E_{2}\setminus E_{1}}$%
. In this case, the sets ${\{E_{1},E_{2}\}}$ are a~\emph{%
linear-interval~cover}~of~$G$.
\end{definition}

Before we proceed with Theorem~\ref{linear-interval-coverable-char-thm}, we
first provide the next auxiliary lemma.

\begin{lemma}
\label{common-linear-extension-lem}Let $Q_{1}=(U,R_{1})$ be an interval
order and $Q_{2}=(U,R_{2})$ be a partial order, such that $Q_{1}$ and $Q_{2}$
do not contradict each other. Then there exists a linear order $Q_{0}$ that
is a linear extension of both $Q_{1}$ and $Q_{2}$.
\end{lemma}

\begin{proof}
Let ${U=\{u_{1},u_{2},\ldots ,u_{n}\}}$ be the ground set of $Q_{1}$ and $%
Q_{2}$. Furthermore let $C(Q_{1})=(U,V,E_{1})$ be the domination bipartite
graph of $Q_{1}$, where $V=\{v_{1},v_{2},\ldots ,v_{n}\}$, cf.~Definition %
\ref{C(P)-def}. Since $Q_{1}$ is an interval order by assumption, $C(Q_{1})$
is a chain graph by Lemma~\ref{interval-order-C(P)-lem}, i.e.~$C(Q_{1})$
does not contain an induced $2K_{2}$. Consider now two edges $u_{i}v_{j}$
and $u_{k}v_{\ell }$ of $C(Q_{1})$, where $\{i,j\}\cap \{k,\ell \}=\emptyset 
$. Then $u_{i}<_{Q_{1}}u_{j}$ and $u_{k}<_{Q_{1}}u_{\ell }$ by Definition %
\ref{C(P)-def}. Furthermore, at least one of the edges $u_{i}v_{\ell }$ and $%
u_{k}v_{j}$ exists in $C(Q_{1})$, since otherwise the edges $u_{i}v_{j}$ and 
$u_{k}v_{\ell }$ induce a $2K_{2}$ in $C(Q_{1})$, which is a contradiction.
Therefore $u_{i}<_{Q_{1}}u_{\ell }$ or $u_{k}<_{Q_{1}}u_{j}$.

Since $Q_{1}$ and $Q_{2}$ do not contradict each other by assumption, we can
define the simple directed graph $G_{0}=(U,E)$, such that $\overrightarrow{%
u_{i}u_{j}}\in E$ if and only if $u_{i}<_{Q_{1}}u_{j}$ or $%
u_{i}<_{Q_{2}}u_{j}$. We will prove that $G_{0}$ is acyclic. Suppose
otherwise that $G_{0}$ has at least one directed cycle, and let $C $ be a
directed cycle of $G_{0}$ with the smallest possible length. Assume first
that $C$ has length $3$, and let its edges be $\overrightarrow{u_{i}u_{j}}$, 
$\overrightarrow{u_{j}u_{k}}$, and $\overrightarrow{u_{k}u_{i}} $. Then at
least two of these edges belong to $Q_{1}$ or to $Q_{2}$. Let without loss
of generality $\overrightarrow{u_{i}u_{j}}$ and $\overrightarrow{u_{j}u_{k}}$
belong to $Q_{1}$, i.e.~$u_{i}<_{Q_{1}}u_{j}$ and $u_{j}<_{Q_{1}}u_{k}$.
Then also $u_{i}<_{Q_{1}}u_{k}$, since $Q_{1}$ is transitive, and thus $%
\overrightarrow{u_{i}u_{k}}\in E$. This contradicts the assumption that $%
\overrightarrow{u_{k}u_{i}}$ is an edge of $C$. Assume now that $C$ has
length greater than $3$. Suppose that two consecutive edges $\overrightarrow{%
u_{i}u_{j}}$ and $\overrightarrow{u_{j}u_{k}}$ of $C$ belong to $Q_{1}$,
i.e.~$u_{i}<_{Q_{1}}u_{j}$ and $u_{j}<_{Q_{1}}u_{k}$. Then also $%
u_{i}<_{Q_{1}}u_{k}$, since $Q_{1}$ is transitive, and thus $\overrightarrow{%
u_{i}u_{k}}\in E$. Therefore we can replace in $C$ the edges $%
\overrightarrow{u_{i}u_{j}}$ and $\overrightarrow{u_{j}u_{k}}$ by the edge $%
\overrightarrow{u_{i}u_{k}}$, obtaining thus a smaller directed cycle than $%
C $, which is a contradiction by the assumption on $C$. Thus no two
consecutive edges of $C$ belong to $Q_{1}$. Similarly, no two consecutive
edges of $C$ belong to $Q_{2}$, and thus the edges of $C$ belong alternately
to $Q_{1}$ and $Q_{2}$. In particular, $C$ has even length.

Consider now three consecutive edges $\overrightarrow{u_{i}u_{j}},%
\overrightarrow{u_{j}u_{k}},\overrightarrow{u_{k}u_{\ell }}$ of $C$, where $%
\overrightarrow{u_{i}u_{j}}$ and $\overrightarrow{u_{k}u_{\ell }}$ belong to 
$Q_{1}$. Then $u_{i}<_{Q_{1}}u_{j}$ and $u_{k}<_{Q_{1}}u_{\ell }$, where $%
\{i,j\}\cap \{k,\ell \}=\emptyset $, and thus $u_{i}<_{Q_{1}}u_{\ell }$ or $%
u_{k}<_{Q_{1}}u_{j}$, as we proved above. That is, $\overrightarrow{%
u_{i}u_{\ell }}\in E$ or $\overrightarrow{u_{k}u_{j}}\in E$. Therefore,
since we assumed that $\overrightarrow{u_{j}u_{k}}$ is an edge of $C$, it
follows that $\overrightarrow{u_{k}u_{j}}\notin E$, and thus $%
\overrightarrow{u_{i}u_{\ell }}\in E$. Therefore, in particular, $%
\overrightarrow{u_{\ell }u_{i}}\notin E$, and thus $C$ does not have length $%
4$, i.e.~it has length at least $6$. Thus we can replace in $C$ the edges $%
\overrightarrow{u_{i}u_{j}},\overrightarrow{u_{j}u_{k}},\overrightarrow{%
u_{k}u_{\ell }}$ by the edge $\overrightarrow{u_{i}u_{\ell }}$, obtaining
thus a smaller directed cycle than $C$, which is a contradiction by the
assumption on $C$.

Therefore, there exists no directed cycle in $G_{0}$, i.e.~$G_{0}$ is a
directed acyclic graph. Thus any topological ordering of $G_{0}$ corresponds
to a linear order $Q_{0}=(U,R_{0})$ that is a linear extension of both $%
Q_{1} $ and $Q_{2}$. This completes the proof of the lemma.\qed
\end{proof}

\begin{theorem}
\label{linear-interval-coverable-char-thm}Let ${P=(U,R)}$ be a partial
order. In the bipartite complement $\widehat{C}(P)$ of the graph $C(P)$,
denote ${E_{0}=\{u_{i}v_{i}\ |\ 1\leq i\leq n\}}$. The following statements
are equivalent:\vspace{-0.1cm}

\begin{enumerate}
\item[(a)] $P=P_{1}\cap P_{2}$, where $P_{1}$ is a linear order and $P_{2}$
is an interval order.

\item[(b)] $\widehat{C}(P)=\widehat{NC}(P_{1})\cup \widehat{C}(P_{2})$ for
two partial orders $P_{1}$ and $P_{2}$ on $V$, where $\widehat{NC}(P_{1})$
and $\widehat{C}(P_{2})$ are chain graphs.

\item[(c)] $\widehat{C}(P)$ is linear-interval coverable, i.e.~$\widehat{C}%
(P)=G_{1}\cup G_{2}$ for two chain graphs $G_{1}=(U,V,E_{1})$ and $%
G_{2}=(U,V,E_{2})$, where $E_{0}\subseteq E_{2}\setminus E_{1}$.
\end{enumerate}
\end{theorem}

\begin{proof}
(a) $\Rightarrow $ (b). Since $P_{1}$ is a linear order, it follows by Lemma~%
\ref{linear-order-NC(P)-lem} that $NC(P_{1})$ is a chain graph. Furthermore,
sine $P_{2}$ is an interval order, it follows by Lemma~\ref%
{interval-order-C(P)-lem} that $C(P_{2})$ is a chain graph. Therefore, since
the class of chain graphs is closed under bipartite complementation, it
follows that $\widehat{NC}(P_{1})$ and $\widehat{C}(P_{2})$ are chain graphs.

Let $u_{i},u_{j}\in U$ such that $u_{i}v_{j}\in E(C(P))$. Then $%
u_{i}<_{P}u_{j}$ by Definition~\ref{C(P)-def}. Furthermore, since $%
P=P_{1}\cap P_{2}$ by assumption, it follows that $u_{i}<_{P_{1}}u_{j} $ and 
$u_{i}<_{P_{2}}u_{j}$, and thus also $u_{i}v_{j}\in E(NC(P_{1}))$ and $%
u_{i}v_{j}\in E(C(P_{2}))$ by Definitions~\ref{C(P)-def} and~\ref{NC(P)-def}%
, respectively. Therefore $C(P)\subseteq NC(P_{1})\cap C(P_{2})$.

Let now $u_{i},u_{j}\in U$ such that $u_{i}v_{j}\in E(NC(P_{1}))$ and $%
u_{i}v_{j}\in E(C(P_{2}))$. Then, it follows in particular that $u_{i}\neq
u_{j}$ (since otherwise $u_{i}v_{j}\notin E(C(P_{2}))$, a contradiction).
Thus, $u_{i}<_{P_{1}}u_{j}$ and $u_{i}<_{P_{2}}u_{j}$ by Definitions~\ref%
{C(P)-def} and~\ref{NC(P)-def}. Therefore, since $P=P_{1}\cap P_{2}$ by
assumption, it follows that $u_{i}<_{P}u_{j}$, and thus $u_{i}v_{j}\in
E(C(P))$ by Definition~\ref{C(P)-def}. That is, $NC(P_{1})\cap
C(P_{2})\subseteq C(P)$. Summarizing, $C(P)=NC(P_{1})\cap C(P_{2})$, and
thus also $\widehat{C}(P)=\widehat{NC}(P_{1})\cup \widehat{C}(P_{2})$.

\medskip

(b) $\Rightarrow $ (a). Since $\widehat{C}(P)=\widehat{NC}(P_{1})\cup 
\widehat{C}(P_{2})$, it follows that $C(P)=NC(P_{1})\cap C(P_{2})$. Let $%
u_{i},u_{j}\in U$ such that $u_{i}<_{P}u_{j}$. Then $u_{i}v_{j}\in E(C(P))$
by Definition~\ref{C(P)-def}. Therefore, since $C(P)=NC(P_{1})\cap C(P_{2})$%
, it follows that also $u_{i}v_{j}\in E(NC(P_{1}))$ and $u_{i}v_{j}\in
E(C(P_{2}))$. Thus, in particular, $u_{i}\neq u_{j}$ (since otherwise $%
u_{i}v_{j}\notin E(C(P_{2}))$, a contradiction). Therefore $%
u_{i}<_{P_{1}}u_{j}$ and $u_{i}<_{P_{2}}u_{j}$ by Definitions~\ref{C(P)-def}
and~\ref{NC(P)-def}. That is, $P\subseteq P_{1}\cap P_{2}$.

Let now $u_{i},u_{j}\in U$ such that $u_{i}<_{P_{1}}u_{j}$ and $%
u_{i}<_{P_{2}}u_{j}$. Then $u_{i}v_{j}\in E(NC(P_{1}))$ and $u_{i}v_{j}\in
E(C(P_{2}))$ by Definitions~\ref{C(P)-def} and~\ref{NC(P)-def}. Therefore,
since $C(P)=NC(P_{1})\cap C(P_{2})$, it follows that also $u_{i}v_{j}\in
E(C(P))$. Thus $u_{i}<_{P}u_{j}$ by Definition~\ref{C(P)-def}. That is, $%
P_{1}\cap P_{2}\subseteq P$. Summarizing, $P=P_{1}\cap P_{2}$. Furthermore,
since by assumption $\widehat{NC}(P_{1})$ and $\widehat{C}(P_{2})$ are chain
graphs, it follows that also $NC(P_{1})$ and $C(P_{2})$ are chain graphs.
Therefore $P_{1}$ is a linear order and $P_{2}$ is an interval order by
Lemmas~\ref{linear-order-NC(P)-lem} and~\ref{interval-order-C(P)-lem},
respectively.

\medskip

(b) $\Rightarrow $ (c). Define $G_{1}=\widehat{NC}(P_{1})$ and $G_{2}=%
\widehat{C}(P_{2})$. Then, it follows by (b) that $G_{1}$ and $G_{2}$ are
chain graphs and that $\widehat{C}(P)=G_{1}\cup G_{2}$. Note now by
Definitions~\ref{C(P)-def} and~\ref{NC(P)-def} that $E_{0}\cap
E(C(P_{2}))=\emptyset $ and that $E_{0}\subseteq E(NC(P_{1}))$,
respectively. Therefore $E_{0}\subseteq E(\widehat{C}(P_{2}))\setminus E(%
\widehat{NC}(P_{1}))$. Thus, since $E_{2}=E(G_{2})=E(\widehat{C}(P_{2}))$
and $E_{1}=E(G_{1})=E(\widehat{NC}(P_{1}))$, it follows that $E_{0}\subseteq
E_{2}\setminus E_{1}$. That is, $\widehat{C}(P)$ is linear-interval
coverable by Definition~\ref{linear-interval-coverable-def}.

\medskip

(c) $\Rightarrow $ (b). We will construct from the edge sets $E_{1}$ and $%
E_{2}$ of $G_{1}$ and $G_{2}$, respectively, a linear order $P_{1}$ and an
interval order~$P_{2}$, such that $\widehat{C}(P)=\widehat{NC}(P_{1})\cup 
\widehat{C}(P_{2})$. Denote first the bipartite complement $\widehat{G}_{2}$
of $G_{2}$ as $\widehat{G}_{2}=(U,V,\widehat{E}_{2})$. Note that $\widehat{G}%
_{2}$ is a chain graph, since $G_{2}$ is also a chain graph by assumption.

\vspace{0.1cm}

\textbf{The interval order }$P_{2}$\textbf{.} We define $P_{2}$, such that $%
u_{i}<_{P_{2}}u_{j}$ if and only if $u_{i}v_{j}\in \widehat{E}_{2}$. We will
now prove that $P_{2}$ is a partial order. Recall that $E_{0}\subseteq E_{2}$
by assumption, and thus $E_{0}\cap \widehat{E}_{2}=\emptyset $. That is, $%
u_{i}v_{i}\notin \widehat{E}_{2}$ for every $i=1,2,\ldots ,n$. Furthermore, $%
\widehat{G}_{2}$ is a chain graph, since $G_{2}$ is a chain graph by
assumption. Therefore, for two distinct indices $i,j$, at most one of the
edges $u_{i}v_{j}$ and $u_{j}v_{i}$ belongs to $\widehat{E}_{2}$, since
otherwise these two edges would induce a $2K_{2}$ in $\widehat{G}_{2}$,
which is a contradiction. Thus, according to our definition of $P_{2}$,
whenever $i\neq j$, it follows that either $u_{i}<_{P_{2}}u_{j}$, or $%
u_{j}<_{P_{2}}u_{i}$, or $u_{i}$ and $u_{j}$ are incomparable in $P_{2}$.
Suppose that $u_{i}<_{P_{2}}u_{j}$ and $u_{j}<_{P_{2}}u_{k}$ for three
indices $i,j,k$. That is, $u_{i}v_{j},u_{j}v_{k}\in \widehat{E}_{2}$ by
definition of $P_{2}$. Since $\widehat{G}_{2}=(U,V,\widehat{E}_{2})$ is a
chain graph, the edges $u_{i}v_{j}$ and $u_{j}v_{k}$ do not build a $2K_{2}$
in $\widehat{G}_{2}$. Therefore, since $u_{j}v_{j}\notin \widehat{E}_{2}$,
it follows that $u_{i}v_{k}\in \widehat{E}_{2}$, i.e.~$u_{i}<_{P_{2}}u_{k}$.
That is, $P_{2}$ is transitive, and thus $P_{2}$ is a partial order.
Furthermore, note by the definition of $P_{2}$ and by Definition~\ref%
{C(P)-def} that $\widehat{G}_{2}=C(P_{2})$. Therefore, since $\widehat{G}%
_{2} $ is a chain graph, it follows by Lemma~\ref{interval-order-C(P)-lem}
that $P_{2}$ is an interval order.

\vspace{0.1cm}

In order to define the linear order $P_{1}$, we first define two auxiliary
orders $Q_{1}$ and $Q_{2}$, as follows.

\vspace{0.1cm}

\textbf{The interval order }$Q_{1}$\textbf{.} We define $Q_{1}$, such that $%
u_{i}<_{Q_{1}}u_{j}$ if and only if $u_{i}v_{j}\in E_{1}$. We will prove
that $Q_{1}$ is a partial order. Recall that $E_{0}\cap E_{1}=\emptyset $ by
assumption. That is, $u_{i}v_{i}\notin E_{1}$ for every $i=1,2,\ldots ,n$.
Furthermore, for two distinct indices $i,j$, at most one of the edges $%
u_{i}v_{j}$ and $u_{j}v_{i}$ belongs to $E_{1}$. Indeed, otherwise these two
edges would induce a $2K_{2}$ in $G_{1}$, which is a contradiction since $%
G_{1}$ is a chain graph by assumption. Thus, according to our definition of $%
Q_{1}$, whenever $i\neq j$, it follows that either $u_{i}<_{Q_{1}}u_{j}$, or 
$u_{j}<_{Q_{1}}u_{i}$, or $u_{i}$ and $u_{j}$ are incomparable in $Q_{1}$.
Suppose that $u_{i}<_{Q_{1}}u_{j}$ and $u_{j}<_{Q_{1}}u_{k}$ for three
indices $i,j,k$. That is, $u_{i}v_{j},u_{j}v_{k}\in E_{1}$ by definition of $%
Q_{1}$. Since $G_{1}$ is a chain graph by assumption, the edges $u_{i}v_{j}$
and $u_{j}v_{k}$ do not build a $2K_{2}$ in $G_{1}$. Therefore, since $%
u_{j}v_{j}\notin E_{1}$, it follows that $u_{i}v_{k}\in E_{1}$, i.e.~$%
u_{i}<_{Q_{1}}u_{k}$. That is, $Q_{1}$ is transitive, and thus $Q_{1}$ is a
partial order. Furthermore, note by the definition of $Q_{1}$ and by
Definition~\ref{C(P)-def} that $G_{1}=C(Q_{1})$. Therefore $Q_{1}$ is an
interval order by Lemma~\ref{interval-order-C(P)-lem}, since $G_{1}$ is a
chain graph by assumption.

\vspace{0.1cm}

\textbf{The partial order }$Q_{2}$\textbf{.} We define the partial order $%
Q_{2}$ as the inverse partial order $\overline{P}$ of $P$. That is, $%
u_{i}<_{Q_{2}}u_{j}$ if and only if $u_{j}<_{P}u_{i}$. Note that $Q_{2}$ is
transitive, since $P$ is transitive.

\vspace{0.1cm}

Before we define the linear order $P_{1}$, we first prove that the partial
orders $Q_{1}$ and $Q_{2}$ do not contradict each other. Suppose otherwise
that $u_{i}<_{Q_{1}}u_{j}$ and $u_{j}<_{Q_{2}}u_{i}$, for some pair $%
u_{i},u_{j}$. Then, since $u_{i}<_{Q_{1}}u_{j}$, it follows that $%
u_{i}v_{j}\in E_{1}$ by definition of $Q_{1}$. Therefore $u_{i}v_{j}\in E(%
\widehat{C}(P))$, since $\widehat{C}(P)=G_{1}\cup G_{2}$ by assumption. On
the other hand, since $u_{j}<_{Q_{2}}u_{i}$, it follows that $%
u_{i}<_{P}u_{j} $ by definition of $Q_{2}$. Therefore $u_{i}v_{j}\in E(C(P))$
by Definition~\ref{C(P)-def}, and thus $u_{i}v_{j}\notin E(\widehat{C}(P))$,
which is a contradiction. Therefore the partial orders $Q_{1}$ and $Q_{2}$
do not contradict each other.

\vspace{0.1cm}

\textbf{The linear order }$P_{1}$\textbf{.} Since the interval order $Q_{1}$
and the partial order $Q_{2}$ do not contradict each other, we can construct
by Lemma~\ref{common-linear-extension-lem} a common linear extension $Q_{0}$
of $Q_{1}$ and $Q_{2}$. That is, if $u_{i}<_{Q_{1}}u_{j}$ or $%
u_{i}<_{Q_{2}}u_{j}$, then $u_{i}<_{Q_{0}}u_{j}$. We define now the linear
order $P_{1}$ as the inverse linear order $\overline{Q_{0}}$ of $Q_{0}$.
Note that $P_{1}$ is also a linear extension of $P$, since $u_{i}<_{P}u_{j}$
implies that $u_{j}<_{Q_{2}}u_{i}$, which in turn implies that $%
u_{i}<_{P_{1}}u_{j}$.

\vspace{0.1cm}Now we prove that $\widehat{C}(P)\subseteq \widehat{NC}%
(P_{1})\cup \widehat{C}(P_{2})$. Let $u_{i}v_{j}\in E_{1}$. Then ${%
u_{i}<_{Q_{1}}u_{j}}$ by the definition of $Q_{1}$, and thus $%
u_{j}<_{P_{1}}u_{i}$ by the definition of $P_{1}$. Therefore $u_{i}\nleq
_{P_{1}}u_{j}$, and thus ${u_{i}v_{j}\notin E(NC(P_{1}))}$ by Definition~\ref%
{NC(P)-def}. Therefore $u_{i}v_{j}\in E(\widehat{NC}(P_{1}))$. Thus $%
E_{1}\subseteq E(\widehat{NC}(P_{1}))$, i.e.~${G_{1}\subseteq \widehat{NC}%
(P_{1})}$. Recall now that $\widehat{C}(P)=G_{1}\cup G_{2}$ by assumption.
Furthermore recall that $\widehat{G}_{2}=C(P_{2})$ as we proved above, and
thus $G_{2}=\widehat{C}(P_{2})$. Therefore, since $G_{1}\subseteq \widehat{NC%
}(P_{1})$, it follows that~$\widehat{C}(P) \subseteq \widehat{NC}(P_{1})
\cup \widehat{C}(P_{2})$.

Finally we prove that $C(P)\subseteq NC(P_{1})\cap C(P_{2})$. Consider now
an edge $u_{i}v_{j}\in E(C(P))$. Then $u_{i}<_{P}u_{j}$ by Definition~\ref%
{C(P)-def}, and thus $u_{j}<_{Q_{2}}u_{i}$ by the definition of~$Q_{2}$.
Furthermore $u_{i}<_{P_{1}}u_{j}$ by the definition of $P_{1}$, and thus $%
u_{i}v_{j}\in E(NC(P_{1}))$ by Definition~\ref{NC(P)-def}. Note now that $%
C(P)=\widehat{G}_{1}\cap \widehat{G}_{2}$, since $\widehat{C}(P)=G_{1}\cup
G_{2}$ by assumption. Therefore, since $u_{i}v_{j}\in E(C(P))$ by
assumption, it follows that also $u_{i}v_{j}\in \widehat{E}_{2}$. That is,
if $u_{i}v_{j}\in E(C(P))$ then $u_{i}v_{j}\in E(NC(P_{1}))$ and $%
u_{i}v_{j}\in \widehat{E}_{2}$. Therefore, since $\widehat{G}_{2}=C(P_{2})$,
it follows that $C(P)\subseteq NC(P_{1})\cap C(P_{2})$.

Summarizing, since $\widehat{C}(P)\subseteq \widehat{NC}(P_{1})\cup \widehat{%
C}(P_{2})$ and $C(P)\subseteq NC(P_{1})\cap C(P_{2})$, it follows that~$%
\widehat{C}(P)=\widehat{NC}(P_{1})\cup \widehat{C}(P_{2})$. This completes
the proof of the theorem.\qed
\end{proof}

The next corollary follows now easily by Theorems~\ref{PI-char-thm} and~\ref%
{linear-interval-coverable-char-thm}.

\begin{corollary}
\label{linear-interval-col}Let $G=(V,E)$ be a cocomparability graph and $P$
be a partial order of $\overline{G}$. Then, $G$ is a PI graph if and only if
the bipartite graph $\widehat{C}(P)$ is linear-interval coverable.
\end{corollary}

We now present Algorithm~\ref{alg-PI-repr-constr}, which constructs a PI
representation $R$ of a cocomparability graph~$G$ by a linear-interval cover 
$\{E_{1},E_{2}\}$ of the bipartite graph $\widehat{C}(P)$ (cf.~Definition~%
\ref{linear-interval-coverable-def}). Since~${E_{0}\subseteq E_{2}\setminus
E_{1}}$ by Definition~\ref{linear-interval-coverable-def}, where~${%
E_{0}=\{u_{i}v_{i}\ |\ 1\leq i\leq n\}}$ and $n$ is the number of vertices
of~$G$, note that $i\neq j$ during the execution of each of the lines~\ref%
{alg-PI-5b},~\ref{alg-PI-6b}, and~\ref{alg-PI-7b} of Algorithm~\ref%
{alg-PI-repr-constr}.

\begin{algorithm}[thb]
\caption{Construction of a PI representation, given a linear-interval cover} \label{alg-PI-repr-constr}
\begin{algorithmic}[1]
\REQUIRE{A cocomparability graph $G$, a partial order $P$ of $\overline{G}$, the domination bipartite graph ${C(P) = (U,V,E)}$, 
and a linear-interval cover~${\{E_{1},E_{2}\}}$ of~${\widehat{C}(P)}$}
\ENSURE{A PI representation $R$ of $G$}

\medskip

\STATE{Let $U=\{u_{1},u_{2},\ldots ,u_{n}\}$, ${V=\{v_{1},v_{2},\ldots ,v_{n}\}}$} \label{alg-PI-1}

\STATE{$Q_{1}\leftarrow \emptyset$; $Q_{2}\leftarrow \emptyset$; $P_{2}\leftarrow \emptyset$} \label{alg-PI-2}

\vspace{0.1cm}

\FOR[construction of the partial orders $Q_{1},Q_{2},P_{2}$]{$i = 1,2,\ldots,n$} \label{alg-PI-3}
     \FOR{$j = 1,2,\ldots,n$} \label{alg-PI-4}
          
          \IF[$i \neq j$]{$u_{i}v_{j} \notin E_{2}$} \label{alg-PI-5a}
               \STATE{$u_{i} <_{P_{2}} u_{j}$} \label{alg-PI-5b}
          \ENDIF
          
          \vspace{0.1cm}
          
          \IF[$i \neq j$]{$u_{i}v_{j}\in E_{1}$} \label{alg-PI-6a}
               \STATE{$u_{i}<_{Q_{1}}u_{j}$} \label{alg-PI-6b}
          \ENDIF
          
          \vspace{0.1cm}
          
          \IF[$i \neq j$]{$u_{j} <_{P} v_{i}$} \label{alg-PI-7a}
               \STATE{$u_{i} <_{Q_{2}} u_{j}$} \label{alg-PI-7b}
          \ENDIF
          
     \ENDFOR
\ENDFOR

\vspace{0.1cm}

\STATE{Compute a linear extension $Q_{0}$ of $Q_{1} \cup Q_{2}$} \label{alg-PI-8}

\STATE{$P_{1} \leftarrow \overline{Q_{0}}$} \label{alg-PI-9}

\vspace{0.1cm}

\STATE{Place the elements of $U$ on a line $L_1$ according to the linear order $P_1$} \label{alg-PI-10}

\STATE{Place a set of $n$ intervals on a line $L_2$ (parallel to $L_1$) according to the interval order $P_{2}$} \label{alg-PI-11}

\STATE{Build the PI representation $R$ of $G$ by connecting the endpoints of the intervals on $L_2$ with the 
corresponding points on $L_1$} \label{alg-PI-12}

\RETURN{$R$} \label{alg-PI-13}
\end{algorithmic}
\end{algorithm}

\begin{theorem}
\label{PI-repr-constr-thm}Let $G$ be a cocomparability graph with $n$
vertices and $P$ be the partial order of $\overline{G}$. Let $%
\{E_{1},E_{2}\} $ be a linear-interval cover of $\widehat{C}(P)$. Then
Algorithm~\ref{alg-PI-repr-constr} constructs in $O(n^{2})$ time
a~PI~representation~$R$ of~$G$.
\end{theorem}

\begin{proof}
Since $\widehat{C}(P)$ admits a linear-interval cover $\{E_{1},E_{2}\}$,
Corollary~\ref{linear-interval-col} implies that~$G$ is a~PI~graph.
Furthermore, it follows by the proof of the implication ((c) $\Rightarrow $
(b)) in Theorem~\ref{linear-interval-coverable-char-thm} that the partial
orders~$P_{1}$ and~$P_{2}$ that are constructed in lines~\ref{alg-PI-3}-\ref%
{alg-PI-9} of Algorithm~\ref{alg-PI-repr-constr} are a linear order and an
interval order, respectively, such that $\widehat{C}(P)=\widehat{NC}%
(P_{1})\cup \widehat{C}(P_{2})$. Furthermore, it follows by the proof of the
implication ((b) $\Rightarrow $ (a)) in Theorem~\ref%
{linear-interval-coverable-char-thm} that $P=P_{1}\cap P_{2}$ for these two
partial orders. Once we have computed in lines~\ref{alg-PI-3}-\ref{alg-PI-9}
the linear order $P_{1}$ and the interval order $P_{2}$, for which $%
P=P_{1}\cap P_{2}$, it is now straightforward to construct a PI
representation~$R$ of~$G$ as follows (cf.~also~\cite{Cerioli08} and the
proof of Theorem~\ref{PI-char-thm}). We arrange a set of $n$ points (resp.~$%
n $ intervals) on a line~$L_{1}$ (resp.~on a line~$L_{2}$, parallel to $%
L_{1} $) according to the linear order~$P_{1}$ (resp.~to the interval order~$%
P_{2}$). Then we connect the endpoints of the intervals on $L_{2}$ with the
corresponding points on $L_{1}$. Regarding the time complexity, each of the
lines~\ref{alg-PI-5a}-\ref{alg-PI-7b} of Algorithm~\ref{alg-PI-repr-constr}
can be executed in constant time, and thus the lines~\ref{alg-PI-3}-\ref%
{alg-PI-7b} can be executed in total~$O(n^{2})$ time. Furthermore, since the
lines~\ref{alg-PI-8}-\ref{alg-PI-12} can be executed in a trivial way in at
most~$O(n^{2})$ time each, it follows that the running time of Algorithm~\ref%
{alg-PI-repr-constr} is~$O(n^{2})$.\qed
\end{proof}

\section{Detecting linear-interval covers using Boolean satisfiability\label{linear-interval-satisfiability-sec}}

The natural algorithmic question that arizes from the characterization of PI
graphs using linear-interval covers in Corollary~\ref{linear-interval-col},
is the following: ``Given a cocomparability graph $G$ and a
partial order $P$ of $\overline{G}$, can we efficiently decide whether the
bipartite graph $\widehat{C}(P)$ has a linear-interval
cover?'' We will answer this algorithmic question in the
affirmative in Section~\ref{recognition-sec}. In this section we translate 
\emph{every} instance of this decision problem (i.e.~whether the bipartite
graph $\widehat{C}(P)$ has a linear-interval cover) to a restricted instance
of $3$SAT (cf.~Theorem~\ref{linear-interval-cover-by-satisfiability-thm}).
That is, for every such a bipartite graph~$\widehat{C}(P)$, we construct a
Boolean formula $\phi $ in conjunctive normal form (CNF), with size
polynomial on the size of~$\widehat{C}(P)$ (and thus also on $G$), such that~%
$\widehat{C}(P)$ has a linear-interval cover if and only if $\phi $ is
satisfiable. In particular, this formula $\phi $ can be written as $\phi
=\phi _{1}\wedge \phi _{2}$, where $\phi _{1}$ has three literals in every
clause and $\phi _{2}$ has two literals in every clause. Moreover, as we
will prove in Section~\ref{recognition-sec}, the satisfiability problem can
be efficiently decided on the formula $\phi $, by exploiting an appropriate
sub-formula of $\phi $ which is gradually mixed (cf.~Definition~\ref%
{gradually-mixed-def}).%

In the remainder of the paper, given a cocomparability graph $G$ and a
partial ordering $P$ of its complement $\overline{G}$, we denote by $%
\widetilde{G}=\widehat{C}(P)$ the bipartite complement of the domination
bipartite graph $C(P)$ of~$P$. Furthermore we denote by $H$ the associated
split graph of $\widetilde{G}$ and by $H^{\ast }$ the conflict graph of~$H$.
Moreover, we assume in the remainder of the paper without loss of generality
that $\chi (H^{\ast })\leq 2$, i.e.~that $H^{\ast }$ is bipartite. Indeed,
as we formally prove in Lemma~\ref{conflict-bilocor-necessity-lem}, if $\chi
(H^{\ast })>2$ then $\widetilde{G}$ does not have a linear-interval cover,
i.e.~$G$ is not a PI graph. Note that every proper $2$-coloring of the
vertices of the conflict graph $H^{\ast }$ corresponds to exactly one $2$%
-coloring of the edges of $H$ that includes no monochromatic $AC_{4}$. We
assume in the following that a proper $2$-coloring (with colors blue and
red) of the vertices of $H^{\ast }$ is given as input; note that $\chi _{0}$
can be computed in polynomial time.

\begin{lemma}
\label{conflict-bilocor-necessity-lem}Let $G$ be a cocomparability graph and 
$P$ be a partial order of $\overline{G}$. Let~$\widetilde{G}=\widehat{C}(P)$%
, $H$~be~the associated split graph of $\widetilde{G}$, and $H^{\ast }$ be
the conflict graph of $H$. If ${\widetilde{G}}$ is linear-interval
coverable, then ${\chi (H^{\ast })\leq 2}$.
\end{lemma}

\begin{proof}
Suppose otherwise that $\chi (H^{\ast })>2$. Then $t(H)>2$, since $\chi
(H^{\ast })\leq t(H)$ by Lemma~\ref%
{chromatic-number-conflict-lower-bound-lem}. Therefore, Lemma~\ref%
{chain-threshold-number-lem} implies that $ch({\widetilde{G}})>2$, and thus $%
G$ is not a trapezoid graph~\cite{MaSpinrad94}. Therefore $G$ is clearly not
a PI graph, and thus ${\widetilde{G}}$ is not linear-interval coverable by
Corollary~\ref{linear-interval-col}, which is a contradiction to the
assumption of the lemma. Therefore~$\chi (H^{\ast })\leq 2$.\qed
\end{proof}

Let $C_{1},C_{2},\ldots ,C_{k}$ be the connected components of $H^{\ast }$.
Some of these components of $H^{\ast }$ may be isolated vertices, which
correspond to uncommitted edges in $H$. We assign to every component $C_{i}$%
, where $1\leq i\leq k$, the Boolean variable $x_{i}$. Since $H^{\ast }$ is
bipartite by assumption, the vertices of each connected component $C_{i}$ of 
$H^{\ast }$ can be partitioned into two color classes $S_{i,1}$ and $S_{i,2}$%
. Without loss of generality, we assume that $S_{i,1}$ (resp.~$S_{i,2}$)
contains the vertices of $C_{i}$ that are colored red (resp.~blue) in $\chi
_{0}$. Note that, since vertices of $H^{\ast }$ correspond to edges of $H$
(cf.~Definition~\ref{conflict-graph-def}), for every two edges $e$ and $%
e^{\prime }$ of $H$ that are in conflict (i.e.~$e||e^{\prime }$) there
exists an index $i\in \{1,2,\ldots ,k\}$ such that one of these edges
belongs to $S_{i,1}$ and the other belongs to $S_{i,2}$. We now assign a 
\emph{literal} $\ell _{e}$ to~every~edge~$e$~of~$H$ as follows: if $e\in
S_{i,1}$ for some $i\in \{1,2,\ldots ,k\}$, then $\ell _{e}=x_{i}$;
otherwise, if $e\in S_{i,2}$, then~$\ell _{e}=\overline{x_{i}}$. Note that,
by construction, whenever two edges are in conflict in $H$, their assigned
literals are one the negation of the other.

\begin{observation}
\label{truth-assignment-coloring-obs}Every truth assignment $\tau $ of the
variables $x_{1},x_{2},\ldots ,x_{k}$ corresponds bijectively to a proper $2$%
-coloring $\chi _{\tau }$ (with colors blue and red) of the vertices of $%
H^{\ast }$, as follows: $x_{i}=0$ in $\tau $ (resp.~$x_{i}=1$ in $\tau $),
if and only if all vertices of the component $C_{i}$ have in $\chi _{\tau }$
the same color as in $\chi _{0}$ (resp.~opposite color than in $\chi _{0}$).
In particular, $\tau =(0,0,\ldots ,0)$ corresponds to the coloring $\chi
_{0} $.
\end{observation}

We now present the construction of the Boolean formulas~$\phi _{1}$ and~$%
\phi _{2}$ from the graphs~$H$ and~$H^{\ast }$, cf.~Algorithms~\ref%
{alg-phi_1-constr} and~\ref{alg-phi_2-constr}, respectively.

\medskip

\textbf{Description of the }$3$\textbf{-CNF formula }$\phi _{1}$\textbf{:}
Consider an $AC_{6}$ in the split graph $H$, and let~$e,e^{\prime
},e^{\prime \prime }$ be its three edges in $H$, such that no two literals
among $\{\ell _{e},\ell _{e^{\prime }},\ell _{e^{\prime \prime }}\}$ are one
the negation of the other. According to Algorithm~\ref{alg-phi_1-constr},
the Boolean formula $\phi _{1}$ has for this triple~$\{e,e^{\prime
},e^{\prime \prime }\}$ of edges exactly the two clauses $\alpha =(\ell
_{e}\vee \ell _{e^{\prime }}\vee \ell _{e^{\prime \prime }})$ and $\alpha
^{\prime }=(\overline{\ell _{e}}\vee \overline{\ell _{e^{\prime }}}\vee 
\overline{\ell _{e^{\prime \prime }}})$. It is easy to check by the
assignment of literals to edges that the clause $\alpha $ (resp.~the clause $%
\alpha ^{\prime }$) of $\phi _{1}$ is false in a truth assignment~$\tau $ of
the variables if and only if all edges $\{e,e^{\prime },e^{\prime \prime }\}$
are colored red (resp.~blue) in the $2$-edge-coloring $\chi _{\tau }$ of $H$
(cf.~Observation~\ref{truth-assignment-coloring-obs}), as the following
observation states.

\begin{observation}
\label{monochromatic-coloring-assignment-obs}Let $\tau $ be any truth
assignment of the variables $x_{1},x_{2},\ldots ,x_{k}$. Let $%
\{e_{1},e\,_{2},e_{3}\}$ be the edges of an $AC_{6}$ in $H$ and let $\alpha
=(\ell _{e_{1}}\vee \ell _{e_{2}}\vee \ell _{e_{3}})$ and $\alpha ^{\prime
}=(\overline{\ell _{e_{1}}}\vee \overline{\ell _{e_{2}}}\vee \overline{\ell
_{e_{3}}})$ be a the corresponding clauses in $\phi _{1}$. This $AC_{6}$ is
monochromatic in the coloring $\chi _{\tau }$ if and only if ${\alpha =0}$
or ${\alpha ^{\prime }=0}$~in~$\tau $.
\end{observation}

\begin{algorithm}[t!]
\caption{Construction of the $3$-CNF Boolean formula $\phi_{1}$} \label{alg-phi_1-constr}
\begin{algorithmic}[1]
\REQUIRE{The bipartite graph ${\widetilde{G} = \widehat{C}(P)}$, the associated split graph~$H$~of~$\widetilde{G}$, 
its conflict graph~$H^{\ast}$, and a~proper $2$-coloring $\chi_{0}$ of the vertices of $H^{\ast}$}
\ENSURE{The $3$-CNF Boolean formula $\phi_{1}$}

\medskip

\STATE{$\phi_{1} \leftarrow \emptyset$} \label{algorithm-phi_1-1}

\vspace{0.1cm}

\FOR[note that this is an $AC_{6}$ in the graph $H$ itself and not in a color subclass of its edges]{all triples of edges $\{e,e^{\prime},e^{\prime\prime}\}\subseteq E(H)$, 
     such that $\{e,e^{\prime},e^{\prime\prime}\}$ build an $AC_{6}$ in $E(H)$} \label{algorithm-phi_1-2}

\vspace{0.1cm}

     \IF{$\ell_{e} \neq \overline{\ell_{e^{\prime}}}$, $\ell_{e^{\prime}} \neq \overline{\ell_{e^{\prime\prime}}}$, and $\ell_{e} \neq \overline{\ell_{e^{\prime\prime}}}$} \label{algorithm-phi_1-3}
          
          \vspace{0.1cm}
          
          \IF{$\phi_{1}$ does not contain $(\ell_{e} \vee \ell_{e^{\prime}} \vee \ell_{e^{\prime\prime}})$ 
              and $(\overline{\ell_{e}} \vee \overline{\ell_{e^{\prime}}} \vee \overline{\ell_{e^{\prime\prime}}})$} \label{algorithm-phi_1-4}
               \STATE{$\phi_{1} \leftarrow \phi_{1} \wedge (\ell_{e} \vee \ell_{e^{\prime}} \vee \ell_{e^{\prime\prime}}) 
                      \wedge (\overline{\ell_{e}} \vee \overline{\ell_{e^{\prime}}} \vee \overline{\ell_{e^{\prime\prime}}})$} \label{algorithm-phi_1-5}
          \ENDIF
     \ENDIF
\ENDFOR

\vspace{0.1cm}
\RETURN{$\phi_{1}$} \label{algorithm-phi_1-6}
\end{algorithmic}
\end{algorithm}

Consider now another $AC_{6}$ of $H$ on the edges $\{e_{1},e_{2},e_{3}\}$,
in which at least one literal among $\{\ell _{e_{1}},\ell _{e_{2}},\ell
_{e_{3}}\}$ is the negation of another literal, for example $\ell _{e_{1}}=%
\overline{\ell _{e_{2}}}$. Then, for \emph{any} proper $2$-coloring of the
vertices of $H^{\ast }$, the edges $e$ and $e^{\prime }$ of $H$ receive
different colors, and thus this $AC_{6}$ is not monochromatic. Thus the next
observation follows by Observation~\ref%
{monochromatic-coloring-assignment-obs}.

\begin{observation}
\label{phi-1-no-monochromatic-obs}The formula $\phi _{1}$ is satisfied by a
truth assignment $\tau $ if and only if the corresponding $2$-coloring $\chi
_{\tau }$ of the edges of $H$ does not contain any monochromatic~$AC_{6}$.
\end{observation}

\textbf{Description of the }$2$\textbf{-CNF formula }$\phi _{2}$\textbf{:}
Denote for simplicity $H=(U,V,E_{H})$, where $U=\{u_{1},u_{2},\ldots
,u_{n}\} $ and $V=\{v_{1},v_{2},\ldots ,v_{n}\}$. Furthermore denote $%
E_{0}=\{u_{i}v_{i}~|~1\leq i\leq n\}$. Let $E^{\prime }=E_{H}\setminus E_{0}$
and $H^{\prime }=H-E_{0}$, i.e.~$H^{\prime }$ is the split graph that we
obtain if we remove from $H$ all edges of $E_{0}$. Consider now a pair of
edges $e=u_{i}v_{t}$ and $e^{\prime }=u_{t}v_{j}$ of $E^{\prime }$, such
that $u_{i}v_{j}\notin E^{\prime }$. Note that $i$ and $j$ may be equal.
However, since $E^{\prime }\cap E_{0}=\emptyset $, it follows that $i\neq t$
and $t\neq j$. Moreover, since the edge $u_{t}v_{t}$ belongs to $E_{H}$ but
not to $E^{\prime }$, it follows that the edges $e$ and $e^{\prime }$ are in
conflict in $H^{\prime }$ but not in $H$ (for both cases where $i=j$ and $%
i\neq j$). That is, although $e$ and $e^{\prime }$ are two non-adjacent
vertices in the conflict graph $H^{\ast }$ of $H$, they are adjacent
vertices in the conflict graph of $H^{\prime }$. For both cases where $i=j$
and $i\neq j$, an example of such a pair of edges $\{e,e^{\prime }\}$ is
illustrated in Figure~\ref{phi-2-fig}. According to Algorithm~\ref%
{alg-phi_2-constr}, for every such pair~$\{e,e^{\prime }\}$ of edges in $H$,
the Boolean formula $\phi _{2}$ has the clause $(\ell _{e}\vee \ell
_{e^{\prime }})$. It is easy to check by the assignment of literals to edges
of $H$ that this clause $(\ell _{e}\vee \ell _{e^{\prime }})$ of $\phi _{2}$
is false in the truth assignment $\tau $ if and only if both $e$ and $%
e^{\prime }$ are colored \emph{red} in the $2$-edge coloring $\chi _{\tau }$
of $H$. 

\begin{algorithm}[t!]
\caption{Construction of the $2$-CNF Boolean formula $\phi_{2}$} \label{alg-phi_2-constr}
\begin{algorithmic}[1]
\REQUIRE{The bipartite graph ${\widetilde{G} = \widehat{C}(P)}$, the associated split graph~$H$~of~$\widetilde{G}$, 
its conflict graph~$H^{\ast}$, and a~proper $2$-coloring $\chi_{0}$ of the vertices of $H^{\ast}$}
\ENSURE{The $2$-CNF Boolean formula $\phi_{2}$}

\medskip

\STATE{Let $H = (U,V,E_{H})$, where $U=\{u_{1}, u_{2}, \ldots, u_{n}\}$ and $V=\{v_{1}, v_{2}, \ldots, v_{n}\}$} \label{algorithm-phi_2-1}
\STATE{$ E_{0} \leftarrow \{u_{i}v_{i}\ |\ 1\leq i\leq n\}$; \ \ $E^{\prime} \leftarrow E_{H} \setminus E_{0}$; \ \ $H^{\prime} \leftarrow H - E_{0}$} \label{algorithm-phi_2-2}

\vspace{0.1cm}

\STATE{$\phi_{2} \leftarrow \emptyset$} \label{algorithm-phi_2-3}

\vspace{0.1cm}

\FOR{every pair $\{i,j\}\subseteq \{1,2,\ldots,n\}$ with $u_{i}v_{j}\notin E^{\prime}$} \label{algorithm-phi_2-4}
     \FOR{$t=1,2,\ldots,n$} \label{algorithm-phi_2-5}
          \IF[the edges $u_{i}v_{t},u_{t}v_{j}$ are in conflict in $H^{\prime}$ but not in $H$]{$u_{i}v_{t},u_{t}v_{j}\in E^{\prime}$} \label{algorithm-phi_2-6}
               
               \vspace{0.1cm}
                      
               \STATE{$e \leftarrow u_{i}v_{t}$; \ $e^{\prime} \leftarrow u_{t}v_{j}$; \ $\phi_{2} \leftarrow \phi_{2} \wedge (\ell_{e} \vee \ell_{e^{\prime}})$} \label{algorithm-phi_2-7}
          \ENDIF
     \ENDFOR
\ENDFOR
\vspace{0.1cm}
\RETURN{$\phi_{2}$} \label{algorithm-phi_2-8}
\end{algorithmic}
\end{algorithm}

\begin{figure}[tbh]
\centering 
\subfigure[] {\label{phi-2-fig-1} 
\includegraphics[scale=0.65]{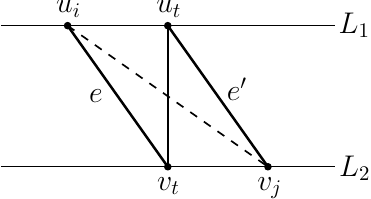}} \hspace{1cm} 
\subfigure[] {\label{phi-2-fig-2} 
\includegraphics[scale=0.65]{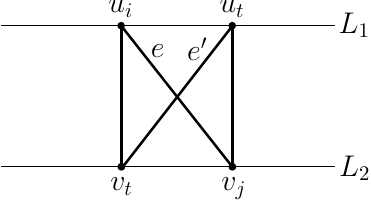}}
\caption{Two edges $e=u_{i}v_{t}$ and $e^{\prime }=u_{t}v_{j}$ of $H$, for
which the formula $\protect\phi _{2}$ has the clause $(\ell _{e}\vee \ell
_{e^{\prime }})$, in the case where (a) ${i\neq j}$ and (b)~$i=j$.}
\label{phi-2-fig}
\end{figure}

Now we provide the main result of this section in Theorem~\ref%
{linear-interval-cover-by-satisfiability-thm}, which relates the existence
of a linear-interval cover in $\widetilde{G}=\widehat{C}(P)$ with the
Boolean satisfiability of the formula $\phi _{1}\wedge \phi _{2}$. Before we
present Theorem~\ref{linear-interval-cover-by-satisfiability-thm}, we first
provide two auxiliary lemmas.

\begin{lemma}
\label{ei-isolated-vertices-lem}Let $G$ be a cocomparability graph and $P$
be a partial order of $\overline{G}$. Let~${\widetilde{G}=}${$\widehat{C}$}${%
(P)}$, $H$~be~the associated split graph of $\widetilde{G}$, and $H^{\ast }$
be the conflict graph of $H$. Denote~${\widetilde{G}=(U,V,}\widetilde{{E}}{)}
$ and ${E_{0}=\{u_{i}v_{i}\ |\ 1\leq i\leq n\}}$. Then, every $e\in E_{0}$
is an isolated vertex of $H^{\ast }$.
\end{lemma}

\begin{proof}
Note by Definition~\ref{bipartite-split-def} that $H=(U\cup V,E_{H})$, where 
$E_{H}=\widetilde{{E}}\cup (V\times V)$. Furthermore all edges of $V\times V$
in $E_{H}$ correspond to isolated vertices in the conflict graph $H^{\ast }$
of $H$ by Observations~\ref{uncommitted-isolated-obs} and~\ref%
{split-graph-clique-no-conflict-obs}. Therefore all non-isolated vertices in 
$H^{\ast }$ correspond to edges of $\widetilde{G}$ (i.e.~they do not belong
to $V\times V$). Consider now an edge $e_{i}=u_{i}v_{i}\in E_{0}\subseteq 
\widetilde{{E}}$, where $1\leq i\leq n$. Suppose that $e_{i}$ is not an
isolated vertex in the conflict graph $H^{\ast }$. Then the edge $e_{i}$ of~$%
\widetilde{G}$ builds with another edge $e=u_{j}v_{k}$ an induced $AC_{4}$
in $H$, i.e.~$e_{i}=u_{i}v_{i}$ and $e=u_{j}v_{k}$ induce a $2K_{2}$ in~$%
\widetilde{G}$. Therefore $u_{j}v_{i},u_{i}v_{k}\notin \widetilde{{E}}$,
i.e.~$u_{j}v_{i},u_{i}v_{k}\in E(C(P))$. Thus $u_{j}<_{P}u_{i}$ and~$%
u_{i}<_{P}u_{k}$ by Definition~\ref{C(P)-def}. Therefore, since $P$ is
transitive (as a partial order), it follows that $u_{j}<_{P}u_{k}$, and thus 
$u_{j}v_{k}\in E(C(P))$, i.e.~$u_{j}v_{k}\notin \widetilde{{E}}$. This is a
contradiction, since we assumed that $e=u_{j}v_{k}$ is an edge of ${%
\widetilde{G}}$, i.e.~$u_{j}v_{k}\in \widetilde{{E}}$. Therefore, $%
e_{i}=u_{i}v_{i}$ is an isolated vertex of $H^{\ast}$.\qed
\end{proof}

\begin{lemma}
\label{AC6-if-monochromatic-all-not-isolated-lem}Let $H$ be a split graph
and $H^{\ast }$ be the conflict graph of $H$, where $H^{\ast }$ is bipartite
with color classes $E_{1}$ and $E_{2}$. Let the vertices $v_{1},\ldots
,v_{6} $ of $H$ build an $AC_{6}$ on the edges~of~$E_{i}$, where $i\in
\{1,2\}$. Then the edges $v_{3}v_{6},v_{4}v_{1},v_{5}v_{2}$ exist in $H$ and 
$v_{4}v_{5}||v_{3}v_{6}$, $v_{2}v_{3}||v_{4}v_{1}$, and $%
v_{6}v_{1}||v_{5}v_{2}$.
\end{lemma}

\begin{proof}
Since $H$ is a split graph, Lemma~\ref{split-no-AP5-double-AP6-lem} implies
that $H$ does not contain any $AP_{5}$ or any double~$AP_{6}$. Therefore,
the $AC_{6}$ of $H$ is an $AP_{6}$, i.e.~an alternating path of length $6$,
cf.~Figure~\ref{AC6-2-fig}. Since $E_{1}$ and $E_{2}$ are the two color
classes of $H^{\ast }$, any two vertices $e$ and $e^{\prime }$ of $H^{\ast }$
in the set $E_{i}$, where $i\in \{1,2\}$, are not adjacent in $H^{\ast }$.
Equivalently, any two edges $e$ and $e^{\prime }$ of $H$ in the set $E_{i}$
are not in conflict, where $i\in \{1,2\}$. Therefore, since by assumption,
all edges $\{v_{2}v_{3},v_{4}v_{5},v_{6}v_{1}\}$ of this $AC_{6}$ belong to
the same color class $E_{i}$ for some $i\in \{1,2\}$, it follows that no
pair of these edges is in conflict in $H$. Thus Lemma~\ref%
{AC6-if-not-in-conflict-all-not-isolated-lem} implies that the edges $%
v_{3}v_{6},v_{4}v_{1},v_{5}v_{2}$ exist in $H$ and that $%
v_{4}v_{5}||v_{3}v_{6}$, $v_{2}v_{3}||v_{4}v_{1}$, and $%
v_{6}v_{1}||v_{5}v_{2}$.\qed
\end{proof}

We are now ready to provide Theorem~\ref%
{linear-interval-cover-by-satisfiability-thm}.

\begin{theorem}
\label{linear-interval-cover-by-satisfiability-thm}%
${\widetilde{G} = \widehat{C}(P)}$ is linear-interval colorable if and only
if ${\phi _{1}\wedge \phi _{2}}$ is satisfiable. Given a satisfying
assignment $\tau $ of ${\phi _{1}\wedge \phi _{2}}$, Algorithm~\ref%
{linear-interval-cover-from-assignment-alg} computes a linear-interval cover
of~${\widetilde{G}}$~in~$O(n^{2})$~time.
\end{theorem}

\begin{proof}
Denote ${\widetilde{G}=(U,V,}\widetilde{{E}}{)}$, where ${%
U=\{u_{1},u_{2},\ldots ,u_{n}\}}$ and ${V=\{v_{1},v_{2},\ldots ,v_{n}\}}$.
Furthermore denote $H=(U,V,E_{H})$, where $E_{H}=\widetilde{{E}}\cup
(V\times V)$, cf.~Definition~\ref{bipartite-split-def}. Let $%
E_{0}=\{u_{i}v_{i}\ |\ 1\leq i\leq n\}$. Since ${\widetilde{G}=}${$\widehat{C%
}$}${(P)}$, note by Definition~\ref{C(P)-def} that $E_{0}\subseteq 
\widetilde{{E}}\subseteq E_{H}$. Let $\chi _{0}$ be the $2$-coloring of the
vertices of $H^{\ast }$ (i.e.~the edges of $H$) that is given as input to
Algorithms~\ref{alg-phi_1-constr} and~\ref{alg-phi_2-constr}. Moreover, let $%
C_{1},C_{2},\ldots ,C_{k}$ be the connected components of $H^{\ast }$.

\medskip

($\Rightarrow $) Suppose that ${\widetilde{G}}$ is linear-interval
colorable. That is, there exist by Definition~\ref%
{linear-interval-coverable-def} two chain graphs $G_{1}=(U,V,E_{1})$ and $%
G_{2}=(U,V,E_{2})$, such that $\widetilde{G}=G_{1}\cup G_{2}$ and $%
E_{0}\subseteq E_{2}\setminus E_{1}$. Let $H_{1}=(U,V,E_{H_{1}})$ and $%
H_{2}=(U,V,E_{H_{2}})$ be the associated split graphs of $G_{1}$ and $G_{2}$%
, respectively. Note that $H=H_{1}\cup H_{2}$ and $E_{0}\subseteq
E_{H_{2}}\setminus E_{H_{1}}$. Since $G_{1}$ and $G_{2}$ are chain graphs,
i.e.~$ch(G_{1})=ch(G_{2})=1$, Lemma~\ref{chain-threshold-number-lem} implies
that $t(H_{1})=t(H_{2})=1$, i.e.~$H_{1}$ and $H_{2}$ are threshold graphs.
Therefore, neither $H_{1}$ nor $H_{2}$ includes an $AC_{4}$.

Recall that the formulas $\phi _{1}$ and $\phi _{2}$ have one Boolean
variable $x_{i}$ for every connected component $C_{i}$ of $H^{\ast }$, $%
i=1,2,\ldots ,k$. We construct a $2$-coloring $\chi _{H}$ of the edges of $H$
as follows. For every edge $e$ of $H$ (i.e.~a vertex of $H^{\ast }$), if $%
e\in E_{H_{1}}$ then we color $e$ red in $\chi _{H}$; otherwise, if $e\in
E_{H_{2}}\setminus E_{H_{1}}$ then we color $e$ blue in $\chi _{H}$. Recall
that $E_{0}\subseteq E_{H_{2}}\setminus E_{H_{1}}$, and thus all edges of $%
E_{0}$ are colored blue in $\chi _{H}$. Since both $H_{1}$ and $H_{2}$ do
not include any $AC_{4}$, it follows by the definition of $\chi _{H}$ that
there exists no monochromatic $AC_{4}$ in $\chi _{H}$. Therefore, every two
edges $e$ and $e^{\prime }$ of $H$, which correspond to adjacent vertices in 
$H^{\ast }$, have different colors in $\chi _{H}$, and thus $\chi _{H}$
constitutes a proper $2$-coloring of the vertices of $H^{\ast }$. Therefore
the coloring $\chi _{H}$ of the edges of $H$ (i.e.~vertices of $H^{\ast }$)
defines a truth assignment $\tau $ of the variables $x_{1},x_{2},\ldots
,x_{k}$ as follows (cf.~Observation~\ref{truth-assignment-coloring-obs}).
For every connected component $C_{i}$ of $H^{\ast }$, where $1\leq i\leq k$,
we define $x_{i}=1$ (resp.~$x_{i}=0$) in $\tau $ if all vertices of $C_{i}$
have in $\chi _{H}$ different (resp.~the same) color as in $\chi _{0}$. We
will now prove that $\tau $ satisfies both formulas $\phi _{1}$ and $\phi
_{2}$.

\vspace{0.1cm}

\textbf{Satisfaction of the Boolean formula }$\phi _{1}$\textbf{.} Let $%
\alpha $ be a clause of $\phi _{1}$. Recall that $\alpha $ corresponds to
some triple $\{e,e^{\prime },e^{\prime \prime }\}$ of edges of $H$ that
builds an $AC_{6}$ in $H$ (cf.~lines~\ref{algorithm-phi_1-2}-\ref%
{algorithm-phi_1-5} of Algorithm~\ref{alg-phi_1-constr}). In particular,
either $\alpha =(\ell _{e}\vee \ell _{e^{\prime }}\vee \ell _{e^{\prime
\prime }})$ or $\alpha =(\overline{\ell _{e}}\vee \overline{\ell _{e^{\prime
}}}\vee \overline{\ell _{e^{\prime \prime }}})$, where $\ell _{e},\ell
_{e^{\prime }},\ell _{e^{\prime \prime }}$ are the literals that have been
assigned to the edges $e,e^{\prime },e^{\prime \prime }$, respectively.
Then, it follows from the description of the formula $\phi _{1}$ (cf.~also
Observation~\ref{monochromatic-coloring-assignment-obs}) that the clause $%
(\ell _{e}\vee \ell _{e^{\prime }}\vee \ell _{e^{\prime \prime }})$
(resp.~the clause $(\overline{\ell _{e}}\vee \overline{\ell _{e^{\prime }}}%
\vee \overline{\ell _{e^{\prime \prime }}})$) is not satisfied in the truth
assignment $\tau $ if and only if the edges $e,e^{\prime },e^{\prime \prime
} $ of $H$ are all red (resp.~all blue) in~$\chi _{H}$.

Let $\alpha =(\ell _{e}\vee \ell _{e^{\prime }}\vee \ell _{e^{\prime \prime
}})$ (resp.~$\alpha =(\overline{\ell _{e}}\vee \overline{\ell _{e^{\prime }}}%
\vee \overline{\ell _{e^{\prime \prime }}})$). Suppose that $\alpha $ is not
satisfied by $\tau $, and thus the edges $e,e^{\prime },e^{\prime \prime }$
of $H$ are all red (resp.~blue) in $\chi _{H}$. Therefore all edges $%
e,e^{\prime },e^{\prime \prime }$ belong to $E_{H_{1}}$ (resp.~to $%
E_{H_{2}}\setminus E_{H_{1}}$, and thus to $E_{H_{2}}$) by the definition of 
$\chi _{H}$. Thus $H$ has an $AC_{6}$ on the edges $e,e^{\prime },e^{\prime
\prime }$, which belong to $H_{1}$ (resp.~to $H_{2}$). Therefore $H_{1}$
(resp.~$H_{2}$) does not have a threshold completion in $H$ by Lemma~\ref%
{threshold-completion-lem}. This is a contradiction, since $H_{1}$ (resp.~$%
H_{2}$) is a threshold graph. Therefore the clause $\alpha =(\ell _{e}\vee
\ell _{e^{\prime }}\vee \ell _{e^{\prime \prime }})$ (resp.~$\alpha =(%
\overline{\ell _{e}}\vee \overline{\ell _{e^{\prime }}}\vee \overline{\ell
_{e^{\prime \prime }}})$) of $\phi _{1}$ is satisfied by the truth
assignment $\tau $, and thus $\tau $ satisfies $\phi _{1}$.

\vspace{0.1cm}

\textbf{Satisfaction of the Boolean formula }$\phi _{2}$\textbf{.} Let $%
\alpha =(\ell _{e}\vee \ell _{e^{\prime }})$ be a clause of $\phi _{2}$.
Recall that $\alpha $ corresponds to some pair of edges $e=u_{i}v_{t}$ and $%
e^{\prime }=u_{t}v_{j}$ of $E_{H}\setminus E_{0}$, such that $%
u_{i}v_{j}\notin E_{H}\setminus E_{0}$ (cf.~lines~\ref{algorithm-phi_2-4}-%
\ref{algorithm-phi_2-7} of Algorithm~\ref{alg-phi_2-constr}). Therefore,
since $u_{t}v_{t}\in E_{0}$, it follows that the edges $\{e,e^{\prime }\}$
build an $AC_{4}$ in $H-E_{0}$ but not in $H$. Suppose that the clause $%
\alpha =(\ell _{e}\vee \ell _{e^{\prime }})$ of $\phi _{2}$ is not satisfied
by the truth assignment $\tau $, i.e.~$\ell _{e}=\ell _{e^{\prime }}=0$ in $%
\tau $. Then, it follows from the description of the formula $\phi _{2}$
that both $e$ and $e^{\prime }$ are colored red in the $2$-edge coloring $%
\chi _{H}$ of $H$. Therefore both edges $e$ and $e^{\prime }$ belong to $%
H_{1}$ by the definition of $\chi _{H}$. However, as we noticed above, the
edges $\{e,e^{\prime }\}$ build an $AC_{4}$ in $H-E_{0}$, and thus they also
build an $AC_{4}$ in $H_{1}\subseteq H-E_{0}$. This is a contradiction by
Corollary~\ref{chain-number-2-char-cor}, since $H_{1}$ is a threshold graph.
Therefore the clause $\alpha =(\ell _{e}\vee \ell _{e^{\prime }})$ of $\phi
_{2}$ is satisfied by the truth assignment $\tau $, and thus $\tau $
satisfies $\phi _{2}$.

\medskip

($\Leftarrow $) Suppose that $\phi _{1}\wedge \phi _{2}$ is satisfiable, and
let $\tau $ be a satisfying truth assignment of $\phi _{1}\wedge \phi _{2}$.
Recall that the formulas $\phi _{1}$ and $\phi _{2}$ have one Boolean
variable $x_{i}$ for every connected component $C_{i}$ of $H^{\ast }$, $%
i=1,2,\ldots ,k$. First, given the truth assignment $\tau $, we construct
the $2$-coloring $\chi _{\tau }$ of the vertices of $H^{\ast }$ according to
Observation~\ref{truth-assignment-coloring-obs}. This $2$-coloring of the
vertices of $H^{\ast }$ defines also a corresponding $2$-coloring of the
edges of $H$. Since $\phi _{1}$ is satisfied by $\tau $, it follows by
Observation~\ref{phi-1-no-monochromatic-obs} that, in the coloring $\chi
_{\tau }$ of its edges, $H$ does not contain any monochromatic~$AC_{6}$.
Therefore Theorem~\ref{AC6-suffices-thm} implies that $H$ does not contain
any monochromatic~$AC_{2k}$ in $\chi_{\tau}$, where $k\geq 3$.

\medskip

\textbf{The vertex coloring }$\chi _{\tau }^{\prime }$ \textbf{of} $H^{\ast
} $\textbf{.} Now we modify the coloring $\chi _{\tau }$ to the coloring $%
\chi _{\tau }^{\prime }$, as follows. For every trivial connected component $%
C_{i} $ of $H^{\ast }$ (i.e.~when $C_{i}$ has exactly one vertex), we color
the vertex of $C_{i}$ blue in $\chi _{\tau }^{\prime }$, regardless of the
color of $C_{i}$ in $\chi _{\tau }$. On the other hand, for every
non-trivial connected component $C_{i}$ of $H^{\ast }$ (i.e.~when $C_{i}$
has at least two vertices), the vertices of $C_{i}$ have the same color in
both $\chi _{\tau }$ and $\chi _{\tau }^{\prime }$. This new $2$-coloring of
the vertices of $H^{\ast }$ defines also a corresponding $2$-coloring of the
edges of $H$. Note in particular by Lemma~\ref{ei-isolated-vertices-lem}
that all edges of $E_{0}$ are colored blue in $\chi _{\tau }^{\prime }$.
Denote by $E_{H_{1}}$ and $E_{H_{2}}$ the sets of red and blue edges of $H$
in $\chi _{\tau }^{\prime }$, respectively. Note that $E_{0}\subseteq
E_{H_{2}}$. Moreover note that $H$ does not have any $AC_{4}$ on the
vertices of $E_{H_{1}}$, or on the vertices of $E_{H_{2}}$, since $\chi
_{\tau }^{\prime }$ is a proper $2$-coloring of the vertices of $H^{\ast }$.
Define the subgraphs $H_{1}=(U,V,E_{H_{1}})$ and $H_{2}=(U,V,E_{H_{2}})$ of $%
H$. Note that $H=H_{1}\cup H_{2}$.

\medskip

$H_{2}$ \textbf{has a threshold completion in }$H$\textbf{.} Suppose now
that $H$ has an $AC_{2k}$ on the edges of $E_{H_{2}}$, for some $k\geq 3$.
Then Theorem~\ref{AC6-suffices-thm} implies that $H$ has also an $AC_{6}$ on
the edges of $E_{H_{2}}$, i.e.~$H$ has an $AC_{6}$, in which all three edges
are blue in $\chi _{\tau }^{\prime }$. Since $H$ does not have any
monochromatic $AC_{6}$ in $\chi _{\tau }$, it follows that for at least one
of the edges $e$ of the blue $AC_{6}$ of $H$ in $\chi _{\tau }^{\prime }$,
the color of $e$ is different in $\chi _{\tau }$ and in $\chi _{\tau
}^{\prime }$. Therefore, it follows by the construction of $\chi _{\tau
}^{\prime }$ from $\chi _{\tau }$ that the vertex of $H^{\ast }$ that
corresponds to $e$ is an isolated vertex in $H^{\ast }$. That is, the edge $%
e $ is uncommitted in $H$. This is a contradiction by Lemma~\ref%
{AC6-if-monochromatic-all-not-isolated-lem}, since $e$ has been assumed to
be an edge of a monochromatic $AC_{6}$ of $H$ in $\chi _{\tau }^{\prime }$.
Therefore $H$ does not have any $AC_{2k}$ on the edges of $E_{H_{2}}$, where 
$k\geq 3$. Thus, since $H$ does not have any $AC_{4}$ on the vertices of $%
E_{H_{2}}$, it follows that $H$ does not have any $AC_{2k}$ on the edges of $%
E_{H_{2}}$, where $k\geq 2$. Therefore $H_{2}$ has a threshold completion in 
$H$ by Lemma~\ref{threshold-completion-lem}.

\medskip

$H_{1}$ \textbf{has a threshold completion in }$H-E_{0}$\textbf{.} Denote
now $H^{\prime }=H-E_{0}$. We will prove that $H_{1}$ has a threshold
completion in $H^{\prime }$. To this end, it suffices to prove by Lemma~\ref%
{threshold-completion-lem} that $H^{\prime }$ does not have any $AC_{2k}$ on
the edges of $E_{H_{1}}$, where $k\geq 2$.

For the sake of contradiction, suppose that $H^{\prime }$ includes an $%
AC_{4} $ on the edges of $E_{H_{1}}$. That is, there exist two edges $%
e,e^{\prime }\in E_{H_{1}}$ that are in conflict in $H^{\prime }$. Note by
the definition of $E_{H_{1}}$ that the edges $e$ and $e^{\prime }$ are
colored red in $\chi _{\tau }^{\prime }$, and thus they are also colored red
in $\chi _{\tau }$. If the edges $\{e,e^{\prime }\}$ also build an $AC_{4}$
in $H $ (i.e.~before the removal of $E_{0}$), then the vertices $e$ and $%
e^{\prime }$ of $H^{\ast }$ are adjacent in $H^{\ast }$, and thus the edges $%
e$ and $e^{\prime }$ of $H$ have different colors in $\chi _{\tau }$, which
is a contradiction. Thus the edges $\{e,e^{\prime }\}$ are in conflict in $%
H^{\prime }$ but not in $H$. Recall now that for every such a pair $%
\{e,e^{\prime }\}$ of edges of $H^{\prime }$ there exists a clause $\alpha
=(\ell _{e}\vee \ell _{e^{\prime }})$ in the formula $\phi _{2}$ (cf.~lines~%
\ref{algorithm-phi_2-4}-\ref{algorithm-phi_2-7} of Algorithm~\ref%
{alg-phi_2-constr}). It follows from the description of the formula $\phi
_{2}$ that the clause $\alpha $ is not satisfied by the truth assignment $%
\tau $ if and only if both edges $e,e^{\prime }$ in $H$ are red in~$\chi
_{\tau }$. However, since $\tau $ is a satisfying assignment of $\phi _{2}$,
every clause of $\phi _{2}$ is satisfied by $\tau $. Therefore at least one
of the edges $e$ and $e^{\prime }$ is colored blue in $\chi _{\tau }$, which
is a contradiction. Therefore $H^{\prime }$ does not include any $AC_{4}$ on
the edges of $E_{H_{1}}$.

Suppose now that $H^{\prime }$ includes an $AC_{2k}$ on the edges of $%
E_{H_{1}}$, where $k\geq 3$. Consider the smallest such $AC_{2k}$ on the
edges of $E_{H_{1}}$, i.e.~an $AC_{2k}$ with the smallest $k\geq 3$. Let $%
w_{1},w_{2},\ldots ,w_{2k}$ be the vertices of $H^{\prime }$ that build this 
$AC_{2k}$. Note by the definition of $E_{H_{1}}$ that all edges of this $%
AC_{2k}$ are colored red in the coloring $\chi _{\tau }^{\prime }$, and thus
they are also colored red in the coloring $\chi _{\tau }$. However, as we
proved above, in the coloring $\chi _{\tau }$ of its edges, $H$ does not
contain any monochromatic~$AC_{2k}$, where $k\geq 3$. Therefore, at least
one of the non-edges of the $AC_{2k}$ in the graph $H^{\prime }$ is an edge
of $E_{0}$ in the graph $H$. Assume without loss of generality that this
edge of $E_{0}$ is $w_{1}w_{2}$. That is, assume that $w_{1}w_{2}\in E_{0}$,
i.e.~$w_{1}w_{2}=u_{i}v_{i}$ for some $i\in \{1,2,\ldots ,n\}$.

Suppose that $w_{3}w_{2k}$ is not an edge of $H^{\prime }$. Then, since $%
w_{1}w_{2}\in E_{0}$, there exists (similarly to above) a clause $\alpha $
in the formula $\phi _{2}$ such that $\alpha $ is not satisfied by the truth
assignment $\tau $ if and only if both edges $w_{2}w_{3}$ and $w_{2k}w_{1}$
are colored red in $\chi _{\tau }$. However, $\tau $ is a satisfying truth
assignment of $\phi _{2}$ by assumption, and thus at least one edge of $%
w_{2}w_{3}$ and $w_{2k}w_{1}$ is colored blue in $\chi _{\tau }$, which is a
contradiction. Therefore $w_{3}w_{2k}$ is an edge of $H^{\prime }$. Suppose
now that the edge $w_{3}w_{2k}$ of $H^{\prime }$ is colored red in $\chi
_{\tau }^{\prime }$, and thus $w_{3}w_{2k}\in E_{H_{1}}$ by the definition
of $E_{H_{1}}$. Then the vertices $w_{3},w_{4},\ldots ,w_{2k}$ build an $%
AC_{2k-2}$ in $H^{\prime }$ on the edges of $E_{H_{1}}$, which is a
contradiction to the minimality assumption of the $AC_{2k}$ in $H^{\prime }$%
. Therefore the edge $w_{3}w_{2k}$ of~$H^{\prime }$ is colored blue in $\chi
_{\tau }^{\prime }$, and thus $w_{3}w_{2k}\in E_{H_{2}}$.

Recall now that both the edges $w_{2}w_{3}$ and $w_{2k}w_{1}$ of $H^{\prime
} $ are red in $\chi _{\tau }^{\prime }$. Therefore, by the definition of
the coloring $\chi _{\tau }^{\prime }$ from $\chi _{\tau }$, it follows that
each of the edges $w_{2}w_{3}$ and $w_{2k}w_{1}$ participates to at least
one $AC_{4}$ in $H$ (or equivalently the corresponding vertices of $%
w_{2}w_{3}$ and $w_{2k}w_{1}$ in $H^{\ast }$ are not isolated vertices). Let
the edges $w_{2}w_{3}$ and $w_{2}^{\prime }w_{3}^{\prime }$ form an $AC_{4}$
in $H$, for some vertices $w_{2}^{\prime }$ and $w_{3}^{\prime }$, where $%
w_{2}w_{2}^{\prime }$ and $w_{3}w_{3}^{\prime }$ are not edges in $H$.
Similarly, let the edges $w_{2k}w_{1}$ and $w_{2k}^{\prime }w_{1}^{\prime }$
form an $AC_{4}$ in $H$, for some vertices $w_{2k}^{\prime }$ and $%
w_{1}^{\prime }$, where $w_{2k}w_{2k}^{\prime }$ and $w_{1}w_{1}^{\prime }$
are not edges in $H$. Note that some of the vertices $\{w_{2}^{\prime
},w_{3}^{\prime },w_{2k}^{\prime },w_{1}^{\prime }\}$ may coincide with each
other, as well as with some of the vertices $\{w_{2},w_{3},w_{2k},w_{1}\}$.
Recall that $\chi _{\tau }^{\prime }$ is a proper $2$-coloring of the
vertices of~$H^{\ast }$. Therefore, since $w_{2}w_{3}$ and $w_{2k}w_{1}$ are
colored red in $\chi _{\tau }^{\prime }$, it follows that $w_{2}^{\prime
}w_{3}^{\prime }$ and $w_{2k}^{\prime }w_{1}^{\prime }$ are colored blue in $%
\chi _{\tau }^{\prime }$. Therefore the vertices $w_{1},w_{2},w_{2}^{\prime
},w_{3}^{\prime },w_{3},w_{2k},w_{2k}^{\prime },w_{1}^{\prime }$ build an $%
AC_{8}$ in $H$ on the edges of $E_{H_{2}}$. This is a contradiction, since
we proved above that $H$ does not have any $AC_{2k}$ on the edges of $%
E_{H_{2}}$, where~$k\geq 2$.

Therefore, it follows that $H^{\prime }$ does not include any $AC_{2k}$ on
the edges of $E_{H_{1}}$, where~$k\geq 3$. Thus, since we already proved
that $H^{\prime }$ does not include any $AC_{4}$ on the edges of~$E_{H_{1}}$%
, it follows that $H^{\prime }$ does not include any $AC_{2k}$ on the edges
of $E_{H_{1}}$, where $k\geq 2$. Therefore~$H_{1}$ has a threshold
completion in $H^{\prime }=H-E_{0}$ by Lemma~\ref{threshold-completion-lem}.

\medskip

Summarizing, $H_{1}$ has a threshold completion in $H^{\prime }=H-E_{0}$,
and $H_{2}$ has a threshold completion in $H$. Furthermore all edges of $%
E_{0}$ belong to the graph $H$, and $H=H_{1}\cup H_{2}$. Let $\widetilde{H}%
_{1}$ be the threshold completion of $H_{1}$ in $H-E_{0}$, and let $%
\widetilde{H}_{2}$ be the threshold completion of $H_{2}$ in $H$. Then $%
\widetilde{H}_{1}$ and $\widetilde{H}_{2}$ are two threshold graphs,
i.e.~they do not include any $AC_{4}$. Furthermore, let $\widetilde{G}%
_{1}=(U,V,\widetilde{E}_{1})$ and $\widetilde{G}_{2}=(U,V,\widetilde{E}_{2})$
be the bipartite graphs obtained by $\widetilde{H}_{1}$ and $\widetilde{H}%
_{2}$, respectively, by removing from them all possible edges of $V\times V$%
. Note that $E_{0}\subseteq \widetilde{E}_{2}\setminus \widetilde{E}_{1}$,
since every edge of $E_{0}$ belongs to $\widetilde{H}_{2}$ and not to $%
\widetilde{H}_{1}$. Furthermore, neither $\widetilde{G}_{1}$ nor $\widetilde{%
G}_{2}$ include any induced $2K_{2}$, since $\widetilde{H}_{1}$ and $%
\widetilde{H}_{2}$ do not include any $AC_{4}$. Therefore both $\widetilde{G}%
_{1}$ and $\widetilde{G}_{2}$ are chain graphs. Moreover, since $H=H_{1}\cup
H_{2}$, it follows that also $H=\widetilde{H}_{1}\cup \widetilde{H}_{2}$ and 
$\widetilde{G}=\widetilde{G}_{1}\cup \widetilde{G}_{2}$. Thus, since $%
E_{0}\subseteq \widetilde{E}_{2}\setminus \widetilde{E}_{1}$, it follows
that $\widetilde{G}$ is linear-interval coverable by Definition~\ref%
{linear-interval-coverable-def} and $\{\widetilde{E}_{1},\widetilde{E}_{2}\}$
is a linear-interval cover of $\widetilde{G}$. This construction of $\{%
\widetilde{E}_{1},\widetilde{E}_{2}\}$ from the satisfying truth assignment $%
\tau $ of $\phi _{1}\wedge \phi _{2}$ is shown in Algorithm~\ref%
{linear-interval-cover-from-assignment-alg}.

\medskip

\begin{algorithm}[t!]
\caption{Construction of a linear-interval cover of ${\widetilde{G} = \widehat{C}(P)}$, if $\phi_1 \wedge \phi_2$ is satisfiable} \label{linear-interval-cover-from-assignment-alg}
\begin{algorithmic}[1]
\REQUIRE{The bipartite graph ${\widetilde{G} = \widehat{C}(P)}$, the associated split graph~$H$~of~$\widetilde{G}$, 
its conflict graph~$H^{\ast}$, a~proper $2$-coloring $\chi_{0}$ of the vertices of $H^{\ast}$, and a satisfying truth assignment $\tau$ of $\phi_1 \wedge \phi_2$}
\ENSURE{A linear-interval cover $\{\widetilde{E}_{1},\widetilde{E}_{2}\}$ of $\widetilde{G}$}

\medskip

\STATE{Let $H = (U,V,E_{H})$, where $U=\{u_{1}, u_{2}, \ldots, u_{n}\}$ and $V=\{v_{1}, v_{2}, \ldots, v_{n}\}$}  \label{alg-cover-1}
\STATE{$ E_{0} \leftarrow \{u_{i}v_{i}\ |\ 1\leq i\leq n\}$}  \label{alg-cover-2}

\vspace{0.1cm}

\FOR{every connected component $C_{i}, 1\leq i\leq k$, of $H^{\ast}$} \label{alg-cover-3}
     
     \IF{$C_i$ is an isolated vertex of $H^{\ast}$} \label{alg-cover-4}
          \STATE{color the vertex of $C_i$ blue} \label{alg-cover-5}
     \ELSE \label{alg-cover-6}

          \STATE{\textbf{if} $x_{i}=0$ in $\tau$ \textbf{then} color every vertex of $C_i$ with the same color as in $\chi_{0}$} \label{alg-cover-7}
          \STATE{\textbf{if} $x_{i}=1$ in $\tau$ \textbf{then} color every vertex of $C_i$ with the opposite color than in $\chi_{0}$} \label{alg-cover-8}
     \ENDIF 
\ENDFOR 

\vspace{0.1cm}

\STATE{$E_{H_{1}} \leftarrow \{e \in E_{H} \ | \ e \text{ is red}\}$; \ $H_{1} \leftarrow (U,V,E_{H_{1}})$} \label{alg-cover-9}

\STATE{$E_{H_{2}} \leftarrow \{e \in E_{H} \ | \ e \text{ is blue}\}$; \ $H_{2} \leftarrow (U,V,E_{H_{2}})$} \label{alg-cover-10}

\vspace{0.1cm}

\STATE{Compute a threshold completion $\widetilde{H}_{1}$ of $H_{1}$ in $H - E_{0}$ (by Lemma~\ref{threshold-completion-construction-lem})} \label{alg-cover-11}

\STATE{Compute a threshold completion $\widetilde{H}_{2}$ of $H_{2}$ in $H$ (by Lemma~\ref{threshold-completion-construction-lem})} \label{alg-cover-12}

\vspace{0.1cm}

\STATE{$\widetilde{E}_{1} \leftarrow E(\widetilde{H}_{1}) \setminus (V \times V)$; \ $\widetilde{E}_{2} \leftarrow E(\widetilde{H}_{2}) \setminus (V \times V)$} \label{alg-cover-13}

\vspace{0.1cm}
\RETURN{$\{\widetilde{E}_{1},\widetilde{E}_{2}\}$} \label{alg-cover-14}
\end{algorithmic}
\end{algorithm}

\textbf{Running time of Algorithm~\ref%
{linear-interval-cover-from-assignment-alg}.}~First note that, since $|U| =
|V| = n$, the split graph $H$ has $O(n^2)$ edges. Therefore, since each edge
of $H$ is processed exactly once in the execution of lines~\ref{alg-cover-3}-%
\ref{alg-cover-8} in Algorithm~\ref%
{linear-interval-cover-from-assignment-alg}, these lines are executed in $%
O(n^2)$ time in total. Similarly, each of the lines~\ref{alg-cover-9},~\ref%
{alg-cover-10}, and~\ref{alg-cover-13} is executed in $O(n^2)$ time. Now,
each of the lines~\ref{alg-cover-11} and~\ref{alg-cover-12} is executed by
Lemma~\ref{threshold-completion-construction-lem} in time linear to the size
of $H$, i.e.~in~$O(n^2)$ time each. Therefore the total running time of
Algorithm~\ref{linear-interval-cover-from-assignment-alg} is~$O(n^2)$. This
completes the proof of the theorem.\qed
\end{proof}

\section{The recognition of linear-interval orders and PI graphs\label{recognition-sec}}

In this section we investigate the structure of the formula $\phi _{1}\wedge
\phi _{2}$ that we computed in Section~\ref{linear-interval-satisfiability-sec}. In particular, we first prove in
Section~\ref{structural-properties-formula-subsec} some fundamental
structural properties of $\phi _{1}\wedge \phi _{2}$, which allow us to find
an appropriate sub-formula of $\phi _{1}\wedge \phi _{2}$ which is gradually
mixed (cf.~Definition~\ref{gradually-mixed-def}). Then we exploit this
sub-formula of $\phi _{1}\wedge \phi _{2}$ in order to provide in Section~%
\ref{recognition-algorithm-subsec} an algorithm that solves the
satisfiability problem on $\phi _{1}\wedge \phi _{2}$ in time linear to its size, cf.~Theorem~\ref{formula-equivalent-gm-formula-thm}. 
Finally, using this satisfiability algorithm, we combine our results
of Sections~\ref{linear-interval-sec} and~\ref%
{linear-interval-satisfiability-sec} in order to recognize efficiently PI
graphs and linear-interval orders in Section~\ref%
{recognition-algorithm-subsec}.%

\subsection{Structural properties of the formula $\protect\phi _{1}\wedge 
\protect\phi _{2}$\label{structural-properties-formula-subsec}}

The three main structural properties of $\phi _{1}\wedge \phi _{2}$ are proved
in Lemmas~\ref{disjoint-clauses-lem},~\ref{congruent-clauses-phi-2-lem-a},~and~\ref{congruent-clauses-phi-2-lem-b}, 
respectively. We first provide two auxiliary technical lemmas.

\begin{lemma}
\label{extension-AP6-same-base-wanted-ceiling-lem}Let $\alpha =(\ell
_{1}\vee \ell _{2}\vee \ell _{3})$ be a clause of $\phi _{1}$. Assume that $%
\alpha $ corresponds to the $AP_{6}$ of~$H$ on the vertices $v_{1},\ldots
,v_{6}$, which has the literals $\ell _{1},\ell _{2},\ell _{3}$ on its edges
(in this order). Then, for every edge $e$ of $H$ with literal $\ell
_{e}=\ell _{2}$, there exists an $AP_{6}$ in $H$ with $v_{1}v_{2}$ as is its
base and $e$ as its ceiling, which has the literals $\ell _{1},\ell
_{2},\ell _{3}$ on its edges (in this order).
\end{lemma}

\begin{proof}
First note that by the construction of $\phi _{1}$ (cf.~Section~\ref%
{linear-interval-satisfiability-sec}) no two literals among $\{\ell
_{1},\ell _{2},\ell _{3}\}$ are one the negation of the other, i.e.~$\ell
_{1}\neq \overline{\ell _{2}}$, $\ell _{1}\neq \overline{\ell _{3}}$, and $%
\ell _{2}\neq \overline{\ell _{3}}$. Therefore also no pair among the edges
of the $AP_{6}$ on the vertices $v_{1},\ldots ,v_{6}$ is in conflict.
Denote for simplicity $e^{\prime }=v_{4}v_{5}$. Since $\ell _{e^{\prime
}}=\ell _{e}=\ell _{2}$, the edges $e^{\prime }$ and $e$ of $H$ correspond
to two vertices of the conflict graph $H^{\ast }$ that lie in the same
connected component of $H^{\ast }$. Thus there exists a path between these
two vertices of $H^{\ast }$. That is, there exists a sequence of edges $%
e_{1},e_{2},\ldots ,e_{t}$ in $H$, where $e_{1}=e^{\prime }$ and $e_{t}=e$,
such that $e_{i}||e_{i+1}$ for every $i\in \{1,2,\ldots ,t-1\}$. Note that $%
\ell _{e_{i}}\in \{\ell _{2},\overline{\ell _{2}}\}$ for all these edges $%
e_{i}$. For every $1\leq i\leq t$ denote $e_{i}=u_{i}w_{i}$, where $%
u_{1}=v_{4}$ and $w_{1}=v_{5}$. Furthermore let $u_{i}u_{i+1}$ and $%
w_{i}w_{i+1}$ be the non-edges between $e_{i}$ and $e_{i+1}$, where $1\leq
i\leq t-1$. For simplicity of the presentation, denote $u_{0}=v_{3}$ and $%
w_{0}=v_{6}$.

We will prove by induction that for every $i\in \{1,2,\ldots ,t\}$ there
exists an $AP_{6}$ in $H$ on the vertices $%
v_{1},v_{2},u_{i-1},u_{i},w_{i},w_{i-1}$ (if $i$ is odd), or on the vertices 
$v_{1},v_{2},u_{i},u_{i-1},w_{i-1},w_{i}$ (if $i$ is even), which has the
literals $\ell _{1},\ell _{2},\ell _{3}$ on its edges (in this order). The
induction basis (i.e.~the case where $i=1$) follows immediately by the
assumption of the lemma.

For the induction step, let first $i\geq 2$ be even. Then $i-1$ is odd, and
thus there exists by the induction hypothesis an $AP_{6}$ in $H$ on the
vertices $v_{1},v_{2},u_{i-2},u_{i-1},w_{i-1},w_{i-2}$ which has the
literals $\ell _{1},\ell _{2},\ell _{3}$ on its edges (in this order). That
is, $\ell _{v_{2}u_{i-2}}=\ell _{1}$, $\ell _{u_{i-1}w_{i-1}}=\ell _{2}$,
and $\ell _{w_{i-2}v_{1}}=\ell _{3}$. Therefore, since $\ell
_{u_{i}w_{i}}\in \{\ell _{2},\overline{\ell _{2}}\}$ and $%
u_{i}w_{i}||u_{i-1}w_{i-1}$ by assumption, it follows that $\ell
_{u_{i}w_{i}}=\overline{\ell _{2}}$. Furthermore, since no pair among the
edges of the $AP_{6}$ is in conflict, Lemma~\ref%
{AC6-if-not-in-conflict-all-not-isolated-lem} implies in particular that the
edges $v_{1}u_{i-1}$ and $v_{2}w_{i-1}$ exist in $H$ and that $\ell
_{v_{1}u_{i-1}}=\overline{\ell _{1}}$ and~$\ell _{v_{2}w_{i-1}}=\overline{%
\ell _{3}}$.

\begin{myclaim}
\label{claim-1}$v_{1}\neq w_{i}$ and $v_{2}\neq u_{i}$.
\end{myclaim}

\begin{proofofclaim}[Proof of Claim \protect\ref{claim-1}]
Since $H$ is a split graph, there exists a partition of its vertices into a
clique $K$ and an independent set $I$. Then, since $H$ has an $AP_{6}$ on
the vertices $v_{1},v_{2},u_{i-2},u_{i-1},w_{i-1},w_{i-2}$, Lemma~\ref%
{split-AC6-alternating-lem} implies that either $v_{1},u_{i-2},w_{i-1}\in K$
and $v_{2},u_{i-1},w_{i-2}\in I$, or $v_{1},u_{i-2},w_{i-1}\in I$ and $%
v_{2},u_{i-1},w_{i-2}\in K$. In the former case, since $w_{i-1}\in K$ and $%
w_{i-1}w_{i}$ is not an edge in $H$, it follows that $w_{i}\in I$. Thus $%
v_{1}\neq w_{i}$, since $v_{1}\in K$. Furthermore, since $w_{i}\in I$ and $%
u_{i}w_{i}$ is an edge in $H$, it follows that $u_{i}\in K$. Thus $v_{2}\neq
u_{i}$, since $v_{2}\in I$. Similarly, in the latter case, since $u_{i-1}\in
K$ and $u_{i-1}u_{i}$ is not an edge in $H$, it follows that $u_{i}\in I$.
Thus $v_{2}\neq u_{i}$, since $v_{2}\in K$. Furthermore, since $u_{i}\in I$
and $u_{i}w_{i}$ is an edge in $H$, it follows that $w_{i}\in K$. Thus $%
v_{1}\neq w_{i}$, since $v_{1}\in I$. Summarizing, in both cases $v_{1}\neq
w_{i}$ and $v_{2}\neq u_{i}$.
\end{proofofclaim}

\medskip

Suppose that $v_{1}w_{i}$ is not an edge in $H$. Then $u_{i}w_{i}$ is in
conflict with $v_{1}u_{i-1}$, since also $u_{i-1}u_{i}$ is not an edge in $H$%
. Therefore $\ell _{u_{i}w_{i}}=\overline{\ell _{v_{1}u_{i-1}}}$. Thus,
since $\ell _{u_{i}w_{i}}=\overline{\ell _{2}}$ and $\ell _{v_{1}u_{i-1}}=%
\overline{\ell _{1}}$, it follows that $\ell _{1}=\overline{\ell _{2}}$,
which is a contradiction, since no two literals among $\{\ell _{1},\ell
_{2},\ell _{3}\}$ are one the negation of the other. Therefore $v_{1}w_{i}$
is an edge in $H$. Furthermore $\ell _{v_{1}w_{i}}=\ell _{3}$, since $\ell
_{v_{2}w_{i-1}}=\overline{\ell _{3}}$ and $w_{i-1}w_{i},v_{1}v_{2}$ are not
edges in $H$. By symmetry it follows that also $v_{2}u_{i}$ is an edge in $H$
and that $\ell _{v_{2}u_{i}}=\ell _{1}$. Thus the vertices $%
v_{1},v_{2},u_{i},u_{i-1},w_{i-1},w_{i}$ build an $AP_{6}$ in $H$, which has
the literals $\ell _{1},\ell _{2},\ell _{3}$ on its edges (in this order).
This completes the induction step whenever $i$ is even.

Let now $i\geq 3$ be odd. Then $i-1$ is even, and thus there exists by the
induction hypothesis an $AP_{6}$ in $H$ on the vertices $%
v_{1},v_{2},u_{i-1},u_{i-2},w_{i-2},w_{i-1}$ which has the literals $\ell
_{1},\ell _{2},\ell _{3}$ on its edges (in this order). That is, $\ell
_{v_{2}u_{i-1}}=\ell _{1}$, $\ell _{u_{i-2}w_{i-2}}=\ell _{2}$, and $\ell
_{w_{i-1}v_{1}}=\ell _{3}$. Thus, since the edges $u_{i-2}w_{i-2}$ and $%
u_{i-1}w_{i-1}$ are in conflict by assumption, it follows that $\ell
_{u_{i-1}w_{i-1}}=\overline{\ell _{2}}$. Furthermore, since the edges $%
u_{i-1}w_{i-1}$ and $u_{i}w_{i}$ are in conflict by assumption, it follows
that $\ell _{u_{i}w_{i}}=\ell _{2}$. Thus the vertices $%
v_{1},v_{2},u_{i-1},u_{i},w_{i},w_{i-1}$ build an $AP_{6}$ in $H$, which has
the literals $\ell _{1},\ell _{2},\ell _{3}$ on its edges (in this order).
This completes the induction step whenever $i$ is odd.

Summarizing, for $i=t$, there exists an $AP_{6}$ in $H$ on the vertices $%
v_{1},v_{2},u_{t-1},u_{t},w_{t},w_{t-1}$ (if $t$ is odd), or on the vertices 
$v_{1},v_{2},u_{t},u_{t-1},w_{t-1},w_{t}$ (if $t$ is even), which has the
literals $\ell _{1},\ell _{2},\ell _{3}$ on its edges (in this order). In
both cases where $t$ is even or odd, this $AP_{6}$ has the non-edge $%
v_{1}v_{2}$ as it base and the edge $e=u_{t}w_{t}$ as its ceiling. This
completes the proof of the lemma.\qed
\end{proof}

\begin{lemma}
\label{share-two-then-three-literals-lem}Let $\alpha =(\ell _{1}\vee \ell
_{2}\vee \ell _{3})$ and $\beta =(\ell _{1}\vee \ell _{2}\vee \ell _{4})$ be
two clauses of $\phi _{1}$ that share two literals $\ell _{1}$ and~$\ell
_{2} $. Then also $\ell _{3}=\ell _{4}$.
\end{lemma}

\begin{proof}
By the construction of the formula $\phi _{1}$ (cf.~Section~\ref%
{linear-interval-satisfiability-sec}), the clauses $\alpha $ and $\beta $
correspond to two $AC_{6}$'s in $H$. Since $H$ is a split graph, Lemma~\ref%
{split-no-AP5-double-AP6-lem} implies that each of these two $AC_{6}$'s is
an $AP_{6}$, i.e.~an alternating path of length $6$ (cf.~Figure~\ref%
{AC6-2-fig}). Let $v_{1},v_{2},v_{3},v_{4},v_{5},v_{6}$ be the vertices of
the first $AP_{6}$, which has the literals $\ell _{1},\ell _{2},\ell _{3}$
on its edges (in this order). Note that, by the construction of $\phi _{1}$,
no two literals among $\{\ell _{1},\ell _{2},\ell _{3}\}$ are one the
negation of the other, i.e.~$\ell _{1}\neq \overline{\ell _{2}}$, $\ell
_{1}\neq \overline{\ell _{3}}$, and $\ell _{2}\neq \overline{\ell _{3}}$.
Furthermore let $w_{1},w_{2},w_{3},w_{4},w_{5},w_{6}$ be the vertices of the
second $AP_{6}$, which has the literals $\ell _{1},\ell _{2},\ell _{4}$ on
its edges (in this order). Since $H$ is a split graph, there exists a
partition of its vertices into a clique $K$ and an independent set $I$.

Consider now the base $v_{5}v_{6}$ and the ceiling $v_{2}v_{3}$ of the first 
$AP_{6}$ (cf.~Definition~\ref{base-ceiling-AP6-def}). That is, the vertices
of this $AP_{6}$ can be ordered as $v_{5},v_{6},v_{1},v_{2},v_{3},v_{4}$
(where $v_{5}v_{6}$ is not an edge); then the literals on its edges are $%
\ell _{3},\ell _{1},\ell _{2}$ (in this order). Since $\ell
_{v_{2}v_{3}}=\ell _{w_{2}w_{3}}=\ell _{1}$, there exists by Lemma~\ref%
{extension-AP6-same-base-wanted-ceiling-lem} an $AP_{6}$ with $v_{5}v_{6}$
as its base and $w_{2}w_{3}$ as its ceiling, which has the literals $\ell
_{3},\ell _{1},\ell _{2}$ on its edges (in this order). Note that the
ordering of the vertices in this $AP_{6}$ can be either $%
v_{5},v_{6},a,w_{3},w_{2},b$, or $v_{5},v_{6},a,w_{2},w_{3},b$, for some
vertices $a$ and $b$ of $H$. We distinguish now these two cases.

\medskip

\textbf{Case~1.}~The $AP_{6}$ with $v_{5}v_{6}$ as its base and $w_{2}w_{3}$
as its ceiling has vertex ordering $v_{5},v_{6},a,w_{3},w_{2},b$. Consider
now the base $aw_{3}$ and the ceiling $bv_{5}$ of this $AP_{6}$. That is,
its vertices can be ordered as $a,w_{3},w_{2},b,v_{5},v_{6}$ (where $aw_{3}$
is not an edge); then the literals on its edges are $\ell _{1},\ell
_{2},\ell _{3}$ (in this order). Since $\ell _{bv_{5}}=\ell
_{w_{4}w_{5}}=\ell _{2}$, there exists by Lemma~\ref%
{extension-AP6-same-base-wanted-ceiling-lem} an $AP_{6}$ with $aw_{3}$ as
its base and $w_{4}w_{5}$ as its ceiling, which has the literals $\ell
_{1},\ell _{2},\ell _{3}$ on its edges (in this order). Note that the
ordering of the vertices in this $AP_{6}$ can be either $%
a,w_{3},c,w_{5},w_{4},d$, or $a,w_{3},c,w_{4},w_{5},d$, for some vertices $c$
and $d$ of $H$. We distinguish now these two cases.

\medskip

\textbf{Case~1.1.}~The $AP_{6}$ with $aw_{3}$ as its base and $w_{4}w_{5}$
as its ceiling has vertex ordering $a,w_{3},c,w_{5},w_{4},d$. Since no two
literals among $\{\ell _{1},\ell _{2},\ell _{3}\}$ are one the negation of
the other, it follows that no pair among the edges of this $AP_{6}$ is in
conflict. Thus Lemma~\ref{AC6-if-not-in-conflict-all-not-isolated-lem}
implies in particular that the edge $w_{3}w_{4}$ exists in $H$. This is a
contradiction to our initial assumption that the vertices $%
w_{1},w_{2},w_{3},w_{4},w_{5},w_{6}$ build an $AC_{6}$ (and thus $w_{3}w_{4}$
is not an edge).

\medskip

\textbf{Case~1.2.}~The $AP_{6}$ with $aw_{3}$ as its base and $w_{4}w_{5}$
as its ceiling has vertex ordering $a,w_{3},c,w_{4},w_{5},d$. Then Lemma~\ref%
{split-AC6-alternating-lem} implies that either $w_{3}\in K$ and $w_{5}\in I$%
, or $w_{3}\in I$ and $w_{5}\in K$. However, due to our initial assumption
that the vertices $w_{1},w_{2},w_{3},w_{4},w_{5},w_{6}$ build an $AC_{6}$,
Lemma~\ref{split-AC6-alternating-lem} implies that either $w_{3},w_{5}\in K$
or $w_{3},w_{5}\in I$, which is a contradiction.

\medskip

\textbf{Case~2.}~The $AP_{6}$ with $v_{5}v_{6}$ as its base and $w_{2}w_{3}$
as its ceiling has vertex ordering $v_{5},v_{6},a,w_{2},w_{3},b$. Consider
now the base $aw_{2}$ and the ceiling $bv_{5}$ of this $AP_{6}$. That is,
its vertices can be ordered as $a,w_{2},w_{3},b,v_{5},v_{6}$ (where $aw_{2}$
is not an edge); then the literals on its edges are $\ell _{1},\ell
_{2},\ell _{3}$ (in this order). Since $\ell _{bv_{5}}=\ell
_{w_{4}w_{5}}=\ell _{2}$, there exists by Lemma~\ref%
{extension-AP6-same-base-wanted-ceiling-lem} an $AP_{6}$ with $aw_{2}$ as
its base and $w_{4}w_{5}$ as its ceiling, which has the literals $\ell
_{1},\ell _{2},\ell _{3}$ on its edges (in this order). Note that the
ordering of the vertices in this $AP_{6}$ can be either $%
a,w_{2},c,w_{5},w_{4},d$, or $a,w_{2},c,w_{4},w_{5},d$, for some vertices $c$
and $d$ of $H$. We distinguish now these two cases.

\medskip

\textbf{Case~2.1.}~The $AP_{6}$ with $aw_{2}$ as its base and $w_{4}w_{5}$
as its ceiling has vertex ordering $a,w_{2},c,w_{5},w_{4},d$. Then Lemma~\ref%
{split-AC6-alternating-lem} implies that either $w_{2}\in K$ and $w_{4}\in I$%
, or $w_{2}\in I$ and $w_{4}\in K$. However, due to our initial assumption
that the vertices $w_{1},w_{2},w_{3},w_{4},w_{5},w_{6}$ build an $AC_{6}$,
Lemma~\ref{split-AC6-alternating-lem} implies that either $w_{2},w_{4}\in K$
or $w_{2},w_{4}\in I$, which is a contradiction.

\medskip

\textbf{Case~2.2.}~The $AP_{6}$ with $aw_{2}$ as its base and $w_{4}w_{5}$
as its ceiling has vertex ordering $a,w_{2},c,w_{4},w_{5},d$. Since no two
literals among $\{\ell _{1},\ell _{2},\ell _{3}\}$ are one the negation of
the other, it follows that no pair among the edges of this $AP_{6}$ is in
conflict. Thus Lemma~\ref{AC6-if-not-in-conflict-all-not-isolated-lem}
implies in particular that the edge $w_{5}w_{2}$ exists in $H$ and that $%
ad||w_{5}w_{2}$. Thus, since $\ell _{ad}=\ell _{3}$, it follows that $\ell
_{w_{5}w_{2}}=\overline{\ell _{3}}$. Recall now that we initially assumed
that the vertices $w_{1},w_{2},w_{3},w_{4},w_{5},w_{6}$ build an $AP_{6}$ in 
$H$, which has the literals $\ell _{1},\ell _{2},\ell _{4}$ on its edges (in
this order). Similarly, Lemma~\ref%
{AC6-if-not-in-conflict-all-not-isolated-lem} implies for this $AP_{6}$ that 
$w_{6}w_{1}||w_{5}w_{2}$. Thus, since $\ell _{w_{6}w_{1}}=\ell _{4}$, it
follows that $\ell _{w_{5}w_{2}}=\overline{\ell _{4}}$. That is, $\ell
_{w_{5}w_{2}}=\overline{\ell _{3}}=\overline{\ell _{4}}$, and thus $\ell
_{3}=\ell _{4}$. This completes the proof of the lemma.\qed
\end{proof}

We are now ready to prove the three main structural properties of the formula $%
\phi _{1}\wedge \phi _{2}$ in Lemmas~\ref{disjoint-clauses-lem},~\ref{congruent-clauses-phi-2-lem-a},~and~\ref{congruent-clauses-phi-2-lem-b}, 
respectively. The proof of the next lemma is
a based on the results of~\cite{RaschleSimon95}.

\begin{lemma}
\label{disjoint-clauses-lem}Let $\alpha $ and $\beta $ be two clauses of $%
\phi _{1}$. If $\alpha $ and $\beta $ share at least one variable, then $%
\{\alpha ,\overline{\alpha }\}=\{\beta ,\overline{\beta }\}$.
\end{lemma}

\begin{proof}
In Theorem~3.2 of~\cite{RaschleSimon95}, the authors consider an arbitrary
graph $G$ and its conflict graph~$G^{\ast }$, which is bipartite. For every
edge $e$ of $G$, denote by $C^{\ast }(e)$ the connected component of $%
G^{\ast }$ in which the vertex $e$ belongs. For simplicity of the
presentation, we will also refer in the following to $C^{\ast }(e)$ as the
set of the corresponding edges in $G$. The authors of~\cite{RaschleSimon95}
assume an arbitrary $2$-coloring of the vertices of $G^{\ast }$ (i.e.~of the
edges of $G$), such that there is no monochromatic double $AP_{6}$,
i.e.~there is no double $AP_{6}$ on the edges of one edge-color class of $G$%
. Furthermore they assume that there is a monochromatic $AP_{6}$ in $G$ on
the vertices $v_{1},\ldots ,v_{6}$ (which is not a double $AP_{6}$). Since
this $AP_{6}$ is monochromatic, it follows that no pair among its three
edges is in conflict in $G$ (since any two edges in conflict would have
different colors). Thus the edges $v_{3}v_{6},v_{4}v_{1},v_{5}v_{2}$ exist
in $G$ and $v_{4}v_{5}||v_{3}v_{6}$, $v_{2}v_{3}||v_{4}v_{1}$, and $%
v_{6}v_{1}||v_{5}v_{2}$ by Lemma~\ref%
{AC6-if-not-in-conflict-all-not-isolated-lem}. The non-edge $v_{1}v_{2}$ is
called the \emph{base} of the $AP_{6}$ (cf.~Definition~\ref%
{base-ceiling-AP6-def}); furthermore we call the edge $v_{3}v_{6}$ the \emph{%
front} of the $AP_{6}$~\cite{RaschleSimon95}. Note here that the choice of
the base $v_{1}v_{2}$ is arbitrary (the $AP_{6}$ has three bases $v_{1}v_{2}$%
, $v_{3}v_{4}$, and $v_{5}v_{6}$). Then, they prove\footnote{%
In~\cite{RaschleSimon95}, the authors prove within the proof of Theorem~3.2
a more general statement (cf.~equations~(2) and~(3) in~\cite{RaschleSimon95}%
). In particular, they flip the colors of all edges $xy$ of $G$, for which
there exists an $AP_{6}$ in $G$ having $v_{1}v_{2}$ as its basis and $xy$ as
its front (cf.~equation~(2) in~\cite{RaschleSimon95}); note here that all
these edges, whose color is being flipped, may correspond to one or more
connected components in the conflict graph $G^{\ast }$. Then they prove that
in the new edge coloring of $G$ no flipped edge participates in a
monochromatic $AP_{6}$ (cf.~equation~(3) in~\cite{RaschleSimon95}). In their
proof, which is correct and technically involved, they actually prove that
this happens also when we flip the colors of only one connected component $%
C^{\ast }(v_{3}v_{6})$ of $G^{\ast }$, where $v_{3}v_{6}$ is the front of
the initial monochromatic~$AP_{6}$ on the vertices~$v_{1},\ldots ,v_{6}$.%
\vspace{-0.5cm}} in Theorem~3.2 that, if we flip the colors of all edges of $%
C^{\ast }(v_{3}v_{6})$ then in the new edge coloring of $G$ no edge of $%
C^{\ast }(v_{3}v_{6})$ participates in a monochromatic $AP_{6}$. Note
furthermore that $v_{4}v_{5}\in C^{\ast }(v_{3}v_{6})$, since $%
v_{4}v_{5}||v_{3}v_{6}$, and thus also the color of $v_{4}v_{5}$ changes by
flipping the colors of the edges in $C^{\ast }(v_{3}v_{6})$.

We now apply the results of~\cite{RaschleSimon95} in our case as follows.
Consider two clauses $\alpha $ and $\beta $ of $\phi _{1}$ that share at
least one variable. That is, each of the dual clauses $\{\alpha ,\overline{%
\alpha }\}$ shares at least one literal with at least one of the dual
clauses $\{\beta ,\overline{\beta }\}$. If $\beta \in \{\alpha ,\overline{%
\alpha }\}$ then clearly $\{\alpha ,\overline{\alpha }\}=\{\beta ,\overline{%
\beta }\}$, and thus the lemma follows.

Let now $\beta \notin \{\alpha ,\overline{\alpha }\}$. Consider the $AC_{6}$
of $H$ on the vertices $v_{1},\ldots ,v_{6}$ that corresponds to the dual
clauses $\{\alpha ,\overline{\alpha }\}$. Since $H$ is a split graph, it
follows by Lemma~\ref{split-no-AP5-double-AP6-lem} that $H$ does not contain
any $AP_{5}$ or any double $AP_{6}$. Therefore this $AC_{6}$ of $H$ on the
vertices $v_{1},\ldots ,v_{6}$ is an $AP_{6}$ (but not a double $AP_{6}$).
Let $e=v_{2}v_{3}$, $e^{\prime }=v_{4}v_{5}$, and $e^{\prime \prime
}=v_{6}v_{1}$. This $AP_{6}$ has the non-edge $v_{1}v_{2}$ as its \emph{base}
and the edge $v_{3}v_{6}$ as its \emph{front}, cf.~Definition~\ref%
{base-ceiling-AP6-def}. Note that either $\alpha =(\ell _{e}\vee \ell
_{e^{\prime }}\vee \ell _{e^{\prime \prime }})$ and $\overline{\alpha }=(%
\overline{\ell _{e}}\vee \overline{\ell _{e^{\prime }}}\vee \overline{\ell
_{e^{\prime \prime }}})$, or $\alpha =(\overline{\ell _{e}}\vee \overline{%
\ell _{e^{\prime }}}\vee \overline{\ell _{e^{\prime \prime }}})$ and $%
\overline{\alpha }=(\ell _{e}\vee \ell _{e^{\prime }}\vee \ell _{e^{\prime
\prime }})$. Assume without loss of generality that $\alpha =(\ell _{e}\vee
\ell _{e^{\prime }}\vee \ell _{e^{\prime \prime }})$ and $\overline{\alpha }%
=(\overline{\ell _{e}}\vee \overline{\ell _{e^{\prime }}}\vee \overline{\ell
_{e^{\prime \prime }}})$. Recall by our assumption that $\alpha $ shares at
least one literal with at least one of the dual clauses $\{\beta ,\overline{%
\beta }\}$. Assume without loss of generality that $\alpha $ shares at least
one literal with $\beta $ (the case where $\alpha $ shares at least one
literal with $\overline{\beta }$ can be handled in exactly the same way).
Furthermore, let without loss of generality $\ell _{e^{\prime }}$ be the
common literal of $\alpha $ and $\beta $, i.e.~let $\beta =(\ell _{e^{\prime
}}\vee \ell _{p}\vee \ell _{q})$.

Since $\alpha $ is a clause of $\phi _{1}$, it follows by the construction
of $\phi _{1}$ that no two literals among $\{\ell _{e},\ell _{e^{\prime
}},\ell _{e^{\prime \prime }}\}$ are one the negation of the other (cf.
lines~\ref{algorithm-phi_1-3}-\ref{algorithm-phi_1-5} of Algorithm~\ref%
{alg-phi_1-constr}). Similarly no two literals among $\{\ell _{e^{\prime
}},\ell _{p},\ell _{q}\}$ are one the negation of the other, since $\beta $
is a clause of $\phi _{1}$. Consider now an \emph{arbitrary} truth
assignment $\tau $ of the variables $x_{1},x_{2},\ldots ,x_{k}$, such that $%
\alpha =0$ in $\tau $, i.e.~$\ell _{e}=\ell _{e^{\prime }}=\ell _{e^{\prime
\prime }}=0$ in $\tau $. Note that such an assignment exists, since no two
literals among $\{\ell _{e},\ell _{e^{\prime }},\ell _{e^{\prime \prime }}\}$
are one the negation of the other. Let $\chi $ be the $2$-coloring of the
vertices of $H^{\ast }$ (i.e.~of the edges of $H$) that corresponds to the
truth assignment $\tau $, cf.~Observation~\ref{truth-assignment-coloring-obs}%
. Since $\alpha =0$ in the truth assignment $\tau $, Observation~\ref%
{monochromatic-coloring-assignment-obs} implies that the $AP_{6}$ on the
vertices $v_{1},\ldots ,v_{6}$ is monochromatic in the edge-coloring $\chi $
of $H$. Then, due to the results of~\cite{RaschleSimon95}, if we flip in $%
\chi $ the colors of all edges of $C^{\ast }(v_{3}v_{6})$, in the new edge
coloring $\chi ^{\prime }$ of $H$ no edge of $C^{\ast }(v_{3}v_{6})$
participates in a monochromatic $AP_{6}$.

Let $\tau ^{\prime }$ be the truth assignment that corresponds to this new
coloring $\chi ^{\prime }$ (cf.~Observation~\ref%
{truth-assignment-coloring-obs}). Then $\tau $ and $\tau ^{\prime }$
coincide on all variables except the variable of the component $C^{\ast
}(v_{3}v_{6})$ of $H^{\ast }$. Note that the color of $e^{\prime
}=v_{4}v_{5} $ has been flipped in the transition from $\chi ^{\prime }$ to $%
\chi $, since $e^{\prime }\in C^{\ast }(v_{3}v_{6})$, and thus $\ell
_{e^{\prime }}=1 $ in $\chi ^{\prime }$. Furthermore, since no edge of $%
C^{\ast }(v_{3}v_{6})$ participates in a monochromatic $AP_{6}$ in $\chi
^{\prime }$, it follows that both clauses $\beta =(\ell _{e^{\prime }}\vee
\ell _{p}\vee \ell _{q})$ and $\overline{\beta }=(\overline{\ell _{e^{\prime
}}}\vee \overline{\ell _{p}}\vee \overline{\ell _{q}})$ are satisfied in $%
\tau ^{\prime }$, i.e.~$\beta =1$ and $\overline{\beta }=1$ in $\tau
^{\prime }$, since both $\beta $ and $\overline{\beta }$ include one of the
literals $\{\ell _{e^{\prime }},\overline{\ell _{e^{\prime }}}\}$. We will
now prove that $\{\ell _{p},\ell _{q}\}\cap \{\ell _{e},\ell _{e^{\prime
\prime }}\}\neq \emptyset $. Assume otherwise that $\{\ell _{p},\ell
_{q}\}\cap \{\ell _{e},\ell _{e^{\prime \prime }}\}=\emptyset $. We
distinguish the following three cases.

\medskip

\textbf{Case~1.}~$\ell _{p}\neq \ell _{e^{\prime }}$ and $\ell _{q}\neq \ell
_{e^{\prime }}$. Then, since no two literals among $\{\ell _{e^{\prime
}},\ell _{p},\ell _{q}\}$ are one the negation of the other, it follows that 
$\ell _{p},\ell _{q}\notin \{\ell _{e^{\prime }},\overline{\ell _{e^{\prime
}}}\}$. Therefore the values of $\ell _{p}$ and $\ell _{q}$ remain the same
in both assignments $\tau $ and $\tau ^{\prime }$. Since $\tau $ has been
assumed to be an arbitrary assignment such that $\ell _{e}=\ell _{e^{\prime
}}=\ell _{e^{\prime \prime }}=0$ in $\tau $, we can choose the assignment $%
\tau $ to be such that $\ell _{p}=\ell _{q}=1$ in $\tau $. Since the value
of $\ell _{e^{\prime }}$ changes to $1$ in $\tau ^{\prime }$, while the
values of $\ell _{p}$ and $\ell _{q}$ are the same in both $\tau $ and $\tau
^{\prime }$, it follows that $\ell _{e^{\prime }}=\ell _{p}=\ell _{q}=1$ in $%
\tau ^{\prime }$, and thus $\overline{\beta }=0$ in $\tau ^{\prime }$, which
is a contradiction.

\medskip

\textbf{Case~2.}~Exactly one of $\{\ell _{p},\ell _{q}\}$ is equal to $\ell
_{e^{\prime }}$. Let without loss of generality $\ell _{p}=\ell _{e^{\prime
}}$ and $\ell _{q}\neq \ell _{e^{\prime }}$, i.e.~$\ell _{q}\notin \{\ell
_{e^{\prime }},\overline{\ell _{e^{\prime }}}\}$. Therefore the value of $%
\ell _{q}$ remains the same in both assignments $\tau $ and $\tau ^{\prime }$%
. Since $\tau $ has been assumed to be an arbitrary assignment such that $%
\ell _{e}=\ell _{e^{\prime }}=\ell _{e^{\prime \prime }}=0$ in $\tau $, we
can choose the assignment $\tau $ to be such that $\ell _{q}=1$ in $\tau $.
Since the value of $\ell _{p}=\ell _{e^{\prime }}$ changes to $1$ in $\tau
^{\prime }$, while the value of $\ell _{q}$ is the same in both $\tau $ and $%
\tau ^{\prime }$, it follows that $\ell _{e^{\prime }}=\ell _{p}=\ell _{q}=1$
in $\tau ^{\prime }$, and thus $\overline{\beta }=0$ in $\tau ^{\prime }$,
which is a contradiction.

\medskip

\textbf{Case~3.}~$\ell _{p}=\ell _{q}=\ell _{e^{\prime }}$. Then $\beta
=(\ell _{e^{\prime }}\vee \ell _{p}\vee \ell _{q})=(\ell _{e^{\prime }})$
and $\overline{\beta }=(\overline{\ell _{e^{\prime }}}\vee \overline{\ell
_{p}}\vee \overline{\ell _{q}})=(\overline{\ell _{e^{\prime }}})$, and thus
it is not possible that both $\beta =1$ and $\overline{\beta }=1$ in $\tau
^{\prime }$, which is again a contradiction.

\medskip

Therefore $\{\ell _{p},\ell _{q}\}\cap \{\ell _{e},\ell _{e^{\prime \prime
}}\}\neq \emptyset $. Thus, since the clauses $\alpha $ and $\beta $ share
also the literal $\ell _{e^{\prime }}$, it follows that $\alpha $ and $\beta 
$ share at least two literals. Therefore $\alpha =\beta $ by Lemma~\ref%
{share-two-then-three-literals-lem}. This is a contradiction, since we
assumed that $\beta \notin \{\alpha ,\overline{\alpha }\}$. Therefore $\beta
\in \{\alpha ,\overline{\alpha }\}$, and thus $\{\alpha ,\overline{\alpha }%
\}=\{\beta ,\overline{\beta }\}$. This completes the proof of the lemma.%
\qed
\end{proof}

\begin{definition}
\label{phi2-partition-def}The clauses of $\phi _{2}$ are partitioned into
the \emph{sub-formulas} $\phi _{2}^{\prime },\phi _{2}^{\prime \prime }$,
such that $\phi _{2}^{\prime }$ contains all tautologies of $\phi _{2}$ and
all clauses of $\phi _{2}$ in which at least one literal corresponds to an
uncommitted edge, while~$\phi _{2}^{\prime \prime }$ contains all the
remaining clauses of $\phi _{2}$.
\end{definition}

\begin{lemma}
\label{congruent-clauses-phi-2-lem-a}Let $\{e_{1},e_{2},e_{3}\}$ be the
three edges of an $AC_{6}$ in $H$, which has clauses in $\phi _{1}$. Let $e$
be an edge of~$H$ such that $(\ell _{e}\vee \ell _{e_{1}})$ is a clause in $%
\phi _{2}^{\prime \prime }$. Then $\phi _{2}^{\prime \prime }$ contains also
at least one of the clauses $\{(\ell _{e}\vee \overline{\ell _{e_{2}}}%
),(\ell _{e}\vee \overline{\ell _{e_{3}}})\}$.
\end{lemma}

\begin{proof}
Recall that $H$ is the associated split graph of $\widetilde{G}$, where $%
\widetilde{G}$ is the bipartite complement $\widehat{C}(P)$ of the
domination bipartite graph $C(P)$ of the partial order $P$, cf.~Definitions~%
\ref{bipartite-split-def} and~\ref{C(P)-def}. For the purposes of the proof
denote $C(P)=(U,V,E)$, where $U=\{u_{1},u_{2},\ldots ,u_{n}\}$, $%
V=\{v_{1},v_{2},\ldots ,v_{n}\}$; then $u_{i}v_{j}\in E$ if and only if $%
u_{i}<_{P}u_{j}$ (cf.~Definition~\ref{C(P)-def}). Furthermore denote $%
\widetilde{G}=(U,V,\widetilde{E})$ for the bipartite complement $\widetilde{G%
}=\widehat{C}(P)$ of $C(P)$. Then $H=(U\cup V,E_{H})$, where $E_{H}=%
\widetilde{E}\cup (V\times V)$ (cf.~Definition~\ref{bipartite-split-def}).
Moreover let $E_{0}=\{u_{i}v_{i}\ |\ 1\leq i\leq n\}$ and observe that $%
E_{0}\subseteq \widetilde{E}\subseteq E_{H}$. Since edges of $E$ correspond
to non-edges of $\widetilde{E}$, it follows by the definition of $E$ that $%
u_{i}v_{j}\notin \widetilde{E}$ if and only if $u_{i}<_{P}u_{j}$. That is,
the non-edges of $\widetilde{E}$ between vertices of $U$ and vertices of $V$
follow the transitivity of the partial order $P$.

Since $H$ is a split graph, Lemma~\ref{split-no-AP5-double-AP6-lem} implies
that the $AC_{6}$ of $H$ is an $AP_{6}$, i.e.~an alternating path of length~$%
6$ (cf.~Figure~\ref{AC6-2-fig}). Furthermore, since $V$ induces a clique and 
$U$ induces an independent set in $H$, Lemma~\ref{split-AC6-alternating-lem}
implies that the vertices of the $AP_{6}$ in $H$ belong alternately to $U$
and to $V$. Thus let $u_{i},v_{j},u\,_{p},v_{q},u_{r},v_{s}$ be the vertices
of the $AP_{6}$ (where $u_{i}v_{j}\notin E_{H}$ according to our notation,
cf.~Definition~\ref{AC-2k-def}). Without loss of generality let $%
e_{1}=u\,_{p}v_{j}$, $e_{2}=u_{r}v_{q}$, and $e_{3}=u_{i}v_{s}$. Since the $%
AP_{6}$ has clauses in $\phi _{1}$ by assumption, note by the construction
of $\phi _{1}$ (cf.~Section~\ref{linear-interval-satisfiability-sec}) that
no two literals among $\{\ell _{e_{1}},\ell _{e_{2}},\ell _{e_{3}}\}$ are
one the negation of the other. Therefore no pair among the edges $%
\{e_{1},e_{2},e_{3}\}$ is in conflict, and thus Lemma~\ref%
{AC6-if-not-in-conflict-all-not-isolated-lem} implies that the edges $%
u_{p}v_{s},u_{i}v_{q},u_{r}v_{j}$ exist in $H$ and $%
e_{2}=u_{r}v_{q}||u_{p}v_{s}$, $e_{1}=u_{p}v_{j}||u_{i}v_{q}$, and $%
e_{3}=u_{i}v_{s}||u_{r}v_{j}$. Therefore $\ell _{u_{i}v_{q}}=\overline{\ell
_{e_{1}}}$, $\ell _{u_{p}v_{s}}=\overline{\ell _{e_{2}}}$, and~$\ell
_{u_{r}v_{j}}=\overline{\ell _{e_{3}}}$.

Since $e_{1}=u\,_{p}v_{j}$ and $(\ell _{e}\vee \ell _{e_{1}})$ is a clause
of $\phi _{2}^{\prime \prime }$ (and thus also of $\phi _{2}$), it follows
by the construction of $\phi _{2}$ (cf.~Section~\ref%
{linear-interval-satisfiability-sec}) that either $e=u_{a}v_{p}$ or $%
e=u_{j}v_{a}$ for some index $a\in \{1,2,\ldots ,n\}$.

\medskip

\textbf{Case~1.}~$e=u_{a}v_{p}$. Denote $E_{H}^{\prime }=E_{H}\setminus
E_{0} $. Then it follows by the construction of $\phi _{2}$ that $%
u_{a}v_{j}\notin E_{H}^{\prime }$, and thus either $u_{a}v_{j}\notin E_{H}$
or $u_{a}v_{j}\in E_{0}$. Furthermore, since $(\ell _{e}\vee \ell _{e_{1}})$
is a clause of $\phi _{2}^{\prime \prime }$ by assumption, it follows by
Definition~\ref{phi2-partition-def} that $e$ is a committed edge in $H$.
That is, there exists an edge $e^{\prime }=u_{b}v_{c}$ such that $e^{\prime
}||e$, and thus $\ell _{e^{\prime }}=\overline{\ell _{e}}$. Since $e^{\prime
}||e$, it follows that $u_{a}v_{c},u_{b}v_{p}\notin E_{H}$. Furthermore,
since $u_{b}v_{p},u_{p}v_{q}\notin E_{H}$, it follows that $u_{b}<_{P}u_{p}$
and $u_{p}<_{P}u_{q}$. Therefore $u_{b}<_{P}u_{q}$, since $P$ is a partial
order, and thus also $u_{b}v_{q}\notin E_{H}$.

Note that either $a=j$ or $a\neq j$ (cf.~Figures~\ref{phi-2-fig-1} and~\ref%
{phi-2-fig-2}, respectively. We distinguish now these two cases, which are
illustrated in Figures~\ref{a-equal-to-j-fig} and~\ref{a-not-equal-to-j-fig}%
, respectively. In these figures, the edges $e_{1},e_{2},e_{3}$ of the $%
AP_{6}$, as well as the edges $e$ and $e^{\prime }$, are drawn by thick
lines and all other edges are drawn by thin lines, while non-edges are
illustrated with dashed lines.

\medskip

\textbf{Case~1.1.}~$a=j$ (cf.~Figure~\ref{a-equal-to-j-fig}). Suppose that $%
u_{i}v_{c}\in E_{H}$. Then $u_{i}v_{c}||u_{a}v_{j}=u_{j}v_{j}$, since $%
u_{i}v_{j},u_{a}v_{c}\notin E_{H}$. Thus the edge $u_{j}v_{j}\in E_{0}$ is
committed, which is a contradiction by Lemma~\ref{ei-isolated-vertices-lem}.
Therefore $u_{i}v_{c}\notin E_{H}$. Suppose now that $u_{p}v_{c}\notin E_{H}$%
. Then $u_{b}v_{c}||u_{p}v_{p}$, since $u_{b}v_{p},u_{p}v_{c}\notin E_{H}$.
Thus the edge $u_{p}v_{p}\in E_{0}$ is committed, which is a contradiction
by Lemma~\ref{ei-isolated-vertices-lem}. Therefore $u_{p}v_{c}\in E_{H}$.
Furthermore $u_{p}v_{c}||u_{i}v_{q}$, since $u_{p}v_{q},u_{i}v_{c}\notin
E_{H}$, and thus $\ell _{u_{p}v_{c}}=\overline{\ell _{u_{i}v_{q}}}$.
Therefore, since $\ell _{u_{i}v_{q}}=\overline{\ell _{e_{1}}}$, it follows
that $\ell _{u_{p}v_{c}}=\ell _{e_{1}}$.

Suppose that $u_{a}v_{q}\notin E_{H}$, and thus $u_{a}<_{P}u_{q}$. Then,
since $u_{i}v_{j}\notin E_{H}$, it follows that $u_{i}<_{P}u_{j}$.
Therefore, since $a=j$ and $P$ is a partial order, it follows that $%
u_{i}<_{P}u_{q}$, and thus $u_{i}v_{q}\notin E_{H}$, which is a
contradiction. Therefore $u_{a}v_{q}\in E_{H}$. Furthermore $%
u_{a}v_{q}||u_{p}v_{c}$, since $u_{a}v_{c},u_{p}v_{q}\notin E_{H}$, and thus 
$\ell _{u_{a}v_{q}}=\overline{\ell _{u_{p}v_{c}}}$. Therefore, since $\ell
_{u_{p}v_{c}}=\ell _{e_{1}}$, it follows that $\ell _{u_{a}v_{q}}=\overline{%
\ell _{e_{1}}}$.

Since $u_{b}v_{q},u_{a}v_{c}\notin E_{H}$, it follows that $e^{\prime
}=u_{b}v_{c}||u_{a}v_{q}$, and thus $\ell _{e^{\prime }}=\overline{\ell
_{u_{a}v_{q}}}$. Therefore, since $\ell _{u_{a}v_{q}}=\overline{\ell _{e_{1}}%
}$, it follows that $\ell _{e^{\prime }}=\ell _{e_{1}}$. Finally, since $%
e^{\prime }||e$, it follows that $\ell _{e}=\overline{\ell _{e^{\prime }}}$,
and thus $\ell _{e}=\overline{\ell _{e_{1}}}$. Therefore the clause $(\ell
_{e}\vee \ell _{e_{1}})$ of $\phi _{2}^{\prime \prime }$ is a tautology,
which is a contradiction by Definition~\ref{phi2-partition-def}.

\medskip

\textbf{Case~1.2.}~$a\neq j$ (cf.~Figure~\ref{a-not-equal-to-j-fig}). Then $%
u_{a}v_{j}\notin E_{0}$. Thus, since~${u_{a}v_{j}\notin E_{H}^{\prime }}$,
it follows that $u_{a}v_{j}\notin E_{H}$. Suppose that $u_{a}v_{s}\in E_{H}$
(cf.~Figure~\ref{a-not-equal-to-j-fig}). Then $u_{a}v_{s}||u_{r}v_{j}$,
since~${u_{a}v_{j},u_{r}v_{s}\notin E_{H}}$, and thus $\ell _{u_{a}v_{s}}=%
\overline{\ell _{u_{r}v_{j}}}$. Therefore, since~$\ell _{u_{r}v_{j}}=%
\overline{\ell _{e_{3}}}$, it follows that~$\ell _{u_{a}v_{s}}=\ell _{e_{3}}$%
. Suppose that~${u_{a}v_{q}\notin E_{H}}$. Then ${u_{r}v_{q}||u_{a}v_{s}}$,
since ${u_{r}v_{s},u_{a}v_{q}\notin E_{H}}$. Therefore ${\ell _{u_{a}v_{s}}=%
\overline{\ell _{e_{2}}}}$, since ${\ell _{u_{r}v_{q}}=\ell _{e_{2}}}$.
Thus, since~${\ell _{u_{a}v_{s}}=\ell _{e_{3}}}$, it follows that ${\ell
_{e_{3}}=\overline{\ell _{e_{2}}}}$. This is a contradiction, since no two
literals among ${\{\ell _{e_{1}},\ell _{e_{2}},\ell _{e_{3}}\}}$ are one the
negation of the other. Therefore ${u_{a}v_{q}\in E_{H}}$. Moreover, since~${%
u_{a}v_{j},u_{p}v_{q}\notin E_{H}}$, it follows that~${%
u_{a}v_{q}||u_{p}v_{j}=e_{1}}$, and thus ${\ell _{u_{a}v_{q}}=\overline{\ell
_{e_{1}}}}$. Furthermore ${u_{a}v_{q}||u_{b}v_{c}=e^{\prime }}$, since~${%
u_{b}v_{q},u_{a}v_{c}\notin E_{H}}$. Therefore ${\ell _{u_{a}v_{q}}=%
\overline{\ell _{e^{\prime }}}}$. Thus ${\ell _{u_{a}v_{q}}=\ell _{e}}$,
since ${\ell _{e^{\prime }}=\overline{\ell _{e}}}$. Therefore, since ${\ell
_{u_{a}v_{q}}=\overline{\ell _{e_{1}}}}$ and ${\ell _{u_{a}v_{q}}=\ell _{e}}$%
, it follows that ${\ell _{e}=\overline{\ell _{e_{1}}}}$. Therefore the
clause $(\ell _{e}\vee \ell _{e_{1}})$ of $\phi _{2}^{\prime \prime }$ is a
tautology, which is a contradiction by Definition~\ref{phi2-partition-def}.

Therefore $u_{a}v_{s}\notin E_{H}$. Then also $u_{a}v_{s}\notin
E_{H}^{\prime }$, and thus $\phi _{2}$ has the clause $(\ell
_{u_{a}v_{p}}\vee \ell _{u_{p}v_{s}})=(\ell _{e}\vee \overline{\ell _{e_{2}}}%
)$, since $e=u_{a}v_{p}$ and $\ell _{u_{p}v_{s}}=\overline{\ell _{e_{2}}}$.
Furthermore, since both $e$ and $u_{p}v_{s}$ are committed in $H$ (as $%
e^{\prime }||e$ and $u_{r}v_{q}||u_{p}v_{s}$), the clause $(\ell _{e}\vee 
\overline{\ell _{e_{2}}})$ belongs to $\phi _{2}^{\prime \prime }$ by
Definition~\ref{phi2-partition-def}. 
\begin{figure}[tbh]
\centering 
\subfigure[] {\label{a-equal-to-j-fig} 
\includegraphics[scale=0.85]{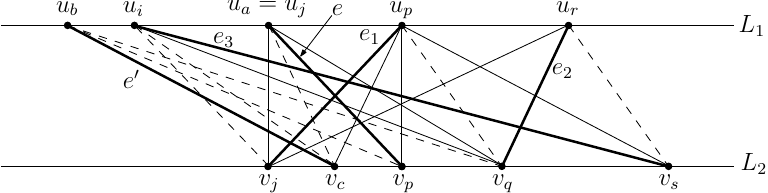}} 
\subfigure[] {\label{a-not-equal-to-j-fig} 
\includegraphics[scale=0.85]{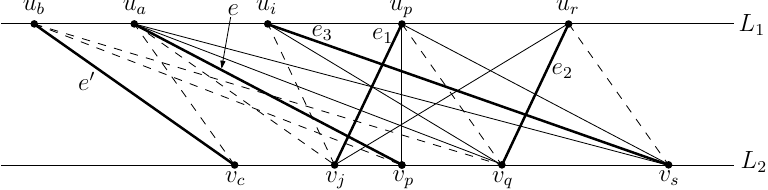}}
\caption{(a) The Case 1.1 and (b) the Case 1.2 in the proof of Lemma~\protect
\ref{congruent-clauses-phi-2-lem-a}.}
\label{a-equal-and-not-equal-to-j-fig}
\end{figure}

\textbf{Case~2.}~$e=u_{j}v_{a}$. This case is exactly symmetric to Case 1.
To see this, imagine exchanging the roles of $U$ and $V$, i.e.~$U$ induces
now a clique (instead of an independent set) and $V$ induces an independent
set (instead of a clique) in $H$. Imagine also flipping the lines $L_{1}$
and $L_{2}$ in Figure~\ref{a-equal-and-not-equal-to-j-fig} (i.e.~$L_{2}$
comes now above~$L_{1}$), such that the vertices of $U$ and $V$ still lie on
the lines~$L_{1}$ and~$L_{2}$, respectively. Similarly to Cases~1.1 and~1.2,
we distinguish the cases $a=p$ (Case~2.1) and $a\neq p$ (Case~2.2),
respectively. Then, Case~2.1 leads to a contradiction (similarly to
Case~1.1), and Case~2.2 implies that the clause $(\ell _{e}\vee \overline{%
\ell _{e_{3}}})$ belongs to $\phi _{2}^{\prime \prime }$ (instead of
the~clause~$(\ell _{e}\vee \overline{\ell _{e_{2}}})$~in~Case~1.2).

\medskip

Summarizing, if $e=u_{a}v_{p}$ then $\phi _{2}^{\prime \prime }$ includes
the clause $(\ell _{e}\vee \overline{\ell _{e_{2}}})$, while if $%
e=u_{j}v_{a} $ then $\phi _{2}^{\prime \prime }$ includes the clause $(\ell
_{e}\vee \overline{\ell _{e_{3}}})$. This completes the proof of the lemma.%
\qed
\end{proof}

\begin{lemma}
\label{congruent-clauses-phi-2-lem-b}Let $\{e_{1},e_{2},e_{3}\}$ be the
three edges of an $AC_{6}$ in $H$, which has clauses in $\phi _{1}$. Let $e$
be an edge of~$H$ such that $(\ell _{e}\vee \overline{\ell _{e_{1}}})$ is a
clause in $\phi _{2}^{\prime \prime }$. Then $\phi _{2}^{\prime \prime }$
contains also at least one of the clauses $\{(\ell _{e}\vee \ell
_{e_{2}}),(\ell _{e}\vee \ell _{e_{3}})\}$.
\end{lemma}

\begin{proof}
Since $H$ is a split graph, Lemma~\ref{split-no-AP5-double-AP6-lem} implies
that the $AC_{6}$ of $H$ is an $AP_{6}$, i.e.~an alternating path of length $%
6$ (cf.~Figure~\ref{AC6-2-fig}). Using the notation of Lemma~\ref%
{congruent-clauses-phi-2-lem-a}, denote by $V$ and $U$ the clique and the
independent set of $H$, respectively. Then the vertices of the $AP_{6}$ in $%
H $ belong alternately to $U$ and to $V$ by Lemma~\ref%
{split-AC6-alternating-lem}. That is, $u_{i},v_{j},u\,_{p},v_{q},u_{r},v_{s}$
are the vertices of the $AP_{6}$ in this order, for some vertices $%
u_{i},u\,_{p},u_{r}\in U$ and $v_{j},v_{q},v_{s}\in V$ (where $%
u_{i}v_{j},u\,_{p}v_{q},u_{r}v_{s}\notin E_{H}$ according to our notation,
cf.~Definition~\ref{AC-2k-def}). Without loss of generality let $%
e_{1}=u\,_{p}v_{j}$, $e_{2}=u_{r}v_{q}$, and $e_{3}=u_{i}v_{s}$. Then,
similarly to the preamble of the proof of Lemma~\ref%
{congruent-clauses-phi-2-lem-a}, it follows that the edges $e_{1}^{\prime
}=u_{i}v_{q}$, $e_{2}^{\prime}=u_{p}v_{s}$, and $e_{3}^{\prime }=u_{r}v_{j}$
exist in $H$ and $e_{1}=u_{p}v_{j}||u_{i}v_{q}=e_{1}^{\prime }$, $%
e_{2}=u_{r}v_{q}||u_{p}v_{s}=e_{2}^{\prime }$, and $%
e_{3}=u_{i}v_{s}||u_{r}v_{j}=e_{3}^{\prime }$. Therefore $\ell
_{e_{1}^{\prime }}=\overline{\ell _{e_{1}}}$, $\ell _{e_{2}^{\prime }}=%
\overline{\ell _{e_{2}}}$, and~$\ell _{e_{3}^{\prime }}=\overline{\ell
_{e_{3}}}$.

Since $u_{i}v_{j},u\,_{p}v_{q},u_{r}v_{s}\notin E_{H}$, it follows that the
vertices $u_{i},v_{q},u_{p},v_{s},u_{r},v_{j}$ (in this order) build an $AC_{6}$ in $H$,
where $\{e_{1}^{\prime },e_{2}^{\prime },e_{3}^{\prime }\}$ are its three
edges. Therefore, by applying Lemma~\ref{split-no-AP5-double-AP6-lem} on
this new $AC_{6}$, it follows that if $(\ell _{e}\vee \ell _{e_{1}^{\prime
}})$ is a clause in $\phi _{2}^{\prime \prime }$, then $\phi _{2}^{\prime
\prime }$ contains also at least one of the clauses $\{(\ell _{e}\vee 
\overline{\ell _{e_{2}^{\prime }}}),(\ell _{e}\vee \overline{\ell
_{e_{3}^{\prime }}})\}$. This completes the proof of the lemma, since $\ell
_{e_{1}^{\prime }}=\overline{\ell _{e_{1}}}$, $\ell _{e_{2}^{\prime }}=%
\overline{\ell _{e_{2}}}$, and~$\ell _{e_{3}^{\prime }}=\overline{\ell
_{e_{3}}}$.\qed
\end{proof}

The next corollary, which follows easily by Definition~\ref%
{gradually-mixed-def} and by Lemmas~\ref{disjoint-clauses-lem}-\ref%
{congruent-clauses-phi-2-lem-b}, allows us to use the linear time algorithm
for gradually mixed formulas (cf.~Theorem~\ref{gradually-mixed-algorithm-thm}%
) in order to solve the SAT problem on~${\phi _{1}\wedge \phi _{2}^{\prime
\prime }}$.

\begin{corollary}
\label{phi1-phi2''-gm-formula-cor}$\phi _{1}\wedge \phi _{2}^{\prime \prime
} $ is a gradually mixed formula.
\end{corollary}

\begin{proof}
First note that, by construction, every clause of $\phi _{1}$ has $3$
literals and every clause of $\phi _{2}$ has $2$ literals. Furthermore, the
first condition of Definition~\ref{gradually-mixed-def} is satisfied due to
Lemma~\ref{disjoint-clauses-lem}. Regarding the second condition of
Definition~\ref{gradually-mixed-def}, consider an arbitrary $AC_{6}$ in $H$
that has clauses in $\phi _{1}$. Denote by $\{e_{1},e_{2},e_{3}\}$ the three
edges of this $AC_{6}$. Then this $AC_{6}$ contributes to the formula $\phi
_{1}$ by the two (dual) clauses $\alpha =(\ell _{e_{1}}\vee \ell
_{e_{2}}\vee \ell _{e_{3}})$ and $\overline{\alpha }=(\overline{\ell _{e_{1}}%
}\vee \overline{\ell _{e_{2}}}\vee \overline{\ell _{e_{3}}})$, cf.~the
construction of $\phi _{1}$ in Section~\ref%
{linear-interval-satisfiability-sec}. If $(\ell _{e}\vee \ell _{e_{1}})$ is
a clause of $\phi _{2}^{\prime \prime }$, then Lemma~\ref%
{congruent-clauses-phi-2-lem-a} implies that $\phi _{2}^{\prime \prime }$
includes also at least one of the clauses $\{(\ell _{e}\vee \overline{\ell
_{e_{2}}}),(\ell _{e}\vee \overline{\ell _{e_{3}}})\}$. Similarly, if $(\ell
_{e}\vee \overline{\ell _{e_{1}}})$ is a clause of $\phi _{2}^{\prime \prime
}$, Lemma~\ref{congruent-clauses-phi-2-lem-b} implies that $\phi
_{2}^{\prime \prime }$ includes also at least one of the clauses $\{(\ell
_{e}\vee \ell _{e_{2}}),(\ell _{e}\vee \ell _{e_{3}})\}$. Therefore the
second condition of Definition~\ref{gradually-mixed-def} is also satisfied
for the formula $\phi _{1}\wedge \phi _{2}^{\prime \prime }$, i.e.~$\phi
_{1}\wedge \phi _{2}^{\prime \prime }$ is a gradually mixed formula.\qed
\end{proof}

\subsection{The recognition algorithm\label{recognition-algorithm-subsec}}

In this section we use Corollary~\ref{phi1-phi2''-gm-formula-cor} 
to design an algorithm that decides satisfiability on $\phi _{1}\wedge \phi
_{2}$ in time linear to its size (cf.~Theorem~\ref%
{formula-equivalent-gm-formula-thm}). This will enable us to combine the
results of Sections~\ref{linear-interval-sec} and~\ref%
{linear-interval-satisfiability-sec} to recognize efficiently
whether a given graph is a PI graph, or equivalently, due to Theorem~\ref%
{PI-char-thm}, whether a given partial order $P$ is the intersection of a
linear order $P_{1}$ and an interval order $P_{2}$.

\begin{theorem}
\label{formula-equivalent-gm-formula-thm}$\phi _{1}\wedge \phi _{2}$ is
satisfiable if and only if $\phi _{1}\wedge \phi _{2}^{\prime \prime }$ is
satisfiable. Given a satisfying truth assignment of $\phi _{1}\wedge \phi
_{2}^{\prime \prime }$ we can compute a satisfying truth assignment of $\phi
_{1}\wedge \phi _{2}$ in linear time.
\end{theorem}

\begin{proof}
If $\phi _{1}\wedge \phi _{2}$ is satisfiable then $\phi _{1}\wedge \phi
_{2}^{\prime \prime }$ is also satisfiable as a sub-formula of $\phi
_{1}\wedge \phi _{2}$. Conversely, suppose that $\phi _{1}\wedge \phi
_{2}^{\prime \prime }$ is satisfiable and let $\tau $ be a satisfying
assignment. Consider an arbitrary clause $\gamma =(\ell _{e_{1}}\vee \ell
_{e_{2}})$ of the sub-formula $\phi _{2}^{\prime }$ of $\phi _{2}$, cf.
Definition~\ref{phi2-partition-def}. If $\gamma $ is a tautology then $%
\gamma $ is satisfied by any truth assignment of $\phi $, and thus also by $%
\tau $. Assume now that $\gamma $ is not a tautology. Then at least one of
its literals $\{\ell _{e_{1}},\ell _{e_{2}}\}$ corresponds to an uncommitted
edge by Definition~\ref{phi2-partition-def}. Recall now by the construction
of $\phi _{1}$ (cf.~Section~\ref{linear-interval-satisfiability-sec}) that
in every clause of $\phi _{1}$, no literal is the negation of another
literal. Thus, for every clause of $\phi _{1}$, no pair among the three
edges in the corresponding $AC_{6}$ is in conflict. Therefore Lemma~\ref%
{AC6-if-not-in-conflict-all-not-isolated-lem} implies that all three edges
of such an $AC_{6}$ are committed. Thus, for every literal $\ell _{e}$ of $%
\phi _{2}^{\prime }$, which corresponds to an uncommitted edge $e$, neither $%
\ell _{e}$ nor $\overline{\ell _{e}}$ appears in $\phi _{1}$. Furthermore
recall that $\phi _{2}^{\prime \prime }$ does not include any literal $\ell
_{e}$ of any uncommitted edge $e$ of $H$ by~Definition~\ref%
{phi2-partition-def}.

Summarizing, for every literal $\ell _{e}$ of $\phi _{2}^{\prime }$, which
corresponds to an uncommitted edge $e$, neither $\ell _{e}$ nor $\overline{%
\ell _{e}}$ appears in $\phi _{1}\wedge \phi _{2}^{\prime \prime }$. That
is, the truth assignment $\tau $ of $\phi _{1}\wedge \phi _{2}$ does not
assign any value to the literal $\ell _{e}$. Furthermore, since $e$ is
uncommitted, no edge of $H$ is assigned the literal $\overline{\ell _{e}}$.
Therefore we can extend (in linear time) the truth assignment $\tau $ to a
truth assignment $\tau ^{\prime }$ that satisfies both $\phi _{1}\wedge \phi
_{2}^{\prime \prime }$ and $\phi _{2}^{\prime }$, by setting $\ell _{e}=1$
for all uncommitted edges $e$ of $H$. That is, $\tau ^{\prime }$ satisfies
the formula $\phi _{1}\wedge \phi _{2}$. Therefore $\phi _{1}\wedge \phi
_{2} $ is satisfiable if and only if $\phi _{1}\wedge \phi _{2}^{\prime
\prime }$ is satisfiable. This completes the proof of the theorem.\qed
\end{proof}

Now we are ready to present our recognition algorithm for PI graphs
(Algorithm~\ref{PI-graph-recognition-alg}). Its correctness and timing
analysis is established in Theorem~\ref{PI-algorithm-analysis-thm}.

\begin{algorithm}[t!]
\caption{Recognition of PI graphs} \label{PI-graph-recognition-alg}
\begin{algorithmic}[1]
\REQUIRE{A graph $G=(V,E)$}
\ENSURE{A PI representation $R$ of $G$, or the announcement that $G$ is not a PI graph}

\medskip

\IF{$G$ is a trapezoid graph} \label{PI-alg-0-1}
     \STATE{Compute a partial order $P$ of the complement $\overline{G}$} \label{PI-alg-0-2}
\ENDIF

\STATE{\textbf{else return} ``$G$ is not a PI graph''} \label{PI-alg-0-3}

\medskip

\STATE{Compute the domination bipartite graph $C(P)$ from $P$} \label{PI-alg-1}

\STATE{$\widetilde{G} \leftarrow \widehat{C}(P)$} \label{PI-alg-2}

\STATE{Compute the associated split graph $H$ of $\widetilde{G}$} \label{PI-alg-3}

\STATE{Compute the conflict graph $H^{\ast}$ of $H$} \label{PI-alg-4}

\vspace{0.1cm}

\IF{$H^{\ast}$ is bipartite} \label{PI-alg-5}
     \STATE{Compute a $2$-coloring $\chi_0$ of the vertices of $H^{\ast}$} \label{PI-alg-6}
     \STATE{Compute the formulas $\phi_1$ and $\phi_2$} \label{PI-alg-7}
     
     \vspace{0.1cm}
     
     \IF{$\phi_{1}\wedge \phi_{2}$ is satisfiable} \label{PI-alg-8}
          \STATE{Compute a satisfying truth assignment $\tau$ of $\phi_{1}\wedge \phi_{2}$ by Theorem~\ref{formula-equivalent-gm-formula-thm}} \label{PI-alg-9}
          \STATE{Compute from $\tau$ a linear-order cover of $\widetilde{G}$ by Algorithm~\ref{linear-interval-cover-from-assignment-alg}} \label{PI-alg-10}
          \STATE{Compute a PI representation $R$ of $G$ by Algorithm~\ref{alg-PI-repr-constr}} \label{PI-alg-11}
     \ELSE \label{PI-alg-12}
          \RETURN{``$G$ is not a PI graph''} \label{PI-alg-13}
     \ENDIF
\ELSE \label{PI-alg-14}
     \RETURN{``$G$ is not a PI graph''} \label{PI-alg-15}
\ENDIF

\vspace{0.1cm}
\RETURN{$R$} \label{PI-alg-16}
\end{algorithmic}
\end{algorithm}

\begin{theorem}
\label{PI-algorithm-analysis-thm}Let $G=(V,E)$ be a graph and $\overline{G}%
=(V,\overline{E})$ be its complement, where $|V|=n$ and $|\overline{E}|=m$.
Then Algorithm~\ref{PI-graph-recognition-alg} constructs in $O(n^{2}m)$ time
a PI representation of~$G$, or it announces that~$G$ is not a PI graph.
\end{theorem}

\begin{proof}
If the given graph $G$ is a trapezoid graph, then Algorithm~\ref%
{PI-graph-recognition-alg} computes in line~\ref{PI-alg-0-2} a partial order~%
$P$ of its complement $\overline{G}$. Otherwise, if $G$ is not a trapezoid
graph, then clearly it is also not a~PI graph, and thus the algorithm
correctly announces in line~\ref{PI-alg-0-3} that $G$ is not a~PI graph.

Let $C(P)$ be the domination bipartite graph of the partial order $P$ (cf.
Definition~\ref{C(P)-def}), and let $\widetilde{G}=\widehat{C}(P)$ be the
bipartite complement of $C(P)$, which are computed in lines~\ref{PI-alg-1}
and~\ref{PI-alg-2} of Algorithm~\ref{PI-graph-recognition-alg},
respectively. Furthermore let $H$ be the associated split graph of $%
\widetilde{G}$ (cf.~Definition~\ref{bipartite-split-def}) and $H^{\ast }$ be
the conflict graph of $H$ (cf.~Definition~\ref{conflict-graph-def}), which
are computed in lines~\ref{PI-alg-3} and~\ref{PI-alg-4} of Algorithm~\ref%
{PI-graph-recognition-alg}, respectively. If $H^{\ast }$ is not bipartite,
i.e.~if $\chi (H^{\ast })>2$, then $\widetilde{G}$ is not linear-interval
coverable by Lemma~\ref{conflict-bilocor-necessity-lem}, and thus $G$ is not
a PI graph by Corollary~\ref{linear-interval-col}. Therefore Algorithm~\ref%
{PI-graph-recognition-alg} correctly announces in line~\ref{PI-alg-15} that $%
G$ is not a PI graph if $H^{\ast }$ is not bipartite.

Suppose now that $H^{\ast }$ is bipartite, i.e.~$\chi (H^{\ast })\leq 2$.
Let $\chi _{0}$ be a $2$-coloring of the vertices of $H^{\ast }$, which is
computed in line~\ref{PI-alg-6} of Algorithm~\ref{PI-graph-recognition-alg}.
Furthermore let $\phi _{1}$ and $\phi _{2}$ be the Boolean formulas that can
be computed by Algorithms~\ref{alg-phi_1-constr} and~\ref{alg-phi_2-constr},
respectively (cf.~line~\ref{PI-alg-7} of Algorithm~\ref%
{PI-graph-recognition-alg}). If the formula $\phi _{1}\wedge \phi _{2}$ is
not satisfiable, then $\widetilde{G}$ is not linear-interval coverable by
Theorem~\ref{linear-interval-cover-by-satisfiability-thm}, and thus $G$ is
not a PI graph by Corollary~\ref{linear-interval-col}. Therefore Algorithm~%
\ref{PI-graph-recognition-alg} correctly announces in line~\ref{PI-alg-13}
that $G$ is not a PI graph if $\phi _{1}\wedge \phi _{2}$ is not satisfiable.

Suppose now that $\phi _{1}\wedge \phi _{2}$ is satisfiable, and let $\tau $
be a satisfying truth assignment of $\phi _{1}\wedge \phi _{2}$, as it is
computed in line~\ref{PI-alg-9} of Algorithm~\ref{PI-graph-recognition-alg}.
Then $\widetilde{G}$ is linear-interval coverable by Theorem~\ref%
{linear-interval-cover-by-satisfiability-thm}, and thus $G$ is a PI graph by
Corollary~\ref{linear-interval-col}. Furthermore, given $\tau $, we can
compute a linear-interval cover of $\widetilde{G}$ using Algorithm~\ref%
{linear-interval-cover-from-assignment-alg} (cf.~line~\ref{PI-alg-10} of
Algorithm~\ref{PI-graph-recognition-alg}). Finally, given this
linear-interval cover of $\widetilde{G}$, we can compute a PI representation 
$R$ of $G$ using Algorithm~\ref{alg-PI-repr-constr} (cf.~line~\ref{PI-alg-11}
of Algorithm~\ref{PI-graph-recognition-alg}). Thus, if $\phi _{1}\wedge \phi
_{2}$ is satisfiable, Algorithm~\ref{PI-graph-recognition-alg} correctly
returns $R$ in line~\ref{PI-alg-16}.

\medskip

\textbf{Time complexity}.~First note that the complement $\overline{G}$ of $%
G $ can be computed in $O(n^{2})$ time, since both $G$ and $\overline{G}$
have $n$ vertices. Furthermore, using the algorithm of~\cite{MaSpinrad94} we
can decide in $O(n^{2})$ time whether $G$ is a trapezoid graph, and within
the same time bound we can compute a trapezoid representation of $G$, if it
exists. Suppose in the following that $G$ is a trapezoid graph. Then we can
then compute in $O(n^{2})$ time a partial order $P$ of the complement $%
\overline{G}$ of $G$ as follows: $u<_{P}v$ if and only if the trapezoid for
vertex $u$ lies entirely to the left of the trapezoid for vertex $v$ in this
trapezoid representation of $G$. Therefore, lines~\ref{PI-alg-0-1}-\ref%
{PI-alg-0-3} of Algorithm~\ref{PI-graph-recognition-alg} can be executed in $%
O(n^{2})$ time in total. Note that we choose to compute the partial order $P$
using the trapezoid graph recognition algorithm of~\cite{MaSpinrad94}, in
order to achieve the $O(n^{2})$ time bound. Alternatively we could solve the
transitive orientation problem on $\overline{G}$ using the standard forcing
algorithm with $O(nm)$ running time (note that $m$ is the number of edges of 
$\overline{G}$).

Denote $\widetilde{G}=(U,V,\widetilde{E})$, where $U=\{u_{1},u_{2},\ldots
,u_{n}\}$ and $V=\{v_{1},v_{2},\ldots ,v_{n}\}$. Furthermore denote $%
E_{0}=\{u_{i}v_{i}\ |\ 1\leq i\leq n\}$. Then $H=(U,V,E_{H})$, where $E_{H}=%
\widetilde{E}\cup (V\times V)$ by Definition~\ref{bipartite-split-def}.
Since $C(P)$ and $H$ have $2n$ vertices each, each of the lines~\ref%
{PI-alg-1}-\ref{PI-alg-3} of Algorithm~\ref{PI-graph-recognition-alg} can be
computed by a straightforward implementation in $O(n^{2})$ time. Note that
the partial order $P$ has ${m}$ pairs of comparable elements, since the
complement $\overline{G}$ of $G$ has $m$ edges. Therefore the domination
bipartite graph $C(P)$ of $P$ has ${m}$ edges (cf.~Definition~\ref{C(P)-def}%
), and thus its bipartite complement $\widetilde{G}=\widehat{C}(P)$ has $|%
\widetilde{E}|=n^{2}-m$ edges.

Consider a pair $\{e,e^{\prime }\}$ of edges of $H$ that are in conflict,
i.e.~$e||e^{\prime }$ in $H$. Then $e,e^{\prime }\notin V\times V$ by
Observation~\ref{split-graph-clique-no-conflict-obs}, since $H$ is a split
graph and $V$ induces a clique in $H$. Therefore both $e$ and $e^{\prime }$
are edges of $\widetilde{G}$, i.e.~$e,e^{\prime }\in \widetilde{E}$, and
thus $e=u_{i}v_{j}$ and $e^{\prime }=u_{p}v_{q}$ for some indices $%
i,j,p,q\in \{1,2,\ldots ,n\}$. Furthermore, since $e$ and $e^{\prime }$ are
in conflict, it follows that $u_{i}v_{q},u_{p}v_{j}\notin \widetilde{E}$.
That is, every pair of conflicting edges in $H$ corresponds to exactly one
pair $\{u_{i}v_{q},u_{p}v_{j}\}$ of non-edges of $\widetilde{G}=\widehat{C}%
(P)$. Equivalently, every edge in the conflict graph $H^{\ast }$ of $H$
corresponds to exactly one pair of edges of $C(P)$. Since $C(P)$ has $m$
edges, it follows that the conflict graph $H^{\ast }$ has at most $O(m^{2})$
edges. Furthermore note that the conflict graph $H^{\ast }$ has ${\binom{n}{2%
}+}|\widetilde{E}|=O(n^{2})$ vertices, since $H$ has ${\binom{n}{2}+}|%
\widetilde{E}|$ edges. Therefore the conflict graph $H^{\ast }$ can be
computed in $O(n^{2}+m^{2})$ time (cf.~line~\ref{PI-alg-4} of Algorithm~\ref%
{PI-graph-recognition-alg}).

Note now that in time linear to the size of $H^{\ast }$, we can check
whether $H^{\ast }$ is bipartite, and we can compute a $2$-coloring $\chi
_{0}$ of the vertices of $H^{\ast }$, if one exists. Therefore lines~\ref%
{PI-alg-5}-\ref{PI-alg-6} of Algorithm~\ref{PI-graph-recognition-alg} can be
executed in $O(n^{2}+m^{2})$ time. Furthermore, in time linear to the size
of $H^{\ast }$, i.e.~in $O(n^{2}+m^{2})$ time, we can compute the connected
components $C_{1},C_{2},\ldots ,C_{k}$ of $H^{\ast }$. Then, having already
computed the $2$-coloring $\chi _{0}$ and the connected components $%
C_{1},C_{2},\ldots ,C_{k}$ of $H^{\ast }$, we can assign to every edge $e$
of $H$ the literal $\ell _{e}\in \{x_{i},\overline{x_{i}}\ |\ 1\leq i\leq
k\} $ (cf.~Section~\ref{linear-interval-satisfiability-sec}). This can be
done in $O(n^{2})$ time, since $H$ has ${\binom{n}{2}+}|\widetilde{E}%
|=O(n^{2})$ edges.

Now we bound the size of the formulas $\phi _{1}$ and $\phi _{2}$ that are
computed by Algorithms~\ref{alg-phi_1-constr} and~\ref{alg-phi_2-constr},
respectively. Regarding the size of $\phi _{2}$, note that, by the
construction of $\phi _{2}$, if $(\ell _{e}\vee \ell _{e^{\prime }})$ is a
clause of $\phi _{2}$, then $e=u_{i}v_{t}$, $e^{\prime }=u_{t}v_{j}$, and $%
u_{i}v_{j}\notin E_{H}\setminus E_{0}$, for some indices $i,j,t\in
\{1,2,\ldots ,n\}$. That is, for every index $t\in \{1,2,\ldots ,n\}$ and
for every pair $(i,j)$ of indices in the set $\{(i,j)\ |\ i=j$ or $%
u_{i}v_{j}\notin E_{H}\}$, the formula $\phi _{2}$ has at most one clause.
Note that every pair $(i,j)$ of the set $\{(i,j)\ |\ u_{i}v_{j}\notin
E_{H}\} $ corresponds to exactly one edge $u_{i}v_{j}$ of the bipartite
graph $C(P)$. Thus, since $C(P)$ has $m$ edges, it follows that $|\{(i,j)\
|\ i=j$ or $u_{i}v_{j}\notin E_{H}\}|\leq n+m$. Therefore $\phi _{2}$ has at
most $n(n+m) $ clauses, and thus $\phi _{2}$ can be computed in $O(n(n+m))$
time by Algorithm~\ref{alg-phi_2-constr}.

Regarding the size of $\phi _{1}$, recall first that every connected
component $C_{i}$ of the conflict graph $H^{\ast }$ has been assigned
exactly one Boolean variable $x_{i}$, where $i\in \{1,2,\ldots ,k\}$.
Furthermore recall that every edge $e$ of $H$ has been assigned a literal $%
\ell _{e}\in \{x_{i},\overline{x_{i}}\ |\ 1\leq i\leq k\}$. Therefore, since
every clause of $\phi _{1}$ appears only once in $\phi _{1}$ (cf.~lines~\ref%
{algorithm-phi_1-4}-\ref{algorithm-phi_1-5} of Algorithm~\ref%
{alg-phi_1-constr}), it follows by the construction of $\phi _{1}$ and by
Lemma~\ref{disjoint-clauses-lem} that $\phi _{1}$ has at most $2\frac{k}{3}$
clauses. Furthermore note that $k=O(n^{2})$, since $H^{\ast }$ has $O(n^{2})$
vertices. Thus $\phi _{1}$ has at most $O(n^{2})$ clauses.

\begin{myclaim}
\label{phi-1-implementation-claim}The following two statements are
equivalent:

\begin{enumerate}
\item[(a)] the formula $\phi _{1}$ contains the clauses $\alpha =(\ell
_{e}\vee \ell _{e^{\prime }}\vee \ell _{e^{\prime \prime }})$ and $\alpha
^{\prime }=(\overline{\ell _{e}}\vee \overline{\ell _{e^{\prime }}}\vee 
\overline{\ell _{e^{\prime \prime }}})$,

\item[(b)] there exist four distinct vertices $a,b,c,d$ in $H$, such that:%

\begin{itemize}
\item $ab\notin E_{H}$ and $bc,cd,da\in E_{H}$,

\item either $a,c\in U$ and $b,d\in V$, or $a,c\in V$ and $b,d\in U$,%

\item the edges $bc,cd,da$ are committed in $H$,

\item $\ell _{bc}=\ell _{e}$, $\ell _{cd}=\overline{\ell _{e^{\prime }}}$, $%
\ell _{da}=\ell _{e^{\prime \prime }}$, and

\item $\ell _{e}\neq \overline{\ell _{e^{\prime }}}$, $\ell _{e^{\prime
}}\neq \overline{\ell _{e^{\prime \prime }}}$, $\ell _{e}\neq \overline{\ell
_{e^{\prime \prime }}}$.
\end{itemize}
\end{enumerate}
\end{myclaim}

\begin{proofofclaim}[Proof of Claim \protect\ref{phi-1-implementation-claim}]
((a) $\Rightarrow $ (b)) Consider first a pair of clauses $\alpha =(\ell
_{e}\vee \ell _{e^{\prime }}\vee \ell _{e^{\prime \prime }})$ and $\alpha
^{\prime }=(\overline{\ell _{e}}\vee \overline{\ell _{e^{\prime }}}\vee 
\overline{\ell _{e^{\prime \prime }}})$ in $\phi _{1}$. These clauses
correspond to an $AC_{6}$ on the edges $\{e,e^{\prime },e^{\prime \prime }\}$
of $H$ by the construction of $\phi _{1}$. Furthermore, since $H$ is a split
graph, Lemma~\ref{split-no-AP5-double-AP6-lem} implies that this $AC_{6}$ of 
$H$ is an $AP_{6}$, i.e.~an alternating path of length $6$ (cf.~Figure~\ref%
{AC6-2-fig}). Let $w_{1},w_{2},w_{3},w_{4},w_{5},w_{6}$ be the vertices of
this $AP_{6}$, such that $e=w_{2}w_{3}$, $e^{\prime }=w_{4}w_{5}$, and $%
e^{\prime \prime }=w_{6}w_{1}$ (note that there always exists an enumeration
of the vertices of the $AP_{6}$ such that the edges $e,e^{\prime },e^{\prime
\prime }$ are met in this order on the $AP_{6}$). Then, since $V$ induces a
clique in $H$ and $U$ induces an independent set in $H$, Lemma~\ref%
{split-AC6-alternating-lem} implies that either $w_{1},w_{3},w_{5}\in U$ and 
$w_{2},w_{4},w_{6}\in V$, or $w_{1},w_{3},w_{5}\in V$ and $%
w_{2},w_{4},w_{6}\in U$. Since $\ell _{e}\neq \overline{\ell _{e^{\prime }}}$%
, $\ell _{e^{\prime }}\neq \overline{\ell _{e^{\prime \prime }}}$, and $\ell
_{e}\neq \overline{\ell _{e^{\prime \prime }}}$ (cf.~line~\ref%
{algorithm-phi_1-3} of Algorithm~\ref{alg-phi_1-constr}), it follows that no
pair among the edges $\{e,e^{\prime },e^{\prime \prime }\}$ is in conflict
in $H$. Therefore Lemma~\ref{AC6-if-not-in-conflict-all-not-isolated-lem}
implies that the edges $w_{3}w_{6},w_{4}w_{1},w_{5}w_{2}$ exist in $H$ and $%
e^{\prime }||w_{3}w_{6}$, $e||w_{4}w_{1}$, and $e^{\prime \prime
}||w_{5}w_{2}$. Thus all six edges $\{e,e^{\prime },e^{\prime \prime
},w_{3}w_{6},w_{4}w_{1},w_{5}w_{2}\}$ are committed. Furthermore $\ell
_{w_{4}w_{1}}=\overline{\ell _{e}}$, $\ell _{w_{3}w_{6}}=\overline{\ell
_{e^{\prime }}}$, and $\ell _{w_{5}w_{2}}=\overline{\ell _{e^{\prime \prime
}}}$. Thus the vertices $a=w_{1}$, $b=w_{2}$, $c=w_{3}$, and $d=w_{6}$ of $H$
satisfy the conditions of the part (b) of the claim.

((b) $\Rightarrow $ (a)) Conversely, consider four vertices $a,b,c,d$ in $H$%
, as specified in the part (b) of the claim. Then, since the edge $cd$ is
committed, there exists an edge $pq\in E_{H}$ such that $pc,qd\notin E_{H}$,
and thus $cd||pq$. Then $\ell _{pq}=\overline{\ell _{cd}}$. Therefore, since 
$\ell _{cd}=\overline{\ell _{e^{\prime }}}$, it follows that $\ell
_{pq}=\ell _{e^{\prime }}$. Thus there exists an $AC_{6}$ in $H$ on the
vertices $a,b,c,p,q,d$, where $\ell _{bc}=\ell _{e}$, $\ell _{pq}=\ell
_{e^{\prime }}$, and $\ell _{da}=\ell _{e^{\prime \prime }}$. Furthermore,
since $\ell _{e}\neq \overline{\ell _{e^{\prime }}}$, $\ell _{e^{\prime
}}\neq \overline{\ell _{e^{\prime \prime }}}$, and $\ell _{e}\neq \overline{%
\ell _{e^{\prime \prime }}}$ by assumption, it follows by the construction
of $\phi _{1}$ (cf.~Algorithm~\ref{alg-phi_1-constr}) that $\phi _{1}$
contains the clauses $\alpha =(\ell _{e}\vee \ell _{e^{\prime }}\vee \ell
_{e^{\prime \prime }})$ and $\alpha ^{\prime }=(\overline{\ell _{e}}\vee 
\overline{\ell _{e^{\prime }}}\vee \overline{\ell _{e^{\prime \prime }}})$.
\end{proofofclaim}

\medskip

Now, due to Claim~\ref{phi-1-implementation-claim}, we can implement
Algorithm~\ref{alg-phi_1-constr} for the computation of $\phi _{1}$ in time~$%
{O(n^{2}m+m^{2})}$ as follows. Recall first that $C(P)$ has $m$ edges. We
iterate for every edge~$u_{i}v_{j}$ of~$C(P)$, i.e.~for every \emph{non-edge}%
~${u_{i}v_{j}\notin E_{H}}$ of $H$. For every such $u_{i}v_{j}$, we mark all
vertices in the sets~$A$ and~$B$, where ${A=\{v\in V\ |\ u_{i}v\in E_{H}%
\text{ and }u_{i}v\text{ is committed in }H\}}$ and ${B=\{u\in U\ |\
uv_{j}\in E_{H}\text{ and }uv_{j}\text{ is committed in }H\}}$. Then we scan
through the adjacency lists of all vertices in $A$ to discover a pair of
vertices $v\in A$ and $u\in B$ such that $uv$ is a committed edge of $H$,
and $\ell _{v_{j}u}\neq \ell _{uv}$, $\ell _{uv}\neq \ell _{vu_{i}}$, and $%
\ell _{v_{j}u}\neq \overline{\ell _{vu_{i}}}$. Since $H$ has $O(n^{2})$
edges, this scan through the adjacency lists of the vertices of $A$ can be
done in $O(n^{2})$ time. If we discover such an edge $uv$, then we add to $%
\phi _{1}$ the clauses $\alpha =(\ell _{v_{j}u}\vee \overline{\ell _{uv}}%
\vee \ell _{vu_{i}})$ and $\alpha ^{\prime }=(\overline{\ell _{v_{j}u}}\vee
\ell _{uv}\vee \overline{\ell _{vu_{i}}})$. Due to Claim~\ref%
{phi-1-implementation-claim}, Algorithm~\ref{alg-phi_1-constr} would add the
same two clauses to $\phi _{1}$.

Due to Lemma~\ref{disjoint-clauses-lem}, no other clause of $\phi _{1}$ has
one of the literals $\{\ell _{v_{j}u},\overline{\ell _{v_{j}u}},\ell _{uv},%
\overline{\ell _{uv}},\ell _{vu_{i}},\overline{\ell _{vu_{i}}}\}$. After we
add the two clauses $\alpha $ and $\alpha ^{\prime }$ to $\phi _{1}$, we
visit all edges $e$ of $H$ which correspond to the same connected component
in $H^{\ast }$ with one of the edges $\{v_{j}u,uv,vu_{j}\}$. Note that
exactly these edges $e$ of $H$ have a literal $\ell _{e}\in \{\ell _{v_{j}u},%
\overline{\ell _{v_{j}u}},\ell _{uv},\overline{\ell _{uv}},\ell _{vu_{i}},%
\overline{\ell _{vu_{i}}}\}$. We then mark all these edges $e$ such that we
avoid visiting them again in any subsequent iteration during the
construction of $\phi _{1}$. Thus we ensure that each clause appears at most
once $\phi _{1}$ (cf.~lines~\ref{algorithm-phi_1-4}-\ref{algorithm-phi_1-5}
of Algorithm~\ref{alg-phi_1-constr}). Note that we can perform all such
markings of edges $e$ (for all iterations during the construction of $\phi
_{1}$) in time linear to the size of $H^{\ast }$, i.e.~in $O(n^{2}+m^{2})$
time. Summarizing, we need in total $O(n^{2}m+m^{2})$ time to compute the
formula $\phi _{1}$. Thus, since the formula $\phi _{2}$ can be computed in $%
O(n(n+m))$ time, it follows that line~\ref{PI-alg-7} of Algorithm~\ref%
{PI-graph-recognition-alg} can be executed in $O(n^{2}m+m^{2})$ time.

Now, we can test whether the formula $\phi _{1}\wedge \phi _{2}$ is
satisfiable in time linear to its size by Theorem~\ref%
{formula-equivalent-gm-formula-thm}; moreover, within the same time bound we
can compute a satisfying truth assignment $\tau $ of $\phi _{1}\wedge \phi
_{2}$, if one exists. Thus, since $\phi _{1}$ has $O(n^{2})$ clauses and $%
\phi _{2}$ has $O(n(n+m))$ clauses, lines~\ref{PI-alg-8}-\ref{PI-alg-9} of
Algorithm~\ref{PI-graph-recognition-alg} can be executed in $O(n(n+m))$
time. Furthermore, line~\ref{PI-alg-10} of Algorithm~\ref%
{PI-graph-recognition-alg} can be executed in $O(n^{2})$ time by Theorem~\ref%
{linear-interval-cover-by-satisfiability-thm}, calling Algorithm~\ref%
{linear-interval-cover-from-assignment-alg} as a subroutine. Finally, line~%
\ref{PI-alg-11} of Algorithm~\ref{PI-graph-recognition-alg} can be executed
in $O(n^{2})$ time by Theorem~\ref{PI-repr-constr-thm}, calling Algorithm~%
\ref{alg-PI-repr-constr} as a subroutine. Summarizing, since $m=O(n^{2})$,
the total running time of Algorithm~\ref{PI-graph-recognition-alg} is $%
O(n^{2}m)$. This completes the proof of the theorem.\qed
\end{proof}

Due to characterization of PI graphs in~Theorem~\ref{PI-char-thm} using
partial orders, the next theorem follows now 
by~Theorem~\ref{PI-algorithm-analysis-thm}.

\begin{theorem}
\label{linear-interval-order-recognition-thm}Let $P=(U,R)$ be a partial
order, where $|U|=n$ and $|R|=m$. Then we can decide in $O(n^{2}m)$ time
whether $P$ is a linear-interval order, and in this case we can compute a
linear order $P_{1}$ and an interval order $P_{2}$ such that $P=P_{1}\cap
P_{2}$.
\end{theorem}

\section{Concluding remarks\label{conclusions-sec}}

In this article we provided the first polynomial algorithm for the
recognition of simple-triangle graphs, or equivalently for the recognition
of linear-interval orders, solving thus a longstanding open problem. For a
graph $G$ with $n$ vertices, where its complement $\overline{G}$ has $m$
edges, our $O(n^{2}m)$-time algorithm either computes a simple-triangle
representation of $G$, or it announces that such one does not exist. The
main tool for our recognition algorithm was a new hybrid tractable subclass
of $3$SAT, called the class of \emph{gradually mixed} formulas. In addition,
we introduced the notion of a \emph{linear-interval cover} of bipartite
graphs, which naturally extends the well-known notion of the chain-cover of
bipartite graphs. There are two main lines for further research. The first
one is to identify more ``islands of tractability'' for hybrid classes of
SAT (and more generally of CSP), while the ultimate goal is to find a
complete characterization of the hybrid classes of CSP that are tractable.
The second line for further research is to resolve the complexity of the
recognition for the related classes with simple-triangle graphs, such as the
classes of \emph{unit} and \emph{proper tolerance} graphs~\cite{GolTol04}
(these are subclasses of parallelogram graphs, and thus also subclasses of
trapezoid graphs), \emph{proper bitolerance} graphs~\cite{GolTol04,Bogart98}
(they coincide with \emph{unit bitolerance} graphs~\cite{Bogart98}), and 
\emph{multitolerance} graphs~\cite{Mertzios-multitolerance_Algorithmica} (they naturally
generalize trapezoid graphs~\cite{Parra98,Mertzios-multitolerance_Algorithmica}). On the
contrary, the recognition problems for the related classes of \emph{triangle}
graphs~\cite{Mertzios-PI-ast-tcs}, \emph{tolerance} and \emph{bounded
tolerance} (i.e.~\emph{parallelogram}) graphs~\cite{MSZ-SICOMP-11}, and 
\emph{max-tolerance} graphs~\cite{Kaufmann06} have been already proved to be
NP-complete.

\end{document}